\RequirePackage{fix-cm}
\documentclass[twocolumn]{svjour3}          % twocolumn
\smartqed  % flush right qed marks, e.g. at end of proof
\usepackage{graphicx}
\usepackage{pgf,pgfpages}
\usepackage{tikz}
\usetikzlibrary{positioning,shapes,arrows,fit,automata,calc,decorations.pathreplacing}

\usepackage{full-IFM}
\usepackage{bbm}
\usepackage{amsmath} 
\usepackage{amssymb}
\usepackage{stmaryrd}
\usepackage{epsfig}
\usepackage{cite,url,xspace} 
\usepackage[final]{listings}
\usepackage{url}
\usepackage{mathpartir}
\usepackage{xcolor}
\usepackage[all]{xy}
\lstdefinelanguage{ABS}{keywords={
    local, view, call, specifies,
    assert,this,dyndelta,new,data,type,def,case,of,cog, class,interface,extends,implements,if,then,else,await,get, Fut,return,skip,while,module, import, export, from, suspend, 
    delta,adds,modifies,removes,original,productline,features,
    core,corefeatures,optionalfeatures,after,when,product,hasAttribute,
    hasMethod,root,extension,group,allof,oneof,require,exclude,original,
    ifin,ifout,opt,null,critical,port,rebind,duration,deadline,now,
    global,  view, result, local, specifies, grammar, send, receive},%
  sensitive=true, comment=[l]{//},
  morecomment=[s]{/*}{*/},
  morestring=[b]"}

\lstdefinestyle{absstyle}{
language=ABS,columns=fullflexible,
 		   mathescape=true,%
 		   showstringspaces=false,%
keywordstyle=\tt, %\sl\sffamily,  %\bf\sffamily,
commentstyle=\sl\sffamily,%
basicstyle=\footnotesize\tttfamily,
inputencoding=latin1, % i would prefer utf8
extendedchars,xleftmargin=2em
}

\lstnewenvironment{abs}[1][]{%
 \lstset{style=absstyle,#1}%
   \csname lst@SetFirstLabel\endcsname}
  {\csname lst@SaveFirstLabel\endcsname}

\newcommand{\Abs}[1]{\lstinline[language=ABS,columns=fullflexible,mathescape=true,keywordstyle=\tt,basicstyle=\footnotesize\ttfamily]|#1|} %before \sl\sffamily

\newenvironment{absinline}[1][]{%
\lstset{language=ABS,
		 showstringspaces=false,
		 columns=fullflexible,mathescape=true,%
         keywordstyle=\bf\sffamily,commentstyle=\sl\sffamily,%
         basicstyle=\footnotesizesize\ttfamily,inputencoding=latin1, % i would prefer utf8
         extendedchars,xleftmargin=2em,#1}}
{}

\lstnewenvironment{absexample}{
~\\
\noindent 
\small{\textcolor{green!80!black}{\bfseries\sffamily{Example:}}}\vspace{-4pt}
\lstset{style=absstyle,
basicstyle=\ttfamily,
framerule=.5pt,
backgroundcolor=\color{green!3},
rulecolor=\color{green!50!black},
frame=tblr,
xleftmargin=4pt,
xrightmargin=4pt,
}}
{\vspace{\baselineskip}}

\lstnewenvironment{absexamplen}[1][]{
\lstset{style=absstyle,
basicstyle=\footnotesize\ttfamily,
framerule=.5pt,
belowskip=0pt,
backgroundcolor=\color{green!3},
rulecolor=\color{green!50!black},
frame=tblr,
captionpos=b,
xleftmargin=4pt,
xrightmargin=4pt,#1
}}
{\vspace{\baselineskip}}

\newcommand{\bigfract}[2]{\frac{^{\textstyle #1}}{_{\textstyle #2}}}
\newcommand{\rulenamex}[1]{\mbox{\scriptsize\sc(#1)}}

\def \mathrule #1#2#3{\begin{array}{l} 
                       {\rulenamex{#1}}
                       \\ \bigfract{#2}{#3}
                      \end{array}}

\newcommand{\subst}[2]{[\raisebox{.5ex}{$#1$}  /
                        \raisebox{-.5ex}{$#2$} ]}

\newcommand{\repl}[2]{(\raisebox{.5ex}{$#1$}  \curvearrowright
                        \raisebox{-.5ex}{$#2$} )}
\newcommand{\extr}[1]{\lceil#1\rceil}

\newcommand{\DFfABS}{{\tt DF4ABS}}
\renewcommand{\bar}[1]{\overline{#1}}

\newcommand{\D}{{\tt D}}

\newcommand{\n}{{\tt n}}

\newcommand{\x}{{x}}

\newcommand{\xbar}{\bar{\x}}
\newcommand{\class}{{\tt class}}

\newcommand{\return}{{\tt return}}

\newcommand{\this}{{\tt this}}
\newcommand{\destiny}{{\tt destiny}}

\newcommand{\get}{{\tt get}}
\newcommand{\await}{{\tt await}}

\newcommand{\seqpoint}{\!{\large \mbox{\tt .}}\!}
\newcommand{\X}{X}
\newcommand{\Y}{Y}
\newcommand{\Z}{Z}

\newcommand{\lred}[1]{\stackrel{#1}{\longrightarrow}}

\renewcommand{\aa}{{\tt w}} 
\newcommand{\async}{\mbox{\tt !}}
\newcommand{\parop}{{\parallel}}
\newcommand{\addition}{\binampersand}

\newcommand{\T}{T}
\newcommand{\Tbar}{\bar{\T}}
\newcommand{\fut}{\textit{fut}}

\newcommand{\fields}[1]{\fieldsN(#1)}
\newcommand{\fieldsN}{{\it fields}}
\newcommand{\parameters}[1]{\parametersN(#1)}
\newcommand{\parametersN}{{\it param}}
\newcommand{\types}[1]{\typesN(#1)}
\newcommand{\typesN}{{\it types}}
\newcommand{\fclass}[1]{\fclassN(#1)}
\newcommand{\fclassN}{{\it class}}

\newcommand{\dom}{\text{dom}}

\newcommand{\gr}[1]{{\it cog\_names}(#1)}

\newcommand{\cog}{{\it cog}}
\newcommand{\invoc}{{\it invoc}}

\newcommand{\objA}{\it{ c}}

\newcommand{\obj}{{\it c}}
\newcommand{\objb}{{\it c}'}

\newcommand{\wt}[1]{\widetilde{#1}}
\newcommand{\fRec}[2]{#1 \leadsto #2}

\newcommand{\interface}[3]{#1(#2) \rightarrow #3}

\newcommand{\sem}[1]{[\![ #1 ]\!]}
\newcommand{\slam}[1]{\raisebox{-.2ex}{\mbox{\large {\tt [}}} #1 \raisebox{-.2ex}{\mbox{\large {\tt ]}}}}

\newcommand{\contract}[1]{{\it unsync}(#1)}
\newcommand{\rtcontract}[1]{{\it rt\_unsync}(#1)}
\newcommand{\ct}{\mbox{\sc ct}}
\newcommand{\cct}{\mbox{\sc cct}}
\newcommand{\eqdef}{\stackrel{\it def}{=}}

\newcommand{\frr}{\mathbbm{r}}
\newcommand{\frs}{\mathbbm{s}}
\newcommand{\frx}{\mathbbm{x}}
\newcommand{\frz}{\mathbbm{z}}

\newcommand{\cntcp}{\mathbbm{cp}}
\newcommand{\concntc}{\mathbbm{k}}
\newcommand{\mcntc}{\mathbbm{C}}

\newcommand{\abstractsemantics}[1]{[\![ #1 ]\!]_{[n]}}

\newcommand{\act}{\mbox{\sc act}}

\newcommand{\names}[1]{\nt{names}(#1)}

\renewcommand{\rfloor}{\parallel}
\newcommand{\lafsa}{{\cal L}}

\newcommand{\fatcomma}{\mbox{\large {\bf ,}}\,}

\newcommand{\pairl}[2]{\langle #1 \fatcomma #2 \rangle}

\renewcommand{\emptyset}{\varnothing}

\newcommand{\lt}{\mathfrak{L}}

\newcommand{\true}{{\tt true}}
\newcommand{\false}{{\tt false}}
\newcommand{\Bool}{{\tt Bool}}
\newcommand{\coreABS}{\mbox{\tt core\hspace{2pt}ABS}}

\newcommand{\ABS}{\mbox{{\tt ABS}}\xspace}

\newcommand{\key}[1]{\mbox{\ttfamily\bfseries #1}}

\newcommand{\many}[1]{\overline{#1}}

\newcommand{\rulenamepre}[1]{{\footnotesize\scshape (#1)}}
\newcommand{\rulename}[1]{\ifmmode \mbox{\rulenamepre{#1}}\else\rulenamepre{#1}\fi}

\newcommand{\feval}[2]{[\![#1]\!]_{#2}}

\newcommand{\nt}[1]{\textit{#1}} 

\newcommand{\Empty}{\varepsilon}

\newcommand{\sep}{\mid}

\newcommand{\ntyperule}[3]{ 
  \begin{array}{c} 
    \textsc{\scriptsize ({#1})} \\ 
    #2 \\ 
    \hline
    #3 
  \end{array} }
 
\newcommand{\nredrule}[2]{ 
  \begin{array}{c} 
    \textsc{\scriptsize ({#1})} \\ 
    #2  
  \end{array}} 
  
\def\myttsize{\fontsize{8}{9}}

\newcommand{\ordermut}[1]{\mathfrak{o}_{#1}}

\newcommand{\CP}{\mathfrak{C}}
\newcommand{\CD}{\mathfrak{D}}

\newcommand{\redmut}{\rightarrow_{\tt mut}}

\newcommand{\fpar}{\parop}
\newcommand{\fsimpl}[1]{\llparenthesis#1\rrparenthesis}
\newcommand{\undefined}{\bot}

\newif \ifIFM
\IFMfalse

\begin{document}

\title{A Framework for Deadlock Detection in {\coreABS}\thanks{Partly funded by the
  EU project FP7-610582 ENVISAGE:
    Engineering Virtualized Services.}}


\author{Elena Giachino \and  Cosimo Laneve \and  Michael Lienhardt }


\institute{{Department of Computer Science and Engineering, University of Bologna -- INRIA Focus Team, Italy}}
% \institute{F. Author \at
%               first address \\
%               Tel.: +123-45-678910\\
%               Fax: +123-45-678910\\
%               \email{fauthor@example.com}           %  \\
% %             \emph{Present address:} of F. Author  %  if needed
%            \and
%            S. Author \at
%               second address
% }

\date{Received: date / Accepted: date}

\maketitle

 \begin{abstract}
We present a framework for statically detecting deadlocks in 
a concurrent object-oriented language with asynchronous method calls and cooperative scheduling of method activations.
Since this language features recursion and dynamic resource creation,
deadlock detection is extremely complex and state-of-the-art solutions either 
give imprecise answers or do not scale.

In order to augment precision and scalability we propose a modular framework
that allows several techniques to be combined. The basic component of the framework is
a front-end inference algorithm that extracts abstract behavioural descriptions of methods, 
called contracts, which retain resource dependency information.
This component is integrated with a number of possible different back-ends that analyse contracts and derive deadlock information. 
As a proof-of-concept, we discuss two such back-ends: 
(i) an evaluator that computes a fixpoint semantics and (ii) an evaluator using
abstract model checking.
%
%We expose the theory and the prototype implementation of our framework. The prototype
%is validated on an industrial case study based on the Fredhopper Access Server (FAS)
%developed by SDL Fredhoppper. %, which is
%verified to be deadlock-free.
\end{abstract}

 \section{Introduction}\label{sec:intro}
 Modern systems are designed to support a high degree of parallelism by letting
as many system components as possible operate concurrently. When such systems 
also
exhibit a high degree of resource and data sharing then
deadlocks represent an insidious and recurring threat. 
In particular, deadlocks arise as a consequence of exclusive resource access and circular wait for accessing resources.
A standard example is when two processes are exclusively holding a
different resource and are requesting access to the resource held by the other.
That is, the correct termination of each of the two process activities 
\emph{depends} on the termination of the other. The presence of a \emph{circular
dependency} makes termination impossible.

Deadlocks may be particularly hard to detect in systems with unbounded (mutual) 
recursion and dynamic resource creation. A paradigm case is 
an adaptive system that creates an unbounded number of processes such as server applications. 
In these systems, the interaction protocols are extremely complex and state-of-the-art solutions either give imprecise answers or do not 
scale -- see Section~\ref{sec:relatedworks} and, for instance,~\cite{Naik2009} and the references therein.

In order to augment precision and scalability we propose a modular framework that 
allows several techniques to be combined. 
We meet scalability requirement by designing a front-end inference system that 
automatically extracts abstract behavioural descriptions 
pertinent to deadlock analysis, called \emph{contracts}, from code. The inference
system is \emph{modular} because it (partially) supports separate inference of 
modules. 
To meet precision of contracts' analysis, as a proof-of-concept we
define and implement two different techniques:
(i) an evaluator that computes a fixpoint semantics and (ii) an evaluator using 
abstract model checking.

Our framework targets {\coreABS}~\cite{johnsen10fmco}, which is an abstract, executable, object-oriented modelling language
with a formal semantics, targeting distributed
systems. In {\coreABS}, method invocations are asynchronous:
the caller continues
after the invocation and the called code runs on a different 
task. Tasks are cooperatively scheduled, that is 
there is a notion of group of objects, called \emph{cog}, and
there is at most one active task 
at each time per cog. The active task explicitly returns the
control in order to let other tasks progress.
The synchronisation between
the caller and the called methods is performed
when the result is strictly necessary~\cite{Creol,ABCL1,Eiffel}. 
Technically, the decoupling of method invocation and the returned 
value is realised using \emph{future variables} (see~\cite{futures} and the references in there), which are pointers to values that may be not available yet.
Clearly, the access to values of future variables may require waiting for
the value to be returned. We discuss the syntax and the semantics of
{\coreABS}, in Section~\ref{sec:formal}.

Because of the presence of
explicit synchronisation operations,  the analysis of deadlocks in {\coreABS} 
is more fine-grained than in thread-based languages (such as {\tt Java}). 
However, as usual with (concurrent) programming languages,  
analyses are hard and time-consuming because most part of the code is irrelevant 
for the properties one intends to derive. 
%For instance, in case of deadlock and synchronisation behaviours, local data 
%and computations do not play any substantial role.
%
For this reason, in Section~\ref{sec.FJg-contracts}, we  design an inference system that
\emph{automatically extracts contracts} from {\coreABS} code. 
These contracts are similar to those
ranging from languages for session types~\cite{GH05} to process contracts~\cite{LP}
and to calculi of processes as Milner's CCS or pi-calculus~\cite{CCS,PIC}.
The inference system mostly collects method behaviours and uses constraints to 
enforce consistencies among behaviours. Then a 
standard semiunification technique is used for solving the set of generated constraints.

Since our inference system addresses a language with asynchronous method invocations,
it is possible that a method triggers behaviours that will last \emph{after} its
lifetime (and therefore will contribute to \emph{future} deadlocks). In order to 
support a more precise analysis,
we  split contracts of methods in \emph{synchronised} and
\emph{unsynchronised contracts}, with the intended meaning that the formers collect 
the invocations that are explicitly synchronised in the method body and the
latter ones collect the other invocations. 
%
%behaviours of the formers
%are not mixed with the behaviours of the latter ones.

The current release of the inference system does not cover the full range of features
of {\coreABS}.  In 
Section~\ref{sec.restrictions} we discuss the restrictions of {\coreABS} and
the techniques that may be used to remove these restrictions.

Our contracts feature recursion and resource creation; therefore their
underlying models are infinite
sta\-tes and their analysis cannot be exhaustive. We propose two techniques for
analysing contracts (and to show the modularity of our framework). 
The first one, which is discussed in 
Section~\ref{sec.contractanalysis}, is a fixpoint technique on models with a limited capacity of 
name creation. 
This entails fixpoint existence and finiteness of models. While we lose  
precision, our technique is sound (in some cases, this technique may signal false
positives). 
The second technique, which is detailed in Section~\ref{sec.mutanalysis}, is an abstract model checking that evaluates the contract
program up-to some point, which  is possible to determine by analysing 
the recursive patterns of the program. This technique is precise when the 
recursions are linear, while it is over-approximating in general. 

We have prototyped an implementation of our framework, called the {\DFfABS}
tool and, in Section~\ref{sec:sdatool}, we assess the
precision and performance  of the prototype. In particular, we have applied it to
an industrial case study that is based on
the Fredhopper Access Server (FAS) developed by SDL Fredhopper\footnote{\url{http://sdl.com/products/fredhopper/}}. 
It is worth to recall that, because of the
modularity of {\DFfABS}, the current analyses techniques may be
integrated and/or replaced by other ones. We discuss this idea in 
Section~\ref{sec:conclusion}.

\medskip

\paragraph{Origin of the material.} The basic ideas of this article have appeared in 
conference proceedings. In particular, the contract language and (a simplified form of) 
the inference system have been introduced 
in~\cite{GiachinoL11,GL2013a}, while the fixpoint analysis technique  has been explored in~\cite{GL2013a} and an introduction to the abstract model checking technique can be found in~\cite{GiachinoL13}, while the details are in~\cite{GL2014}. This article is a thoroughly revised and enhanced version of~\cite{GL2013a} that presents the whole framework in a uniform setting and includes the full proofs of all the results. A more detailed comparison with other related work is postponed to 
Section~\ref{sec:relatedworks}.

\section{The language {\coreABS}}\label{sec:formal}
%!TEX root = SoSyM.tex

The syntax and the semantics (of the concurrent object level) of
{\coreABS} are defined in the following two subsections; the
third subsection is devoted to the discussion of examples, and the 
last one to the definition of deadlock. 
In this contribution we overlook the functional level of {\coreABS} that defines
data types and functions because their analysis can be performed with techniques that
may (easily) complement those discussed in this paper (such as data-flow analysis).
Details of {\coreABS}, its semantics and its standard type system can be also found 
in~\cite{johnsen10fmco}.

\subsection{\em Syntax}
Figure~\ref{fig:core:lan} displays  {\coreABS} syntax,
\begin{figure*}[t]
{\footnotesize
%\begin{minipage}[t]{.65\textwidth}
\mysyntax{\mathmode}{
\simpleentry P  [] %\vect{D}\ \vect{F}\ \vect{I}\ 
	\vect{I}\ \vect{C}\ \wrap{\vect{T\ x\ssemi}\ s} [program]
%\\                                           
%\simpleentry Dl [] D \bnfor F\bnfor I\bnfor C  [declaration]
\\                                                              
\simpleentry T  [] \name D \bnfor  \name{Fut} \arglang{T} 
\bnfor \name I                   
                           [type]
\\
%\simpleentry D  [] \ddata{\name D \arglang{\vect V}}{\vect{\name{Co}~[~(\vect T)~]}}
%%[|\vect{\name{Co}[(\vect T)}]}               
%[data type]\\
%\simpleentry F  [] \dfun{T}{\name{f}~[\arglang{\vect T}]}{\vect{T\ x}}{e}                                           [function]\\
\simpleentry I  [] \dsinterface{I}{\vect{S \ssemi}}                                                                 [interface]\\
\simpleentry S  [] T\ \name m(\vect{T\ x })                                                                   [method signature]\\
\simpleentry C  [] \dsclass{C}(\vect{T\ x}){\ [\nimplements\ \vect{\name I}]\ 
\wrap{\vect{T\ x\ssemi}\ \vect{M}}}    [class]\\
%\simpleentry Fl [] T\ x                                                                                  [field declaration]\\
\simpleentry M  [] S\wrap{\vect{T\ x\ssemi}\ s}                                                               [method definition]\\
\simpleentry s []
      \eskip \bnfor x=z \bnfor \eif ess\bnfor \ereturn e \bnfor s\semi s \bnfor\await\ e \mbox{\tt ?}\qquad \qquad [statement]
%\oris  %\ewhile es\bnfor
% \quad~ [{}]
\\
\simpleentry z [] e\bnfor e \seqpoint \name m(\vect e)\bnfor e \async \name m(\vect e)\bnfor \efnew{}C(\vect e)  \bnfor \efnew{\ncog}C(\vect e) \bnfor \eget e  
\qquad \qquad [expression with side effects]
%\oris  \efnew{\ncog}C(\vect e)\quad~ [{}]
\\
\simpleentry e [] v\bnfor x\bnfor\ethis
\bnfor \text{\it arithmetic-and-bool-exp} %\bnfor\name{fun}(\vect e)
%\bnfor \ecase{e}{\vect{p\Rightarrow e}} 
 [expression]
 \\
\simpleentry v [] \enull\bnfor \text{\it primitive values}%\name{Co}{[(\vect v)]} 
[value]\\
%\simpleentry p [] \_\bnfor x\bnfor \enull\bnfor\name{Co}{[(\vect p)]} [pattern]\\
%\simpleentry g [] e \bnfor x \mbox{\tt ?}\iffalse \bnfor g \land g\fi [guard]
}
}
\caption{The language {\coreABS} }\label{fig:core:lan}
\end{figure*}
where an overlined element corresponds to any finite sequence of such element.
The elements of the sequence are separated by commas, except for $\vect{C}$, which has no separator.
For example $\vect{T}$ means a (possibly empty) sequence 
$T_1 \mbox{\tt ,} \cdots \mbox{\tt ,} T_n$. When we write $\vect{T \ x \ssemi}$ we mean
a sequence $T_1\ x_1 \semi \cdots \semi$ $T_n\ x_n \semi$ when the sequence is not
empty; we mean the empty sequence otherwise.

A program $P$ is a list of interface and class declarations (resp. $I$ and $C$)
followed by a \emph{main function} $ \wrap{\vect{T\ x\ssemi}\ s}$.
A type $T$ is the name of 
 either a primitive type  $\name D$ such as $\name{Int}$, $\name{Bool}$, 
$\name{String}$,  or
a \emph{future type} $\name{Fut}\arglang{T}$, or an interface name $\name I$.

A class declaration $\dsclass{C}(\vect{T\ x}){\ 
\wrap{\vect{T'\ x'\ssemi}\ \vect{M}}}$ has a name $\name{C}$ %, may implement several interfaces,
and declares its fields $\vect{T \ x}, \vect{T'\ x'}$ and its methods $M$.
The fields $\vect{T \ x}$ will be initialised when the object is created; the 
fields $\vect{T'\ x'}$ will be initialised by the main function of the class (or by
the other methods).

A statement $s$ may be either one of the standard operations of an
imperative language or one of the operations for scheduling.
This operation is  $\await \, x \mbox{\tt ?}$ (the other one is \Abs{get}, see below), 
which suspends
method's execution until the argument $x$,
is resolved. % When the guard is
%a future lookup $x \mbox{\tt ?}$,  
This means that $\await$ requires the value of $x$
to be resolved before resuming method's execution. 

An expression $z$ may have side effects (may change the state of the
system) and is either an object creation {\tt new C}$(\bar{e})$ in the same group of the
creator or an object creation {\tt new cog C}$(\bar{e})$ in a new group.
In {\coreABS}, (runtime) objects are partitioned in groups, called \emph{cogs}, which 
own a lock for regulating the executions of threads. Every threads acquires 
its own cog lock
in order to be evaluated and releases it upon termination or suspension. Clearly,
threads running on different cogs may be evaluated in parallel, while threads running 
on the same cog do compete for the lock and interleave their evaluation.
The two operations {\tt new C}$(\bar{e})$ and {\tt new cog C}$(\bar{e})$ allow one to add an object to 
a previously created cog or to create new singleton cogs, respectively.  
 
An expression $z$ may also be either a (synchronous) method call 
$e \seqpoint \name m(\vect e)$ or 
an \emph{asynchronous} method call $e \async \name m(\vect e)$.
Synchronous method invocations suspend the execution of the caller, 
without releasing the lock of the corresponding cog; asynchronous method invocations
do not suspend caller's execution.
Expressions $z$ also include the operation $\eget e$  that suspends method's
execution until the value of $e$ is computed. The type of $e$ is a future type that
is associated with a method invocation. The difference between $\await\ x \mbox{\tt ?}$ and 
$\eget e$ is that the former releases cog's lock when the value of $x$ is still 
unavailable; the latter does not release cog's lock (thus being the potential 
cause of a deadlock). 

A {\em pure} expression $e$ is either a value, or a variable $x$, or the reserved identifier \Abs{this}.
%, a function application $\name{fun}(\vect e)$, or a pattern matching $\ecase{e}{\vect{p\Rightarrow e}}$.
Values include the $\enull$ object, and 
primitive type values, such as \Abs{true} and \Abs{1}.

In the whole paper, we assume that sequences of field declarations $\vect{T \ x}$, method declarations $\vect{M}$, and 
parameter declarations $\vect{T \  x}$ do not contain duplicate names. 
It is  
also assumed that every class and interface name in a program has a unique definition.

\subsection{\em Semantics}

\begin{figure*}[t]
\centering \myttsize 
$\begin{array}{rcl@{\qquad\qquad}rcl}
\nt{cn} & ::= & \epsilon \sep \nt{fut}(f,\nt{val}) \sep  \nt{ob}(o,a,p,q) 
              \sep \nt{invoc}(o,f,\m,\many{v}) \sep  \nt{cog}(c,\nt{act})\sep \nt{cn}~\nt{cn}
&
\nt{act}&::=& o \sep
\varepsilon
\\
\nt{p} & ::= & \{l \mid s\} \sep \key{idle} 
	%\sep \key{error} 
	& \nt{val}&::=& \nt{v}\sep \bot
\\
q & ::= & \epsilon \sep \{l \mid s\} 
	%\sep \key{error} 
                \sep q~q 
&a & ::= & [ \cdots , x\mapsto v, \cdots ] %T\ x\  \nt{v} \sep a,a
\\
s & ::=& \key{cont}(f)\sep \ldots& \nt{v} & ::= & o \sep f \sep \ldots\\
\\

\end{array}$
\caption{Runtime syntax of {\coreABS}.}
\label{fig:rtsyn}
\end{figure*}
{\coreABS} semantics is defined as a transition relation
between \emph{configurations}, noted \nt{cn} and defined in Figure~\ref{fig:rtsyn}. 
Configurations are sets of elements -- therefore 
we identify configurations that 
are equal up-to associativity and commutativity --
and are denoted by the juxtaposition of the
elements \nt{cn}~\nt{cn}; the empty configuration is denoted by $\varepsilon$.
The transition relation uses three infinite sets of
names: \emph{object names}, ranged over by 
$o$, $o'$, $\cdots$, \emph{cog names}, ranged over by $c$, 
$c'$, $\cdots$, and \emph{future names}, ranged over by $f$, $f'$, $\cdots$.
Object names  are partitioned according to the class and the cog
they belongs. We assume there are infinitely many
object names per class and the function $\fresh(\name{C})$ returns a
new object name of class \name{C}. Given an object name $o$, the function 
$\text{class}(o)$ returns its class. The function $\fresh(\,)$ returns either a fresh
cog name or a fresh future name; the context will disambiguate between the twos.

\emph{Runtime values} are either values $v$ in Figure~\ref{fig:core:lan} or
 object and future names or an undefined value, which is denoted by $\undefined$.
  
\emph{Runtime statements} extend normal statements with $\key{cont}(f)$ that is used
to model explicit continuations in synchronous invocations. 
With an abuse of
notation, we range over runtime values with $v$, $v'$, $\cdots$ and over 
runtime statements with $s$, $s'$, $\cdots$.
We finally use 
$a$ and $l$, possibly indexed, to range over maps from fields to runtime values 
and local 
variables to runtime values, respectively. The map $l$ also binds the special
name $\text{destiny}$ to a future value.

The elements of configurations are 
\begin{itemize}
\item[--]
\emph{objects} $ob(o,a,p,q)$ where $o$ is an object name; $a$ returns the values
of object's fields, $p$ is either $\text{idle}$, representing inactivity, or
is the  \emph{active process} $\{ l \mid s \}$, where $l$ returns the values of local identifiers and $s$ is the 
statement to evaluate; $q$ is a set of processes to evaluate. 
\item[--]
\emph{future binders} $\fut(f,v)$ where $v$, called \emph{the reply value} may be also $\bot$ meaning that the value has still not computed.
\item[--]
\emph{cog binders} $\cog(c,o)$ where $o$ is the active object; it may be $\varepsilon$
meaning that the cog $c$ has no active object.
\item[--]
\emph{method invocations} $\invoc(o,f,\m,\vect{v})$.
\end{itemize}

\noindent
The following auxiliary functions are used in the semantic rules (we assume a fixed 
{\coreABS} 
program): 
%
% \Elena{in realt\`a queste auxiliary functions dovrebbero riferirsi alla contract class table, in cui i metodi e i campi sono etichettati coi loro record/ contratti. Il problema \`e che a questo punto record e contratti non sono ancora stati introdotti}
\begin{itemize}
\item[--]
$\dom(l)$ and $\dom(a)$ return the domain of $l$ and $a$, respectively.

\item[--] $l[x \mapsto v]$ is the function such that $(l[x \mapsto v])(x) = v$ and
$(l[x \mapsto v])(y) = l(y)$, when $y \neq x$. Similarly for $a[x \mapsto v]$.

\item[--]
$\feval{e}{(a+ l)}$ returns the value of $e$ by computing the arithmetic and boolean 
expressions and retrieving the value of the identifiers that is stored either in $a$ or in $l$.
Since $a$ and $l$ are assumed to have disjoint domains, we denote the union map with 
$a+l$. $\feval{\vect{e}}{(a+ l)}$ returns the tuple of values of $\vect{e}$.
When $e$ is a future name, the function $\feval{\cdot}{(a+ l)}$ is the identity. Namely
$\feval{f}{(a+ l)} = f$. 

\item[--]
$\name{C}.\m$ returns the term $(\vect{T \ x}) \{ \vect{T' \ z} ; s\}$ that contains
the arguments and the body of the method $\m$ in the class $\name{C}$. 

\item[--]
$\text{bind}(o,f,\m,\many{v},\name{C}) = \{ [\text{destiny} \mapsto f, \bar{x} \mapsto \bar{v}, \bar{z} \mapsto \undefined] \mid s
\subst{o}{\this} \}$, where 
 $\name{C}.\m = (\vect{T \ x}) \{ \vect{T' \ z} ; s\}$. 
 %Binding may not succeed, in this case we get the \key{error} process.
%\Elena{il quinto argomento di $bind$ qui \`e un cog, mentre nella semantica \`e una classe}

\item[--]
$\text{init}(\name{C},o)$ returns the process 
\[
\{ \emptyset [\text{destiny} \mapsto f_\undefined] 
	\mid s \subst{o}{\name{this}} \}
\]
 where
$\{ \vect{T \ x} ; s \}$ is the main function of the class $\name{C}$. The special 
name $\text{destiny}$ is %useless, therefore the function $l$ is empty).
initialised to a fresh (future) name $f_\undefined$.

\item[--]
$\text{atts}(\name{C},\vect{v},c)$ returns the map $[\cog \mapsto c, 
\vect{x} \mapsto \vect{v}, \bar{x'} \mapsto \undefined]$, where the class $\name{C}$ is defined as
\[
\dsclass{C}(\vect{T\ x}){\wrap{\vect{T'\ x'\semi}\ \vect{M}}}
\]
and where $\cog$ is a special field storing the cog name of the object.
\end{itemize}

The transition relation rules are collected in 
Figures~\ref{fig:sem1} and~\ref{fig:sem2}. They define transitions of 
objects $ob(o,a,p,q)$ according to 
the shape of the statement in $p$.
\begin{figure*}[t]
\centering \myttsize 
\renewcommand{\arraystretch}{0.9} 
$\begin{array}{l}
\qquad 
\nredrule{Skip}{
	ob(o,a,\{l \mid \key{skip};s\},q) 
	\\
	\to {ob(o,a,\{l \mid s\},q)}}
\hfill
\ntyperule{Assign-Local}{
	x \in \dom(l)  \quad  v=\feval{e}{(a+ l)}
	}{
	ob(o,a,\{l \mid x=e;s\},q) \\
	\to {ob(o,a,\{l[x\mapsto v] \mid s\},q)}}
\hfill
\ntyperule{Assign-Field}{ 
	x \in \dom(a)\setminus\dom(l) \quad v=\feval{e}{(a+l)}
	}{ob(o,a,\{l \mid x=e;s\},q)\\
	\to {ob(o,a [x\mapsto v],\{l \mid s\},q) }}
\qquad
\\
\\
\qquad
\ntyperule{Cond-True}{	
	 \true = \feval{e}{(a+l)}
	}{ ob(o,a,\{l \mid \key{if}\ e\ \key{then}\:\{s_1\}\:\key{else}\:\{s_2\};s\},q)\\
	\to ob(o,a,\{l \mid s_1;s\},q)}
\hfill\qquad
\ntyperule{Cond-False}{ \false = \feval{e}{(a+l)} 	
	}{ ob(o,a,\{l \mid \key{if}\ e\ \key{then}\:\{s_1\}\:\key{else}\:\{s_2\};s\},q)\\
	\to ob(o,a,\{l \mid s_2;s\},q)}
\qquad
\\
\\
\qquad \quad
\ntyperule{Await-True}{
	f = \feval{e}{(a+l)}  \quad v \neq \bot
	}{
	%\{
	ob(o,a,\{l \mid \await \, e \, \mbox{\tt ?} ; s\},q)\ \fut(f,v) %\}
	\\
	\to 
	ob(o,a,\{l \mid s\},q)\ \fut(f,v) 
	}
\hfill
\ntyperule{Await-False}{f = \feval{e}{(a+l)} 
	}{
	ob(o,a,\{l \mid \await \, e \, \mbox{\tt ?} ;s \},q)\  \fut(f, \bot)
	\\
	\to 
	ob(o,a,\text{idle} ,q \cup \{ l \mid
	\await\, e \, \mbox{\tt ?} ;s \})\ \fut(f, \bot)
	}
\qquad \quad
\\
\\
\nredrule{Release-Cog}{
	ob(o,a,\text{idle},q)\ cog(c, o)\\\to
    ob(o,a,\text{idle}, q)\ cog(c, \epsilon)
    }
\quad
\ntyperule{Activate}{
	c = a(\text{cog})
	}{
	ob(o,a,\text{idle},q \cup \{ l \mid s\})\ \cog(c, \epsilon) %\ cn\}
	\\ \to
	ob(o,a,\{ l \mid s\},q)\ \cog(c, o) %\ cn\}
	}
\quad
\ntyperule{Read-Fut}{
	f = \feval{e}{(a+l)}  \quad v \neq \bot
	}{ob(o,a,\{l \mid x=\eget e;s\},q) \ \textit{fut}(f,v)\\
\to ob(o,a,\{l \mid x=v;s\},q)\ \textit{fut}(f,v)}
\\
\\
\qquad
\ntyperule{New-Object}{
	o' = \text{fresh}(\name{C}) \quad
 	p=\text{init}(\name{C},o')
	\\ 
	a'=\text{atts}(\name{C},\feval{\many{e}}{(a+l)}, c)
	}{
	ob(o,a,\{l \mid x=\key{new}\ \name{C}(\many{e});s\},q)\ cog(c, o)
	\\
	\to ob(o,a,\{l \mid x=o';s\},q)\ cog(c, o) \\
  	ob(o',a',\text{idle},\{p\})}
\hfill
\ntyperule{New-Cog-Object}{
	c' = \text{fresh}(\,) \quad o' = \text{fresh}(\name{C}) \quad  
	p=\text{init}(\name{C},o') 
	\\  a'=\text{atts}(\name{C},\feval{\many{e}}{(a+l)},c')
	}{
	ob(o,a,\{l \mid x=\key{new}\ \key{cog}\ \name{C}(\many{e});s\},q)
	\\
	\to ob(o,a,\{l \mid x=o';s\},q)\\ ob(o',a',p,\emptyset)\quad cog(c', o')
	}
\qquad
\end{array}$
\caption{\label{fig:sem1} Semantics of \coreABS (1).}
\end{figure*}
\begin{figure*}[t]
\centering \myttsize 
\renewcommand{\arraystretch}{0.9} 
$\begin{array}{c}
\ntyperule{Async-Call}{
	o'=\feval{e}{(a+l)}
	\quad \many{v}=\feval{\many{e}}{(a+l)}\quad f = \text{fresh}(~)}
	{ob(o,a,\{l \mid x=e \async \m(\many{e});s\},q) \\\to  ob(o,a,\{l \mid x=f;s\},q) \
  	\invoc(o',f,\m,\many{v})\ \fut(f,\bot)}
\qquad
\ntyperule{Bind-Mtd}{
	\{ l \mid s \} =\text{bind}(o,f,\m,\many{v},\text{class}(o))
	}{
	ob(o,a,p,q)\  invoc(o,f,\m,\many{v})\\\to ob(o,a,p,q \cup \{ l \mid s \} )
	}
\\
\\
\qquad \ntyperule{Cog-Sync-Call}{
	o'=\feval{e}{(a+l)} \quad \many{v}=\feval{\many{e}}{(a+l)} \quad 
	f = \text{fresh}(\,)
	\\
  	c= a'(\text{cog}) \quad f'=l(\text{destiny}) 
	\\ 
	\{l' \mid s'\} = \text{bind}(o', f, \m, \many{v}, class(o'))
	}{
	ob(o,a,\{l \mid x=e \seqpoint \m(\many{e});s\},q)\\ ob(o',a',\text{idle},q')\ cog(c, o)\\
	\to ob(o,a,\text{idle},q \cup \{l \mid \key{await} \; f \key{?} ; x=f.\key{get};s\})\ \textit{fut}(f,\bot)
	\\ ob(o',a',\{l' \mid s';\key{cont} \ f'\},q')\ \cog(c, o')
	}
\hfill
\ntyperule{Cog-Sync-Return-Sched}{
	c = a'(\text{cog})  \quad f = l'(\text{destiny})
	}{
	ob(o,a,\{l \mid \key{cont} \ f \},q)\ cog(c, o)\\
  	ob(o',a',\text{idle},q' \cup \{l' \mid s\})\\
	\to ob(o,a,\text{idle},q)\ cog(c, o')\\
	ob(o',a',\{l' \mid s\},q')
	} \qquad
\\
\\
\ntyperule{Self-Sync-Call}{
	f'=l(\text{destiny}) \quad o=\feval{e}{(a+l)} \quad 
	\many{v}=\feval{\many{e}}{(a+l)}
	\\
	f= \text{fresh}(\,)\quad \{l' \mid s'\} = \text{bind}(o, f, \m, \many{v}, class(o))
	}{
	ob(o,a,\{l \mid x=e \seqpoint \m(\many{e});s\},q)\\
	\to ob(o,a,\{l' \mid s';\key{cont}(f')\},q \cup \{ l \mid \key{await} \; f \key{?} ;
	 x=f.\key{get};s\})\
	\textit{fut}(f,\bot)
	}
\qquad
%\ntyperule{Return}{
%	v = \feval{e}{(a+l)}  \quad f = l(\text{destiny}) }
%	{ob(o,a,\{l \mid \key{return}\ e\},q) \ \textit{fut}(f,\bot)\\
%	\to ob(o,a,\nidle,q) \ \textit{fut}(f,v)
%	}
\ntyperule{Return}{
	v = \feval{e}{(a+l)}  \quad f = l(\text{destiny}) }
	{ob(o,a,\{l \mid \key{return}\ e;s\},q) \ \textit{fut}(f,\bot)\\
	\to ob(o,a,\{l \mid s\},q) \ \textit{fut}(f,v)
	}
\\
\\
\qquad
\ntyperule{Rem-Sync-Call}{
	o'=\feval{e}{(a+l)}\quad f= \text{fresh}(\,)
	\quad  a(\text{cog}) \neq a'(\text{cog})
	}{
	ob(o,a,\{l \mid x=e \seqpoint \m(\many{e});s\},q)\ ob(o',a',p,q')\\
\to ob(o,a,\{l \mid f=e \async \m(\many{e}); x=f.\key{get};s\},q)\\ ob(o',a',p,q')}
\qquad
\ntyperule{Self-Sync-Return-Sched}{
	f = l'(\text{destiny})
	}{
	ob(o,a,\{l \mid \key{cont}(f)\},q \cup \{l' \mid s\})\\
	\to ob(o,a,\{l' \mid s\},q)
	}
\qquad
\ntyperule{Context}{
	\nt{cn} \to \nt{cn}'
	}{
	\nt{cn} \; \nt{cn}''
	\to \nt{cn}' \; \nt{cn}''
	}
\end{array}$
\caption{\label{fig:sem2} Semantics of \coreABS (2).} 
\end{figure*}
We focus on rules concerning the concurrent part of {\coreABS}, since
the other ones are standard.
Rules \rulename{Await-True} and \rulename{Await-False} model
the $\await\ e \mbox{\tt ?}$ operation: if the (future) value of $e$ has been computed then 
 $\name{await}$ terminates; otherwise the active process
becomes $\text{idle}$. In this case, if the object owns the control of the cog then
it may release such control -- rule \rulename{Release-Cog}. Otherwise, when the
cog has no active process, the object gets the control of the cog and activates 
one of its processes -- rule \rulename{Activate}.
Rule \rulename{Read-Fut} permits the retrieval of the value returned 
by a method; the object does not release the control of the cog until this value
has been computed.

The two types of object creation are modeled by 
\rulename{New-Object} and \rulename{New-Cog-Object}. The first one creates the new
object in the same cog. The new object is idle because the  cog has already an
active object. The second one creates the object in a new cog and makes it 
active by scheduling the process corresponding to the main function of the class.
The special field $\cog$ is initialized accordingly; the other object's fields
are initialized by evaluating the arguments of the operation -- see definition of 
\text{atts}.

Rule  \rulename{Async-Call} defines asynchronous method invocation
$x=e \async \m(\many{e})$. This 
rule creates a fresh future name that is assigned to the identifier $x$. The 
evaluation of the called method is transferred to a different process -- see rule
\rulename{Bind-Mtd}. Therefore the caller can progress without waiting for callee's
termination.
Rule  \rulename{Cog-Sync-Call} defines synchronous method invocation on an object in
the \emph{same} cog (because of 
the premise $a'(\text{cog}) = c$ and the element $\cog(c,o)$ 
in the configuration). The control is passed to the called object that executes
the body of the called method followed by a special statement $\name{cont}(f')$, where $f'$ is 
a fresh future name. When the evaluation of the body terminates, the caller process
is scheduled again using the name $f'$ -- see rule \rulename{Cog-Sync-Return-Sched}.
Rules \rulename{Self-Sync-Call} and \rulename{Rem-Sync-Call} deal with synchronous
method invocations of the same object and of objects in different cogs, respectively.
The former is similar to \rulename{Cog-Sync-Call} except that there is no control
on cogs. The latter one implements the synchronous invocation through an asynchronous
one followed by an explicit synchronisation operation.

It is worth to observe that the rules \rulename{Activate}, \rulename{Cog-Sync-Call} and
\rulename{Self-Sync-Call} are different from the corresponding ones in~\cite{johnsen10fmco}.
In fact, in~\cite{johnsen10fmco} rule \rulename{Activate} uses an unspecified \emph{select}
predicate that activates  one task from the queue of processes to evaluate. According to
the rules  \rulename{Cog-Sync-Call} and
\rulename{Self-Sync-Call} in that paper, the activated process might be a caller of a synchronous invocation, which has a \Abs{get} operation. To avoid potential deadlock of
a wrong \emph{select} implementation, we have prefixed the \Abs{get}s in 
\rulename{Cog-Sync-Call} and \rulename{Self-Sync-Call} with \Abs{await} operations.

The initial configuration of a {\coreABS} program with main function 
$\{ \bar{T}:\bar{x} \, ; \, s\}$ is  
\[
\begin{array}{l}
{\it ob}({\it start}, \varepsilon, \{ [ \text{destiny} \mapsto f_{\it start}, 
\bar{x} \mapsto \undefined] \, | \, s \}, \varnothing) \\ \cog(\text{start},{\it start})
\end{array}
\]
 where
$\text{start}$ and ${\it start}$ are special cog and object names, respectively, 
and $f_{\it start}$ is a fresh future name.
As usual, let $\lred{}^*$ be the reflexive and transitive closure of $\lred{}$. 

A configuration $\nt{cn}$ is \emph{sound} if\vspace{-.45em}
\begin{itemize}
\item[(i)] different elements $\cog(c,o)$ and
$\cog(c',o')$ in $\nt{cn}$ are such that $c \neq c'$ and $o \neq \varepsilon$ implies $o \neq o'$,
\item[(ii)] if $ob(o,a,p,q)$ and $ob(o',a',p',q')$ are different objects in $\nt{cn}$
such that $a(\text{cog}) = a'(\text{cog})$ then either $p = \text{idle}$ or
$p' = \text{idle}$.
\end{itemize}\vspace{-.4em}
We notice that the initial configurations of {\coreABS} programs are sound.
 The following statement guarantees that the property
``there is at most one active object per cog'' is an 
invariance of the transition relation.

\begin{proposition}
If $\nt{cn}$ is \emph{sound} and $\nt{cn} \lred{} \nt{cn}'$ then $\nt{cn}'$ is sound as well.
\end{proposition}

As an example of {\coreABS} semantics, in Figure~\ref{fig:red_cpx_sched} we have 
detailed the transitions of the program in Example~\ref{ex.cpx_sched}. 
The non-interested reader may safely skip it.

\subsection{\em Samples of concurrent programs in {\coreABS}}
The {\coreABS} code of two concurrent programs are discussed. These codes will be
analysed in the following sections. 

\begin{example}
Figure~\ref{fig.Math} collects three different implementations of the factorial 
function in a class \name{Math}.
  \begin{figure}[t]
    \begin{absexamplen}
class Math { 
   Int fact_g(Int n){ 
      Fut<Int> x ;
      Int m ;
      if (n==0) { return 1; } 
      else { x = this!fact_g(n-1); m = x.get; 
              return n*m; } 
   } 
   Int fact_ag(Int n){ 
      Fut<Int> x ;
      Int m ;
      if (n==0) { return 1; }
      else { x = this!fact_ag(n-1); 
              await x?; m = x.get;
              return n*m; }
   } 
   Int fact_nc(Int n){ 
      Fut<Int> x ;
      Int m ;
      Math z ;
      if (n==0) { return 1 ; } 
      else {  z = new cog Math();
              x = z!fact_nc(n-1); m = x.get; 
              return n*m; } 
   } 
}
    \end{absexamplen}
    \caption{The class \name{Math}}\label{fig.Math}
  \end{figure}
  The function \Abs{fact\_g} is the standard definition of factorial:
  the recursive invocation \Abs{this\!fact\_g(n-1)} is followed by a
  \Abs{get} operation that retrieves the value returned by the
  invocation.  Yet, \Abs{get} does not allow the task to release the
  cog lock; therefore the task evaluating \Abs{this\!fact_g(n-1)} is
  fated to be delayed forever because its object (and, therefore, the
  corresponding cog) is the same as that of the caller.  The function
  \Abs{fact\_ag} solves this problem by permitting the caller to
  release the lock with an explicit \Abs{await} operation, before
  getting the actual value with \Abs{x.get}. An alternative solution
  is defined by the function \Abs{fact\_nc}, whose code is similar to
  that of \Abs{fact\_g}, except for that \name{fact\_nc} invokes
  \Abs{z\!fact\_nc(n-1)} recursively, where \name{z} is an object in a
  new cog. This means the task of \Abs{z\!fact\_nc(n-1)} may start
  without waiting for the termination of the caller.
\end{example}

Programs that are particularly hard to verify are those that may manifest misbehaviours
according to the schedulers choices. The following example discusses one case.
%The example is also interesting because it highlights the relevance of collecting
%the cogs of objects fields together with the cog of the objects.

\begin{example}\label{ex.cpx_sched}
The class \Abs{CpxSched} of Figure~\ref{fig:exampleC} defines three  methods.
\begin{figure}
  \centering
  
  \begin{absexamplen}
interface I  { Fut<Unit> m1(I y); Unit m2(I z); 
                Unit m3() ; }

class CpxSched (I u) implements I { 
    Fut<Unit> m1(I y) { 
        Fut<Unit> h;
        Fut<Unit> g ;
        h = y!m2(u);
        g = u!m2(y); 
        return g;
    }
    Unit m2(I z) { 
        Fut<Unit> h ;
        h = z!m3(); 
        h.get; 
     }
    Unit m3(){ 
    }
}
  \end{absexamplen}
  \caption{The class \Abs{CpxSched}}
\label{fig:exampleC}
\end{figure}
Method \Abs{m1} asynchronously invokes \Abs{m2} on its own argument \Abs{y}, passing to it the field \Abs{x} as argument . Then it asynchronously invokes \Abs{m2} on the field \Abs{x}, passing its same argument \Abs{y}.
Method \Abs{m2} invokes \Abs{m3} on the argument \Abs{z} and blocks waiting for the result.
Method \Abs{m3} simply returns.
% %
% Method \Abs{m4} invokes \Abs{m2} on the argument \Abs{w} passing the field \Abs{x} as parameter and returns the future reference of that invocation.

Next, consider the following main function:

\begin{absexamplen} 
  { I x; I y; I z;
    Fut<Fut<Unit>> w ;
    x = new CpxSched(null);
    y = new CpxSched(x);
    z = new cog CpxSched(null);
    w = y!m1(z); }
\end{absexamplen}

The initial configuration is 
\[
ob(start,\Empty,\{l \; | \; s\},\emptyset) \cog(\text{start},{\it start})
\]
where $l = [\text{destiny} \mapsto f_{\it start}, 
x \mapsto \undefined,  y \mapsto \undefined, z \mapsto \undefined, w \mapsto \undefined ]$ and $s$ is the statement of the main function.
The sequence of transitions of this configuration is illustrated in 
Figure~\ref{fig:red_cpx_sched}, 
where 
\begin{center}
\begin{tabular}{rl} 
$s'$, $s''$, & $s'''$ are the obvious sub-statements of the 
\\
& main function
\\
$l_{o} =$ & $[\text{destiny} \mapsto f'', z \mapsto o'',  u \mapsto \undefined, h \mapsto \undefined]$ 
\\
$l_{o'} =$ & $[\text{destiny} \mapsto f, y \mapsto o'',  g \mapsto \undefined, h \mapsto \undefined]$
\\ 
$l_{o''} =$ & $[\text{destiny} \mapsto f', z \mapsto o,  u \mapsto \undefined, h \mapsto \undefined]$
\\
$l'_{o} = $ & $[\text{destiny} \mapsto f'''']$
\\
$l'_{o''} =$ & $[\text{destiny} \mapsto f''']$
\\ 
$a_{\tt null} = $ & $[\nt{cog} \mapsto \text{start}, {\tt x} \mapsto {\tt null}]$
\\
$a_o = $ & $[\nt{cog} \mapsto \text{start},{\tt x} \mapsto o]$
\\
$s_{o'}=$ & \text{\small {\tt h= y!m2(this.x); g= this.x!m2(y); return g;}}
\\
$s_{o}= $ & $s_{o''}={}$ \text{\small {\tt h= z!m3(); h.get;}}
\end{tabular}
\end{center}
\begin{figure*}
\newcommand{\msp}{\;}
  \centering
  $$
\begin{array}{l}
  \nt{ob}(\nt{start},\Empty,\{l \; | \; s\},\emptyset) \msp \cog(\text{start},{\it start}) \\
\qquad \lred{}^2 \quad \rulename{New-Object} \mbox{ and } \rulename{Assign-Local} \\
 \nt{ob}(\nt{start},\Empty,\{l[{\tt x} \mapsto o] \,| \, \mbox{\Abs{y = new C(x)}} ; s''\},\emptyset) \msp \nt{cog}(\text{start},\nt{start})
   \msp \nt{ob}(o,a_{\tt null},\text{idle},\emptyset) \\
\qquad  \lred{}^2 \quad \rulename{New-Object} \mbox{ and } \rulename{Assign-Local}  \\
  \nt{ob}(\nt{start},\Empty,\{l[{\tt x} \mapsto o, {\tt y} \mapsto o'] \,| \, \mbox{\Abs{z = new cog C(null)}}; s'''\},\emptyset) \msp \nt{cog}(\text{start},\nt{start}) 
   \msp \nt{ob}(o,a_{\tt null}, \text{idle},\emptyset) \msp \nt{ob}(o',a_o,\text{idle},\emptyset)  \\
\qquad  \lred{}^2 \quad \rulename{New-Cog-Object} \mbox{ and }  \rulename{Assign-Local} \\
  \nt{ob}(\nt{start},\Empty,\{l[{\tt x} \mapsto o, {\tt y} \mapsto o', {\tt z} \mapsto o''] \,| \, \mbox{\Abs{w = y}!\Abs{m1(z);}} \},\emptyset) \msp \nt{cog}(\text{start},\nt{start}) 
   \msp \nt{ob}(o,a_{\tt null}, \text{idle},\emptyset) \msp \nt{ob}(o',a_o,\text{idle},\emptyset)\\
  \qquad\quad \nt{ob}(o'',[\nt{cog} \mapsto c,{\tt x} \mapsto {\tt null}],\text{idle},\emptyset) \msp \nt{cog}(c,o'') \\
\qquad \lred{} \quad \rulename{Async-Call} \\
  \nt{ob}(\nt{start},\Empty,\{l[{\tt x} \mapsto o, {\tt y} \mapsto o', {\tt z} \mapsto o''] \,| \, \mbox{\Abs{w =} f;} \},\emptyset) \msp \nt{cog}(\text{start},\nt{start}) 
   \msp \nt{ob}(o,a_{\tt null},\text{idle},\emptyset) \msp  \nt{ob}(o',a_o,\text{idle},\emptyset) \\
  \qquad\quad\nt{ob}(o'',[\nt{cog} \mapsto c,{\tt x} \mapsto {\tt null}],\text{idle},\emptyset) \msp \nt{cog}(c,o'') \msp\nt{invoc}(o',f,\m1,o'') \msp \nt{fut}(f,\bot) \\
\qquad \lred{} \quad \rulename{Bind-Mtd} \\
  \nt{ob}(\nt{start},\Empty,\{l[{\tt x} \mapsto o, {\tt y} \mapsto o', {\tt z} \mapsto o''] \,| \, \mbox{\Abs{w =} f;} \},\emptyset) \msp \nt{cog}(\text{start},\nt{start}) 
   \msp \nt{ob}(o,a_{\tt null},\text{idle},\emptyset) \msp \nt{ob}(o',a_o,\text{idle},\{l_{o'}\,|\, s_{o'}\}) \\
   \nt{ob}(o'',[\nt{cog} \mapsto c,{\tt x} \mapsto {\tt null}],\text{idle},\emptyset)
  \qquad\quad \nt{cog}(c,o'')\msp \nt{fut}(f,\bot) \\
\qquad  \lred{}^+  \rulename{Activate} \mbox{ and twice }\rulename{Async-Call}\mbox{ and } \rulename{Return}  \\
  \nt{ob}(\nt{start},\Empty,\{l[{\tt x} \mapsto o, {\tt y} \mapsto o', {\tt z} \mapsto o''] \,| \, \mbox{\Abs{w =} f;} \},\emptyset) \msp \nt{cog}(\text{start},\nt{start}) 
   \msp \nt{ob}(o,a_{\tt null},\text{idle},\emptyset) \msp \nt{ob}(o',a_o,\text{idle},\emptyset) \\
  \qquad\quad\nt{ob}(o'',[\nt{cog} \mapsto c,{\tt x} \mapsto {\tt null}],\text{idle},\emptyset) \msp \nt{cog}(c,o'')\msp \nt{fut}(f,f'') \msp \nt{fut}(f',\bot)
  \msp \nt{fut}(f'',\bot) \hfill (\star)
  \\
  \qquad  \quad \nt{invoc}(o'',f',\m2,o)  \msp \nt{invoc}(o,f'',\m2,o'') 
    \\
\qquad \lred{}^2 \quad \mbox{ twice }\rulename{Bind-Mtd}  \\
  \nt{ob}(\nt{start},\Empty,\{l[{\tt x} \mapsto o, {\tt y} \mapsto o', {\tt z} \mapsto o''] \,| \, \mbox{\Abs{w =} f;} \},\emptyset)  \msp  \nt{cog}(\text{start},\nt{start}) 
   \msp \nt{ob}(o,a_{\tt null},\text{idle},\{l_{o}\,|\,s_{o}\}) \msp  \nt{ob}(o',a_o,\text{idle},\emptyset) \\
  \qquad\quad\nt{ob}(o'',[\nt{cog} \mapsto c,{\tt x} \mapsto {\tt null}],\text{idle}, \{l_{o''}\,|\,s_{o''}\}) \msp \nt{cog}(c,o'') \msp \nt{fut}(f,f'')
   \msp \nt{fut}(f',\bot) \msp \nt{fut}(f'',\bot) \\
\qquad \lred{}^+ \quad \mbox{ twice }\rulename{Activate} \mbox{ and twice }\rulename{Async-Call}  \\
  \nt{ob}(\nt{start},\Empty,\{l[{\tt x} \mapsto o, {\tt y} \mapsto o', {\tt z} \mapsto o'', {\tt w} \mapsto f] \,| \, \text{idle} \},\emptyset) \msp \nt{cog}(\text{start},o) 
   \msp\nt{ob}(o,a_{\tt null},\{l_{o}[\mbox{\Abs{h}}\mapsto f''']\,|\,\mbox{\Abs{h.get;}}s'_{o}\},\emptyset) \\
  \qquad\quad\nt{ob}(o',a_o,\text{idle},\emptyset)
   \msp \nt{ob}(o'',[\nt{cog} \mapsto c,{\tt x} \mapsto {\tt null}],\{l_{o''}[\mbox{\Abs{h}}\mapsto f'''']\,|\,\mbox{\Abs{h.get;}}s'_{o''}\},\emptyset) \msp \nt{cog}(c,o'')
   \msp \nt{fut}(f,f'') \msp \nt{fut}(f',\bot) \\ 
  \qquad\quad \nt{fut}(f'',\bot) \msp\nt{invoc}(o'',f''',\m3,\varepsilon) \msp \nt{fut}(f',\bot) \msp \nt{fut}(f''',\bot) \msp \nt{invoc}(o,f'''',\m3,\varepsilon)
   \msp \nt{fut}(f'',\bot) \msp \nt{fut}(f'''',\bot)  \\
\qquad \lred{}^2 \quad \mbox{ twice }\rulename{Bind-Mtd}  \\
  \nt{ob}(\nt{start},\Empty,\{l[{\tt x} \mapsto o, {\tt y} \mapsto o', {\tt z} \mapsto o'', {\tt w} \mapsto f] \,| \, \text{idle} \},\emptyset) \msp \nt{cog}(\text{start},o) 
   \msp \nt{ob}(o,a_{\tt null},\{l_{o}[\mbox{\Abs{h}}\mapsto f''']\,|\,\mbox{\Abs{h.get;}}s'_{o}\},\{l'_{o}\,|\mbox{\Abs{skip;}}\})  \\
  \qquad\quad\nt{ob}(o',a_o,\text{idle},\emptyset) \msp
   \nt{ob}(o'',[\nt{cog} \mapsto c,{\tt x} \mapsto {\tt null}],\{l_{o''}[\mbox{\Abs{h}}\mapsto f'''']\,|\,\mbox{\Abs{h.get;}}s'_{o''}\},\{l'_{o''}\,|\,\mbox{\Abs{skip;}}\}) \msp \nt{cog}(c,o'') \\
  \qquad\quad \nt{fut}(f,f'') \msp \nt{fut}(f',\bot) \msp \nt{fut}(f'',\bot) \msp \nt{fut}(f',\bot)\msp \nt{fut}(f''',\bot) \msp \nt{fut}(f'',\bot)\msp \nt{fut}(f'''',\bot)
\end{array}
$$
\caption{Reduction of Example~\ref{ex.cpx_sched}}
\label{fig:red_cpx_sched}
\end{figure*}
We notice that the last configuration of Figure~\ref{fig:red_cpx_sched} is stuck, 
\emph{i.e.}~it cannot progress anymore. In fact, it is a deadlock according the
forthcoming Definition~\ref{def.locks}.
\end{example}

\subsection{\em Deadlocks}
\label{sec.deadlockanalysis}

The definition below identifies deadlocked configurations by detecting chains of dependencies 
between tasks that cannot progress. To ease the reading, we write
\begin{itemize}
\item[--]
 $p[f.\get]^a$ 
whenever $p = \{l | s\}$ and $s$ is $x = y.\get ; s'$ and $\feval{y}{(a+l)}=f$;
\item[--]
$p[\await\ f]^a$ 
whenever $p = \{l | s\}$ and $s$ is $\await\ e\mbox{\tt ?} ; s'$ and $\feval{e}{(a+l)}=f$;
\item[--]
$p.f$ whenever $p=\{l|s\}$ and $l(\text{destiny}) = f$.
\end{itemize}

\begin{definition}
\label{def.locks}
A configuration $\nt{cn}$ is \emph{deadlocked}  if there are
$$
\begin{array}{l}
\quad \nt{ob}(o_0,a_0,p_0,q_0), \cdots, \nt{ob}(o_{n-1},a_{n-1},p_{n-1},q_{n-1}) \in \nt{cn}
\\
\text{and}
\\
\quad
p_i' \in p_i \cup q_i, \qquad \text{with} \; 0 \leq i \leq n-1
\end{array}
$$ 
such that (let $+$ be computed modulo $n$ in the following) 
\begin{enumerate}
\item
$p_0' = p_0[f_0.\get]^{a_0}$ and if $p_i'[f_i.\get]^{a_i}$ then $p_{i}' = p_{i}$;

\item  
if $p_i'[f_i.\get]^{a_i}$ or $p_{i}'[\await \ f_{i}]^{a_{i}}$ 
then $\nt{fut}(f_i,\bot) \in \nt{cn}$ and
\begin{itemize}
\item[--]
either $p_{i+1}'[f_{i+1}.\get]^{a_{i+1}}$ 
and 
$p_{i+1}' = \{ l_{i+1} | s_{i+1}\}$ and $f_i = l_{i+1}(\text{destiny})$;
\item[--]
or $p_{i+1}'[\await \; f_{i+1}]^{a_{i+1}}$ and
$p_{i+1}' = \{ l_{i+1} | s_{i+1}\}$ and 
$f_i = l_{i+1}(\text{destiny})$;
\item[--]
or 
$p_{i+1}' = p_{i+1} = \text{idle}$ and $a_{i+1}(\nt{cog}) = a_{i+2}(\nt{cog})$ 
and $p_{i+2}'[f_{i+2}.\get]^{a_{i+2}}$
(in this case $p_{i+1}$ is idle, by soundness).
\end{itemize}
\end{enumerate}

A configuration $\nt{cn}$ is \emph{deadlock-free} if, for every 
$\nt{cn} \lred{}^* \nt{cn}'$, $\nt{cn}'$ is not deadlocked. 
A {\coreABS} program is \emph{deadlock-free} if its initial configuration
is \emph{deadlock-free}.
\end{definition}

According to Definition~\ref{def.locks}, a configuration is deadlocked when there
is a circular dependency between processes. The processes involved
in such circularities are performing a \Abs{get} or \Abs{await} synchronisation or 
they are idle and will never grab the lock because another active process in the
same cog will not
return. We notice that, by Definition~\ref{def.locks}, at least one active process is blocked
on a \Abs{get} synchronisation. We also notice that the objects in Definition~\ref{def.locks} may be not pairwise different (see example 1 below). 
The following examples should make the definition clearer; the reader is recommended to instantiate the definition every time. 

\begin{enumerate}
\item 
(self deadlock)
$$\begin{array}{l}
  \nt{ob}(o_1,a_1,\{l_1 | x_1=e_1.\get;s_1\},q_1) \\ 
  \nt{ob}(o_2,a_2, \text{idle},q_2\cup\{l_2|s_2\}) \\ 
  \nt{fut}(f_2,\bot),
\end{array}$$ 

where $\feval{e_1}{(a_1+l_1)}=l_2(destiny)=f_2$ and
$a_1(\nt{cog})=a_2(\nt{cog})$. In this case, the object $o_1$ keeps the control of
its own cog while waiting for the result of a process in $o_2$. This process cannot
be scheduled because the corresponding cog is not released. A similar situation
can be obtained with one object:
$$\begin{array}{l}
  \nt{ob}(o_1,a_1,\{l_1 | x_1=e_1.\get;s_1\},q_1 \cup \{l_2|s_2\}) \\ 
  \nt{fut}(f_2,\bot),
\end{array}$$ 
where $\feval{e_1}{(a_1+l_1)}=l_2(destiny)=f_2$.
In this case, the objects of the Definition~\ref{def.locks} are
\[
\nt{ob}(o_1,a_1,p_1,q_1)
\quad
\nt{ob}(o_1,a_1,p_2,q_2 \cup \{ l_2|s_2\})
\]
where $p_1' = p_1 = \{l_1 | x_1=e_1.\get;s_1\}$, $p_2' = \{ l_2|s_2\}$ and
$q_1 = q_2 \cup \{ l_2|s_2\}$.

\item ($\get$-$\await$ deadlock)
$$\begin{array}{l}
\nt{ob}(o_1,a_1,\{l_1 | x_1=e_1.\get;s_1\},q_1) \\
\nt{ob}(o_2,a_2,\{l_2 | \await\ e_2\mbox{\tt ?};s_2\},q_2) \\ 
\nt{ob}(o_3,a_3,\text{idle}, q_3 \cup \{l_3 | s_3\})
\end{array}$$
where $l_3(destiny)=\feval{e_2}{a_2+l_2}$, $l_2(destiny)=\feval{e_1}{a_1+l_1}$,
$a_1(cog)=a_3(cog)$ and $a_1(cog) \neq a_2(cog)$. In this case, the objects $o_1$ and
$o_2$ have different cogs. However $o_2$ cannot progress because it is waiting 
for a result of a process that cannot be scheduled (because it has the same cog of 
$o_1$).

\item ($\get$-\text{idle} deadlock)
$$
\begin{array}{l}
   \nt{ob}(o_1,a_1,\{l_1 | x_1=e_1.\get;s_1\},q_1) \\ 
  \nt{ob}(o_2,a_2, \text{idle},q_1\cup\{l_2|s_2\}) \\ 
   \nt{ob}(o_3,a_3,\{l_3 | x_3=e_3.\get;s_3\},q_3) \\ 
  \nt{ob}(o_4,a_4, \text{idle},q_4\cup\{l_4|s_4\}) \\ 
   \nt{ob}(o_5,a_5,\{l_5 | x_5=e_5.\get;s_5\},q_5)\\ 
  \nt{fut}(f_1,\bot),\ \nt{fut}(f_2,\bot),\ \nt{fut}(f_4,\bot)
\end{array}
$$
where $f_2=l_2(destiny)=\feval{e_1}{a_1+l_1}$, $f_4=l_4(destiny)=\feval{e_3}{a_3+l_3}$, $f_1=l_1(destiny)=\feval{e_5}{a_5+l_5}$ and $a_2(\nt{cog})=a_3(\nt{cog})$ and  $a_4(\nt{cog})=a_5(\nt{cog})$.

 % $$\ntask{\task_1}{\obj_1}{\locked}{\task_2.\get},\ntask{\task_2}{\obj_2}{\unlocked}{\e_2},\ntask{\task_3}{\obj_2}{\locked}{\task_4.\get},\ntask{\task_4}{\obj_4}{\unlocked}{\e_4},\ntask{\task_5}{\obj_4}{\locked}{\task_1.\get}$$.
\end{enumerate}

A deadlocked configuration has at least one object that is stuck (the one performing
the \Abs{get} instruction). This means that the configuration may progress, but future
configurations will still have one object stuck.

\begin{proposition}
If $\nt{cn}$ is deadlocked and
$\nt{cn} \lred{} \nt{cn}'$ then $\nt{cn}'$ is deadlocked as well.
\end{proposition}

Definition~\ref{def.locks} is about runtime entities that have no
static counterpart. Therefore we consider a  notion weaker than  deadlocked configuration. 
This last
notion will be used in the Appendices to demonstrate the correctness of 
the inference system  in Section~\ref{sec.FJg-contracts}.

\begin{definition}
\label{def.obj-circularity}
A configuration $\nt{cn}$ has
\begin{itemize}
\item[(i)]
a \emph{dependency} $(c,c')$ if
$$
\nt{ob}(o,a,\{l | x=e.\get;s\},q), \nt{ob}(o',a',p',q') \in \nt{cn}
$$
with $\feval{e}{(a+l)} = f$ and $a(\nt{cog})=c$ and $a'(\nt{cog})=c'$  
and 
\begin{itemize}
\item[(a)] either $\nt{fut}(f,\bot) \in \nt{cn}$,
$l'(\text{destiny}) = f$ and $\{l' | s'\} \in 
p' \cup q'$; 

\item[(b)] 
or $\nt{invoc}(o',f,\m,\vect{v}) \in \nt{cn}$.
\end{itemize}

\item[(ii)]
 a   \emph{dependency} $(c,c')^{\aa}$ if 
$$
\nt{ob}(o,a,p,q), \nt{ob}(o',a',p',q') \in \nt{cn}
$$
and $ \{l | \await\ e\mbox{\tt ?};s\} \in p \cup q$
and $\feval{e}{(a+l)} = f$ and
\begin{itemize}
\item[(a)] either $\nt{fut}(f,\bot) \in \nt{cn}$,
$l'(\text{destiny}) = f$ and $\{l' | s'\} \in 
p' \cup q'$; 

\item[(b)] 
or $\nt{invoc}(o',f,\m,\vect{v}) \in \nt{cn}$.
\end{itemize}
\end{itemize}

 Given a set $A$ of dependencies, let the {\get}\emph{-closure} of
 $A$, noted $A^\get$, be the least set such that
 \begin{enumerate}
 \item[1. ] $A \subseteq A^\get$;
 \item[2. ] if $(c,c') \in A^\get$ and $(c',c'')^{[\aa]} \in A^\get$ then
 $(c,c'') \in A^\get$, where $(c',c'')^{[\aa]}$ denotes either the pair 
 $(c',c'')$ or the pair $(c',c'')^{\aa}$.
 \end{enumerate}

 A configuration contains a \emph{circularity} if the {\get}-closure of its 
set of dependencies has a pair $(c,c)$.
 \end{definition}

\begin{proposition}
If a configuration is deadlocked then it has a circularity. The converse is
false.
\end{proposition}
\begin{proof}
The statement is a straightforward consequence of the definition of
deadlocked configuration. To show that the converse is false, consider the configuration
$$\begin{array}{l}
\nt{ob}(o_1,a_1,\{l_1 | x_1=e_1.\get;s_1\},q_1) \\ 
\nt{ob}(o_2,a_2,\text{idle}, q_2 \cup \{l_2 | \await\ e_2\mbox{\tt ?};s_2\}) \\
\nt{ob}(o_3,a_3,\{l_3 | \return\ e_3\},q_3)\quad\nt{cn}
\end{array}$$
where $l_3(destiny)=\feval{e_1}{a_1+l_1}$, $l_1(destiny)=\feval{e_2}{a_2+l_2}$,
$c_2=a_2(cog)=a_3(cog)$ and $c_1=a_1(cog)\neq c_2$. 
This configuration has the dependencies 
$$\{(c_1,c_2), (c_2,c_1)^\aa \}$$ whose {\get}-closure 
contains the circularity $(c_1,c_1)$. However
the configuration is not deadlocked.\qed

\end{proof}

\begin{example}\label{ex.cpx_sched-cont}
The final configuration of Figure~\ref{fig:red_cpx_sched} is \emph{deadlocked}
according to Definition~\ref{def.locks}.
In particular, there are two objects $o$ and $o''$ running on different cogs 
whose active processes have a $\get$-synchronisation on the result of process
in the other object: $o$ is performing a $\get$ on a future $f'''$ which is $l'_{o''}(\text{destiny})$, and $o''$ is performing a $\get$ on a future $f''''$ which is $l'_{o}(\text{destiny})$ and $\nt{fut}(f''',\bot)$ and 
$\nt{fut}(f'''',\bot)$. We notice that, if in the configuration $(\star)$ we choose
to evaluate $\nt{invoc}(o'',f',\m2,o)$ when the evaluation of 
$\nt{invoc}(o,f'',\m2,o'')$ has been completed (or conversely) then no deadlock is
manifested.
\end{example}
%%% Local Variables: 
%%% mode: latex
%%% TeX-master: "subm-SoSyM"
%%% End: 

\section{Restrictions of  {\coreABS} in the current release of the contract 
inference system}
\label{sec.restrictions}
%!TEX root = SoSyM.tex

The contract inference system we describe in the next section has been prototyped. 
To verify its feasibility, the current release of the prototype addresses a subset of 
{\coreABS} features. 
These restrictions permit to ease the initial
development of the inference system and do not jeopardise its
extension to the full language. Below we discuss the restrictions and, for each
of them, either we explain the reasons why they will be retained in the next
release(s) or we detail the techniques that will be used to remove them.
(It is also worth to notice that, notwithstanding the following restrictions, it is
still possible  to 
verify large commercial cases, 
such as a core component of FAS discussed in this paper.)

% FJf is a pure object oriented language.
% On the other hand, \ABS include a large variety of values different from object, called data-types that abstract classic data structures, like lists, sets or maps.
% what is the problem? Addition of an object to a set for instance.
% dealt with as if it were only one object...
% they are usually used to construct lists or maps or sets of objects, and to go over elements of such construct.
% In that case, we can abstract sets by one of their element, and while loop by only one execution.

\paragraph{Returns.}
{\coreABS} syntax admits \Abs{return} statements with continuations -- see Figure~\ref{fig:core:lan}
-- that, according to the semantics, are executed \emph{after the return value has been
delivered to the caller}. These continuations can be hardly controlled by programmers and
usually cause unpredictable behaviours, in particular as regards deadlocks. To increase the
precision of our analysis we assume that {\coreABS} programs have empty continuations of 
\Abs{return} statements. We observe that this constraint has an exception at run-time: in order
to define the semantics of synchronous method invocation, rules \rulename{Cog-Sync-Call}
and \rulename{Self-Sync-Call} append a \Abs{cont} $f$ continuation to statements in order to 
let the right caller be scheduled. Clearly this is the {\coreABS} definition of synchronous
invocation and it does not cause any misbehaviour. 

\paragraph{Fields assignments.}
Assignments in {\coreABS} (as usual in object-oriented languages) may update the 
fields of objects that are accessed concurrently by other threads, thus
could lead to indeterminate behaviour. 
In order to simplify the analysis, we constrain field assignments as follows.
If the field is \emph{not of future type} then 
we keep field's record structure unchanging. 
For instance, if a field contains an object
of cog $a$, then that field may be only updated with objects belonging to $a$
(and this correspondence must hold recursively with respect to the fields of
objects referenced by $a$).
When the field is of a primitive
type (\Abs{Int}, \Abs{Bool}, etc.)
this constraint is equivalent to the standard type-correctness.
It is possible to be more liberal as regards fields assignments.
In~\cite{GiachinoLascu} an initial study for covering full-fledged field
assignments was undertaken using so-called union types (that is, by extending
the syntax of future records with a $+$ operator, as for contracts, see below) and
collecting all records in the inference rule of the field assignment (and the conditional).
When the field is \emph{of future type} then we disallow assignments. In fact, 
assignments to such fields allow a programmer to define unexpected behaviours.
Consider for example the following class  \Abs{Foo} implementing \Abs{I_Foo}:
\begin{absexamplen}[numbers=right]
class Foo(){
	Fut<T> x ;
	Unit foo_m () {
		Fut<T> local ;
		I_Foo y = new cog Foo() ;
		I_Foo z = new cog Foo() ;
		local = y!foo_n(this) ;
		x = z!foo_n(this) ;
		await local? ;
		await x?
	}
	T foo_n(I_Foo x){ . . . }
}
\end{absexamplen}
If the main function is
\begin{absexamplen}[firstnumber=14,numbers=right]
	{    I_Foo x ;
	   Fut<Unit> u ; 
	   Fut<Unit> v ;
	   x = new cog Foo() ;
	   u = x!foo_m() ;
	   v = x!foo_m() ; 	}
\end{absexamplen}
then the two invocations in lines 18 and 19 run in parallel.
Each invocation of \Abs{foo_m} invokes \Abs{foo_n} twice that apparently terminate when \Abs{foo_m} returns (with the two final \Abs{await} statements). 
However this may be not the case because the invocation of \Abs{foo_n} line 8 is stored in a field:
 consider that the first invocation of \Abs{foo_m} (line 18) starts executing, sets the field \Abs{x} with its own future $f$,
 and then, with the \Abs{await} statement in line 9, the second invocation of \Abs{foo_m} (line 19) starts.
That second invocation {\em replaces} the content of the field \Abs{x} with its own future $f'$:
 at that point, the second invocation (line 19) will synchronise with $f'$ before terminating,
 then the first invocation (line 18) will resume and also synchronised with $f'$ before terminating.
Hence, even after both invocations (line 18 and 19) are finished, the invocation of \Abs{foo_n} in line 8 may still be running.
It is not too difficult to trace such residual behaviours 
in the inference system (for instance, by grabbing them using a function
like $\contract{\Gamma}$). However, this extension will  
entangle the inference system and for this reason we decided to deal with generic
field assignments in a next release.
%
%without requiring any sophisticated technique to be
%solved. 

It is worth to recall that these restrictions
does not apply to local variables of methods, as they can only be 
accessed by the method in which they are declared. 
Actually, the foregoing inference algorithm tracks changes of local variables.

% FJg cannot handle assignement in general, as it associate to each field and variable only one record.
% we do the following extension of FJg and restriction on ABS:
% it is actually sound to change the record of a local variable, because there are no concurrent task that can access it. <- I should find a better explanation, but it's its core.
% all the objects assigned to the same field should be equal.
%
%

%THIS EXTENSION HAS BEEN DONE!
%In the current prototype we disallow assignments between future variables. These 
%kind of assignments create \emph{aliases} and their correct management would entangle
%the inference system (which is already complex).
%The extension of the inference system with a standard 
%alias analysis will permit the removal of this constraint.

\paragraph{Interfaces.}
In {\coreABS} objects are typed with interfaces, which may 
have several implementations. As a consequence,
when a method is invoked, it is in general not possible to statically determine 
which method will be executed at runtime (dynamic dispatch). This is
problematic for our technique because it breaks the association of a 
unique abstract behaviour with a method invocation.
In the current inference system this issue is avoided by constraining codes
to have interfaces implemented by at 
most one class.
%%
%As a consequence, we can construct a partial function from interfaces to class and each time we have a method call in the code,
% we can precisely know, with that function, which method of which class will be executed.
%
This restriction will be relaxed by admitting that methods
have multiple contracts, one for every possible implementation. In turn, method invocations are defined as the \emph{union} of the possible contracts a method has.  

\paragraph{Synchronisation on booleans.}
In addition to synchronisation on method invocations, {\coreABS} permits 
synchronisations on Booleans, with the statement {\Abs{await} $e$}.
When {$e$} is \Abs{False}, the execution of the method is suspended, and when 
it becomes \Abs{True}, the \Abs{await} terminates and the execution of the method may proceed.
It is possible that the expression $e$  refers to a field of an object that can be modified 
by another method. In this case,
the \Abs{await} becomes synchronised with any method that
may set the field to \name{true}.
This subtle synchronisation pattern is difficult to infer and, for this reason, we
have restricted the current release of {\DFfABS}.

Nevertheless, the current release of {\DFfABS} adopts a naive solution for 
 \Abs{await} statements on booleans, namely let programmers annotate them with the dependencies they create.
For example, consider the annotated code:

\begin{absexamplen}
 class ClientJob(...) {
  Schedules schedules = EmptySet; 
  ConnectionThread thread; 
  ...
  Unit executeJob() {
    thread = ...; 
    thread!command(ListSchedule);
    [thread] await schedules != EmptySet; 
    ... 
  }
}
\end{absexamplen}

The statement \Abs{await} compels the task to wait for \Abs{schedules} 
to be set to something different from the empty set. 
Since \Abs{schedules} is a field of the object, any concurrent thread (on that object) may update it.
%It is not evident how to extract this implicit dependency relation from the guard
%of \Abs{await}. 
%Therefore we constrain the programmer to  provide an annotation making 
%explicit the dependency. 
%
In the above case,
the object that will modify the boolean guard is stored in the variable \Abs{thread}. 
Thus the annotation
%Exactly this information is written in the annotation 
\Abs{[thread]} in the \Abs{await} statement. The current inference system of {\DFfABS}
is extended with rules dealing with \Abs{await} on boolean guard and, of course, the
correctness of the result depends on the correctness of the \Abs{await} annotations.
A data-flow analysis of boolean guards in \Abs{await} statements may produce
a set of cog dependencies that can be used in the inference rule of the
corresponding statement. While this is an interesting issue, it will not be
our primary concern in the near future.
%This way, it will be possible to define a dependency for that \Abs{await} statement in our inference system.

%This means that the programmer is obliged to write a set of dependencies between the \Abs{this} object and the objects whose
%methods may modify the field. [By default, the inference system adds dependencies
%between \Abs{this} and \emph{any} object name either in the formal parameters or 
%created by previous statements of the method.]

\paragraph{Recursive object structures.}
In \coreABS, like in any other object-oriented language, it is possible to
define circular object structures, such as an object storing a pointer to itself
in one of its fields.
%%%%%%%%%%%%%%%%%%%%%%%
%%%%%%%%%%%%%%%%%%%%%%% TODO
%%%%%%%%%%%%%%%%%%%%%%%
Currently, the contract inference system cannot deal with recursive structures, 
because the
semi-unification process associates each object with a finite tree structure.
In this way, it is not possible to capture circular definitions, such as the
recursive ones. 
Note that this restriction still allows recursive definition of classes.
We will investigate whether it is possible to extend the semi-unification 
process by associating 
\emph{regular terms}~\cite{tipoW-Coppo:TTPK-98}
to objects in the semi-unification process. These regular terms might be
derived during the inference by extending the {\coreABS}
code with annotations, as done for letting syntonisations on booleans.

\paragraph{Discussion.}
% three things to say
% 1. not that restrictive (we can check a industrial size case study)
% 2. many limitations are also present in other tools (DECO with local futures and one implementation)
% 3. many of our limitations are easy to fix, it just takes time to implement.
%   The two scientific challenge we really have are precise assignement and recursive structures.
The above restrictions do not severely restrict both programming and the precision of our analysis.
As we said, despite these  limitations, we were able to apply our tool to the 
industrial-sized case study FAS from SDL Fredhopper
and detect that it was dead-lock-free. It is also worth to observe that
most of our restrictions can be removed with a simple extension of the current implementation. The restriction that may be challenging to remove is the 
one about recursive object structures, which requires the extension of 
semi-unification to such structures.
We finally observe that other deadlock analysis tools have
restrictions similar to those discussed in this section. 
For instance, {\tt DECO} doesn't allow
futures to be passed around (i.e.~futures cannot be returned or put in an object's field)
and constraints interfaces to be implemented by at most one class~\cite{Antonio2013}.
Therefore, while {\tt DECO} supports the analysis in presence of field updates, 
our tool supports futures to be returned.

\section{Contracts and the contract inference system}
\label{sec.FJg-contracts}
%!TEX root = SoSyM.tex

%
The deadlock detection framework we present in this paper relies on 
abstract descriptions, 
called \emph{contracts}, that are extracted from programs by an inference system.
The syntax of these descriptions, which is defined in Figure~\ref{fig:synCon}, 
uses \emph{record names} $\X$, $\Y$, $\Z$, $\cdots$, and \emph{future names} $f$, $f'$,
$\cdots$.
\begin{figure*}[t]
  \centering
\mysyntax{\mathmode}{ 
% Simple records, used for communication between methods
  \simpleentry \frr [] \unit \bnfor \X \bnfor
    \mbox{$[\cog {:} \obj, \vect{x} {:} \vect{\frr}]$} \bnfor \mbox{$\fRec{\objA}{\frr}$} [future record]
\\
\\
% Contracts. Only references simple records
  \simpleentry \cntc [] \pinull %\bnfor \nu f\bnfor \nu c
    \bnfor \pinull.(\obj,\objA') 
    \bnfor \mbox{$\pinull$}.(\obj,\objA')^\aa
    \bnfor \C.\m \; \frr(\bar{\frr})\rightarrow \frr'
    \bnfor \C!\m \; \frr(\bar{\frr})\rightarrow\frr'
    \bnfor \mbox{$\C!\m \;  \frr(\bar{\frr})\rightarrow\frr' \seqpoint (\obj,\objA')$}   
    \qquad \qquad~ [contract]
    \oris  \mbox{$\C!\m \; \frr(\bar{\frr})\rightarrow\frr' \seqpoint (\obj,\objA')^{\aa}$} 
     \bnfor \cntc\fatsemi \cntc 
     \bnfor  \cntc + \cntc 
     %\bnfor \cntc \parop \cntc
     \bnfor \cntc \rfloor \cntc [{}]
\\
\\
% Records with pointers. Used inside a method. This syntax is more permissive than necessary, as we need pointers only at the first layer (not in depth)
%\simpleentry \frx []  \unit \bnfor \X \bnfor
%    \mbox{$[\cog {:} \obj, \vect{x} {:} \vect{\frx}]$} \bnfor
%    \mbox{$\fRec{\objA}{\frr}$} \bnfor f 
%    [extended future record]
\simpleentry \frx []  \frr\bnfor f [extended future record]
\\
\\
% record and contract in the future
	\simpleentry \frz []
    (\frr, \cntc) \bnfor (\frr,\pinull)^\checkmark
    [future reference values]
}
  \caption{Syntax of future records and contracts.}
\label{fig:synCon}
\end{figure*}
Future records $\frr$, which
encode the values of expressions in contracts, may be one of the
following: 
\begin{itemize}
\item[--]
a dummy value $\unit$ that models 
primitive types,
\item[--]
a record name $X$ that 
 represents a place-holder for a value and can be
instantiated by substitutions, 
\item[--]
$[\cog {:} \objA, \bar{x} {:} \bar{\frr}]$
that defines an object with its cog name $\objA$ and the values for fields and parameters of the object, 
\item[--]
and  
$\fRec{\objA}{\frr}$, which specifies that accessing
$\frr$ requires control of the cog $\objA$
(and that the control is to be released once the method has been evaluated).
The future record $\fRec{\objA}{\frr}$ is associated with method
invocations: $\objA$ is the cog of the object on which the method is invoked.
%the object set $\objA$ contains the possible receiver objects of the invoked method.
%
%A future record $\frr + \frs$ is associated to the conditional \name{if} 
%statement, and defines a future record that can be either of the form $\frr$ or of the 
%form $\frs$.
%
%\Cosimo{THIS STUFF IS NOT RELEVANT?
The name $\obj$ in $[\cog {:} \obj, \xbar {:} \bar{\frr}]$ and $\fRec{\obj}{\frr}$ 
will be called \emph{root} of the future record. % and is returned by the (partial)
% function $\rt{\cdot}$.
%A future record is \emph{linear} if the object names and the record names
%occur linearly.
%The function $\gr{ \cdot }$ returns the object and record names in the argument.
%% (we usually apply it
%%to $\Fr$ and $\frr$; let also $\rt{\objA[\fbar {:} \bar{\frr}]}=\objA$.
%%
%}
\end{itemize}

Contracts $\cntc$ collect the method 
invocations and the dependencies inside statements. 
In addition to $\pinull$, $\pinull.(\obj,\objA')$, and $\pinull.(\obj,\objA')^{\aa}$ that
respectively represent the empty behaviour, the dependencies due to a
\Abs{get} and an \Abs{await} operation, we have basic contracts that deal
with method invocations.
There are several possibilities:
\begin{itemize}
\item 
$\C.\m \, \frr(\bar{\frr})\rightarrow\frr'$
(resp.~$\C!\m \, \frr(\bar{\frr})\rightarrow\frr'$) 
specifies that 
the method {\m} of class {\C} is going to be invoked \emph{synchronously}
(resp.~\emph{asynchronously})
on an object $\frr$, with arguments
$\bar{\frr}$, and an object $\frr'$ will be returned;
\item
$\C!\m \, \frr(\bar{\frr})\rightarrow\frr'\seqpoint (\obj,\objA')$ indicates that 
the current method execution requires the termination of method 
$\C!\m$ running on an object of cog $\objA'$ in order to release the object of the cog $\obj$. 
This is the contract of an asynchronous method invocation followed by a \Abs{get} operation on 
the same future name.
\item
$\C!\m \, \frr(\bar{\frr})\rightarrow\frr'\seqpoint (\obj,\objA')^{\aa}$, 
indicating that the current method execution requires
the termination of method $\C.\m$ running on an object of cog $\objA'$ in order to 
progress. This is the contract of an asynchronous method invocation followed by
an \Abs{await} operation, and, possibly but not necessarily, by a \Abs{get}
operation.
In fact, a \Abs{get} operation on the same future name does not add any
dependency, since it is guaranteed to succeed because of the preceding
\Abs{await}.
\end{itemize}
The composite contracts $\cntc\fatsemi \cntc'$ and $\cntc + \cntc'$
define the abstract behaviour of sequential compositions and
conditionals, respectively.
The contract $\cntc \rfloor \cntc'$ require a 
separate discussion because it models parallelism, which is not explicit
in {\coreABS} syntax. We will discuss this issue later on.

%defines  of contracts.
% the abstract behaviour of (expressions with side-effects in) statements. 
%The contract  
%, and parallel $\cntc \parop \cntc'$.
% defines the abstract behaviour of conditionals.
%branching and is  associated to conditionals.

\begin{example}
As an example of contracts, let us discuss the terms:
\begin{enumerate}
\item[(a)] ${\tt C.m} [\cog {:} \obj_1, \x {:} [ \cog {:} \obj_2]]() \rightarrow[\cog {:} \obj_1', \x {:} [ \cog{:}\obj_2]] \, \fatsemi$\\
 ${\tt C.m} [\cog {:} \obj_3 , \x {:} [\cog {:} \obj_4 ]]() \rightarrow [\cog {:} \obj_3', \x  {:} [\cog {:} \obj_4]]$;
 
 \medskip
 
\item[(b)] 
${\tt C!m} [\cog {:} \obj_1, \x  {:} [ \cog {:} \obj_2]]() \rightarrow [\cog {:} \obj_1',
\x  {:} [\cog {:} \obj_2 ]] \seqpoint (\obj,\obj_1) \, 
\fatsemi \, {\tt C!m} [\cog {:} \obj_3, \x  {:} [ \cog {:} \obj_4 ]]() \rightarrow
[\cog {:} \obj_3', \x  {:} [ \cog {:} \obj_4 ]] \seqpoint (\obj,\obj_3)^\aa$. 
\end{enumerate}
The contract (a) defines a sequence of two synchronous invocations of method {\m} of class {\C}. We notice that the cog names $c_1'$ and $c_3'$ are free: this indicates
that {\C.\m} returns an object of a new cog.
%
% the future record of the first one is 
%$ [\cog{:} \obj_1, \x{:}[\cog {:}\obj_2] ]$, the future record of the second one is 
%$ [\cog {:} \obj_3, \x  {:} [ \cog {:} \obj_4]]$. This contract is enforcing 
%the constraint $\obj_1 = \obj_3 = c$ on cogs, where $c$ is callee's cog (this means 
%that the callee is accessing synchronously to objects belonging to its cog).
%
As we will see below, a {\coreABS} expression with  
this contract is {\tt x.m() ; y.m() ;}.
%, with {\tt x} and {\tt y} variables
%of class {\C}. 
%In turn, the analysis of (a) will add to its abstract model
%two disjoint sets of cog dependencies corresponding to those  of 
%${\tt C.m} [\cog {:} \obj_1, \x {:} [ \cog {:} \obj_2]]() \rightarrow[\cog {:} \obj_1', \x {:} [ \cog{:}\obj_2]]$ and of 
% ${\tt C.m} [\cog {:} \obj_3 , \x {:} [\cog {:} \obj_4 ]]() \rightarrow [\cog {:} \obj_3', \x  {:} [\cog {:} \obj_4]]$.

The contract (b) defines an asynchronous invocation of {\C.\m} followed by a \Abs{get} 
statement and an asynchronous one followed by an \Abs{await}.
% ,
%additionally,
%expresses that the value of the first invocation is required as well
%as the termination of the second invocation. 
The cog $c$ is the one of the caller.
A {\coreABS} expression retaining 
this contract is {\tt u = x!m() ; w = u.get ; v = y!m() ; await v? ;}.
%, with {\tt x} and {\tt y} variables
%of class {\C}.
\end{example}

\medskip

\ifIFM
\else

\medskip

The inference of contracts uses two additional syntactic categories: $\frx$ of future
record values and $\frz$ of typing values. The former one extends future records with
\emph{future names}, which are used to carry out the \emph{alias analysis}. 
In particular, every
local variable  of methods and every object field and parameter of future type is 
associated to a future name. Assignments between these terms, such as $x = y$, amounts 
to copying 
future names instead of the corresponding values ($x$ and $y$ become aliases). 
The category $\frz$ collects the typing values of future names, which are 
either $(\frr,\cntc)$, for \emph{unsynchronised futures}, or $(\frr,\pinull)^\checkmark$,
for \emph{synchronised ones}
(see the comments below).

\smallskip

The abstract behaviour of methods is defined by \emph{me\-thod contracts} 
$\mcontract{}{\frr(\bar{\frs} %_1, \cdots, \frs_n
)}{\pairl{\cntc}{\cntc'}}{\frr'}$,
where 
$\frr$ is the future record of the receiver of the method, $\bar{\frs}$ are the future records of the arguments, $\pairl{\cntc}{\cntc'}$ is the abstract behaviour
 of the body, where $\cntc$ is called 
\emph{synchronised contract} and  $\cntc'$ is called \emph{unsynchronised contract}, 
 and 
$\frr'$ is the future record of the returned object. 

Let us explain why method contracts  use pairs of contracts. In {\coreABS}, 
invocations in method bodies are of two types: (\emph{i})  \emph{synchronised}, that is
the asynchronous invocation has a subsequent \Abs{await} or \Abs{get} operation in the
method body and  (\emph{ii}) \emph{unsynchronised}, the asynchronous invocation 
has no corresponding  \Abs{await} or \Abs{get} in the same method body. (Synchronous
invocations can be regarded as asynchronous invocations followed by a \Abs{get}.)
For example, let
\begin{absexamplen}
x = u!m() ;
await x? ;
y = v!m() ;
\end{absexamplen}
be the main function of a program (the declarations are omitted).
In this statement, the invocation {\tt u!m()} is synchronised before the execution
of {\tt v!m()}, which is unsynchronised. {\coreABS} semantics  
tells us that the body
of {\tt u!m()} is \emph{performed before} the body of {\tt v!m()}. However, while
this ordering holds for the synchronised part of {\tt m}, it may not hold
for the unsynchronised part. In particular, the unsynchronised part of {\tt u!m()}
may run \emph{in parallel} with the body of {\tt v!m()}. For this reason, in this case,
our inference system returns the pair
\[
\pairl{\C\mbox{\tt !m}\ [\cog {:} c'](~) \rightarrow \unit \seqpoint (c,c')^\aa}{
\C\mbox{\tt !m}\ [\cog {:} c''](~) \rightarrow \unit}
\]
where $c$, $c'$ and $c''$ being the 
cog of the caller, of {\tt u} and {\tt v}, respectively. Letting 
$\C\mbox{\tt !m}\ [\cog {:} c'](~) \rightarrow \unit = \pairl{\cntc_u}{\cntc_u'}$ and
$\C\mbox{\tt !m}\ [\cog {:} c''](~) \rightarrow \unit = \pairl{\cntc_v}{\cntc_v'}$, one
has (see Sections~\ref{sec.contractanalysis} and~\ref{sec.mutanalysis})
\[
\begin{array}{l}
\pairl{\C\mbox{\tt !m}\ [\cog {:} c'](~) \rightarrow \unit \seqpoint (c,c')^\aa}{
\C\mbox{\tt !m}\ [\cog {:} c''](~) \rightarrow \unit}
\\
\quad = \quad \pairl{\cntc_u \seqpoint (c,c')^\aa}{\cntc_u' \rfloor (\cntc_v \fatsemi \cntc_v')}
\end{array}
\]
that adds the dependency $(c,c')^\aa$ to the synchronised contract of {\tt u!m()}
and makes the parallel (the $\rfloor$ operator) of  the unsynchronised contract of 
{\tt u!m()} and the contract of {\tt v!m()}.
Of course, in alternative to separating \emph{synchronised} and \emph{unsynchronised contracts},
one might collect all the dependencies in a unique contract.
This will imply that the dependencies in different configurations will be
gathered in the same set, thus significantly reducing the precision of the analyses
in Sections~\ref{sec.contractanalysis} and~\ref{sec.mutanalysis}.

The above discussion also highlights the need of contracts 
$\cntc \rfloor \cntc'$. In particular, this operator models \emph{ parallel
behaviours}, which is not a first class operator in {\coreABS}, while it is 
central in the semantics (the objects in the configurations). 
We illustrate the point with a statement similar to the above one, where we have
swapped the second and third instruction
\begin{absexamplen}
x = u!m() ;
y = v!m() ;
await x? ;
\end{absexamplen}
According to {\coreABS} semantics, it is possible that the bodies of
{\tt u!m()} and of {\tt v!m()} run in parallel by interleaving their executions. 
In fact, in this case, our inference system returns the pair of contracts (we keep
the same notations as before) 
\[
\begin{array}{rll}
\langle & \C\mbox{\tt !m}\ [\cog {:} c'](~) \rightarrow \unit \seqpoint (c,c')^\aa
\rfloor \C\mbox{\tt !m}\ [\cog {:} c''](~) \rightarrow \unit \; \fatcomma &
\\ 
& \C\mbox{\tt !m}\ [\cog {:} c''](~) \rightarrow \unit & \rangle
\end{array}
\]
which turns out to be equivalent to 
$$
\C\mbox{\tt !m1}\ [\cog {:} c'](~) \rightarrow \unit.(c,c')^\aa \rfloor 
\C\mbox{\tt !m2}\ [\cog {:} c'](~) \rightarrow \unit
$$
(see Sections~\ref{sec.contractanalysis}
and~\ref{sec.mutanalysis}).

The subterm $\frr(\bar{\frs})$ of the method contract 
%$\mcontract{\C.\m}{\frr(\bar{\frs} %_1, \cdots, \frs_n
%)}{\cntc}{\frr'}$
is called 
\emph{header};
$\frr'$ is called \emph{returned future record}. We assume that cog and record
names in the header occur linearly. Cog and record names in the header \emph{bind} the 
cog and record names in $\cntc$ and in $\frr'$.
The header and the returned
future record, written $\interface{\frr}{\bar{\frs}}{\frr'}$, are called \emph{contract signature}. In a method contract $\mcontract{}{\frr(\bar{\frs} %_1, \cdots, \frs_n
)}{\pairl{\cntc}{\cntc'}}{\frr'}$, cog and record names occurring in $\cntc$ or $\cntc'$ or 
$\frr'$ may be \emph{not bound} by header. These \emph{free names} correspond to 
\Abs{new cog} instructions and will be replaced by fresh cog names during the analysis.

\fi

\medskip

\begin{figure*}[t]
expressions and addresses
\myrules{\mathmode\compactmode}{
% Pure Expression. We leave out all the rules about data types, it is a little bit technical and annoying
\entry[{\scriptsize (T-Var)}]{\Gamma(x) = \frx}{\inferpe{\Gamma}{\obj}{x}{\frx}} \and
\entry[{\scriptsize (T-Fut)}]{\Gamma(f) = \frz}{\inferpe{\Gamma}{\obj}{f}{\frz}} 
\and
\entry[{\scriptsize (T-Field)}]{x\not\in\dom(\Gamma)\\\Gamma(\ethis.x)=\frr}{\inferpe{\Gamma}{\obj}{\x}{\frr}} \and
\entry[{\scriptsize (T-Value)}]{\inferpe{\Gamma}{\obj}{e}{f}\\\inferpe{\Gamma}{\obj}{f}{(\frr,\cntc)^{[\checkmark]}}}
  {\inferpe{\Gamma}{\obj}{e}{\frr}} 
\and
\entry[{\scriptsize (T-Val)}]{e \quad  \textit{primitive value or arithmetic-and-bool-exp}}{\inferpe{\Gamma}{\obj}{e}{\unit}}
\and
\entry[{\scriptsize (T-Pure)}]{\inferpe{\Gamma}{\obj}{e}{\frr}}{\inferee{\Gamma}{\obj}{e}{\frr}{\pinull}{\etrue}{\Gamma}}
}	
expressions with side effects
\myrules{\mathmode\compactmode}{
\entry[{\scriptsize (T-Get)}]{\inferpe{\Gamma}{\obj}{x}{f}\\
\inferpe{\Gamma}{\obj}{f}{(\frr, \cntc)} \\\X,\objb \fresh
\\
\Gamma'=\Gamma[f \mapsto (\frr, \pinull)^\checkmark]}
  {\inferee{\Gamma}{\obj}{x.\nget}{\X}{\cntc\seqpoint(\obj,\objb) \rfloor\contract{\Gamma'}}{\frr=\fRec{\objb}{\X}}{\Gamma'}} 
 \and
\entry[{\scriptsize (T-Get-tick)}]{\inferpe{\Gamma}{\obj}{x}{f}\\
\inferpe{\Gamma}{\obj}{f}{(\frr, \cntc)^\checkmark}\\\X,\objb \fresh}
  {\inferee{\Gamma}{\obj}{x.\nget}{\X}{\pinull}{\frr=\fRec{\objb}{\X}}{\Gamma}} \and
\entry[{\scriptsize (T-NewCog)}]{\inferpe{\Gamma}{\obj}{\bar{e}}{\bar{\frr}}\\\\
    \fields{\C} = \bar{T\ x} \\ \parameters{\C}= \bar{T'\ x'} \\  \bar{\X}, \obj' \fresh}
  {\inferee{\Gamma}{\obj}{\nnew\ \ncog\ \C(\bar{e})}{[\cog {:} \obj', \bar{x{:}\X}, \bar{x'{:}\frr}]}{\pinull}{\etrue}{\Gamma}} \and
\entry[{\scriptsize (T-New)}]{\inferpe{\Gamma}{\obj}{\bar{e}}{\bar{\frr}}\\\\
    \fields{\C} = \bar{T\ x} \\ \parameters{\C}= \bar{T'\ x'}\\ \bar{\X} \fresh}
  {\inferee{\Gamma}{\obj}{\nnew\ \C(\bar{e})}{[\cog {:} \obj, \bar{x{:}\X}, \bar{x'{:}\frr}]}{\pinull}{\etrue}{\Gamma}} \and
\entry[{\scriptsize (T-AInvk)}]{\inferpe{\Gamma}{\obj}{e}{\frr}\\\inferpe{\Gamma}{\obj}{\bar{e}}{\bar{\frs}}\\\\
    \fclass{\types{e}}=\C \\\fields{\C}\cup \parameters{\C} = \bar{T\ x} \\ \X, \overline{\X}, \objb, f \fresh}
  {\inferee{\Gamma}{\obj}{e!\m(\bar{e})}{f}{\pinull}{[\cog {:} \objb , \bar{x{:} \X}] = \frr\land\C.\m \preceq\frr(\bar{\frs}) \rightarrow X}
    {\Gamma[f \mapsto (\fRec{\objb}{\X},\,\C!\m~\frr(\bar{\frs}) \rightarrow \X)]}}
\and
\entry[{\scriptsize (T-SInvk)}]{\inferpe{\Gamma}{\obj}{e}{\frr}\\\inferpe{\Gamma}{\obj}{\bar{e}}{\bar{\frs}}\\\\
     \fclass{\types{e}}=\C\\ \fields{\C} \cup\parameters{\C} = \bar{T\ x}\\  \X,\bar{\X} \fresh}
  {\inferee{\Gamma}{\obj}{e.\m(\bar{e})}{\X}{\C.\m~\frr(\bar{\frs}) \rightarrow X \rfloor\contract{\Gamma}}
    {[\cog {:} \objb , \bar{x {:} X}] = \frr \land\C.\m \preceq\frr (\bar{\frs}) \rightarrow X}{\Gamma}}
}
\caption{Contract inference for expressions and expressions with side effects.}
\label{fig:inf:exp}
\end{figure*}

\begin{figure*}[t]
{%\footnotesize
statements
\myrules{\mathmode\compactmode}{
\simpleentry[T-Skip]{\inferst{\Gamma}{\obj}{\eskip}{\pinull}{\etrue}{\Gamma}} \and
\entry[{\scriptsize (T-Field-Record)}]{x\not\in\dom(\Gamma)\\\Gamma(\ethis.x) = \frr\\\\\inferee{\Gamma}{\obj}{z}{\frr'}{\cntc}{\mathcal U}{\Gamma'}}
  {\inferst{\Gamma}{\obj}{x=z}{\cntc}{\mathcal U\land\frr=\frr'}{\Gamma'}}
\and
\entry[{\scriptsize (T-Var-Record)}]{\Gamma(x) = \frr\\\\\inferee{\Gamma}{\obj}{z}{\frr'}{\cntc}{\mathcal U}{\Gamma'}}
  {\inferst{\Gamma}{\obj}{x=z}{\cntc}{\mathcal U}{\Gamma'[x\mapsto\frr']}}
\and
\entry[{\scriptsize (T-Var-Future)}]{\Gamma(x)= f\\\inferee{\Gamma}{\obj}{z}{f'}{\cntc}{\mathcal U}{\Gamma'}}
  {\inferst{\Gamma}{\obj}{x=z}{\cntc}{\mathcal U}{\Gamma'[x \mapsto f']}} 
\and
\entry[{\scriptsize (T-Var-FutRecord)}]{\Gamma(x)= f\\\inferee{\Gamma}{\obj}{z}{\frr}{\cntc}{\mathcal U}{\Gamma'}}
  {\inferst{\Gamma}{\obj}{x=z}{\cntc}{\mathcal U}{\Gamma'[f \mapsto (\frr,\pinull)]}} 
  \and
  \entry[{\scriptsize (T-Await)}]{\inferpe{\Gamma}{\obj}{e}{f} \\
\inferpe{\Gamma}{\obj}{f}{(\frr, \cntc)}
\\\X,\objb\fresh\\\Gamma'=\Gamma [f \mapsto (\frr,\pinull)^\checkmark]}
  {\inferst{\Gamma}{\obj}{\nwait\ e?}{\cntc \seqpoint (\obj,\objb)^{\aa} \rfloor\contract{\Gamma'}}{\frr=\fRec{\objb}{\X}}{\Gamma'}}
\and
\entry[{\scriptsize (T-Await-Tick)}]{\inferpe{\Gamma}{\obj}{e}{f}\\
\inferpe{\Gamma}{\obj}{f}{(\frr, \cntc)^\checkmark}
\\\X,\objb\fresh}
  {\inferst{\Gamma}{\obj}{\nwait\ e?}{\pinull}{\frr=\fRec{\objb}{X}}{\Gamma }}
\and
\entry[{\scriptsize (T-If)}]{\inferpe{\Gamma}{\obj}{e}{\Bool}\\
    \inferst{\Gamma}{\obj}{s_1}{\cntc_1}{\mathcal U_1}{\Gamma_1}\\\inferst{\Gamma}{\obj}{s_2}{\cntc_2}{\mathcal U_2}{\Gamma_2}\\
        \mathcal{U}=\Bigl(\bigland_{x\in\dom(\Gamma)}\hspace*{-1em}\Gamma_1(x)=\Gamma_2(x) \Bigr) \land
                    \Bigl(\bigland_{x\in{\tt Fut}(\Gamma)}\hspace*{-1em}\Gamma_1(\Gamma_1(x)) = \Gamma_2(\Gamma_2(x)) \Bigr)\\
        \Gamma' = \Gamma_1 + \Gamma_2 |_{ \{ f \; | \; f \notin \Gamma_2({\tt Fut}(\Gamma) )\}} 
        }
  {\inferst{\Gamma}{\obj}{\eif{e}{s_1}{s_2}}{\cntc_1+\cntc_2}{\mathcal U_1\land \mathcal U_2 \land \mathcal U}{\Gamma'}} \and
\entry[{\scriptsize (T-Seq)}]{\inferst{\Gamma}{\obj}{s_1}{\cntc_1}{\mathcal U_1}{\Gamma_1}\\\inferst{\Gamma_1}{\obj}{s_2}{\cntc_2}{\mathcal U_2}{\Gamma_2}}
  {\inferst{\Gamma}{\obj}{s_1;s_2}{\cntc_1\fatsemi\cntc_2}{\mathcal U_1\land \mathcal U_2}{\Gamma_2}}\and
\entry[{\scriptsize (T-Return)}]{\inferpe{\Gamma}{\obj}{e}{\frr}\\\Gamma(\destiny)=\frr'}
  {\inferst{\Gamma}{\obj}{\nreturn\ e}{\pinull}{\frr=\frr'}{\Gamma}}
}}
\caption{Contract inference for statements.}
\label{fig:inf:stmt}
\end{figure*}

\subsection{\em Inference of contracts}

Contracts are extracted from {\coreABS} programs by means of an inference algorithm.
Figures~\ref{fig:inf:exp}
and~\ref{fig.meth-class} illustrate the set of  rules. 
The following auxiliary operators are used:
\begin{itemize}
\item $\fields{\C}$ and $\parameters{\C}$ respectively
return the sequence of fields and parameters and their types of a class $\C$.
Sometime we write {\small$\fields{\C} = \bar{T} \, \bar{x}, \;  \bar{ \name{Fut} \arglang{T'} } \, \bar{x'}$} to separate fields with a non-future type by those with 
future types. Similarly for parameters;

\item $\types{e}$ returns the type of an expression $e$, which is either an interface (when $e$ is an object) or a data type;
\item $\fclass{I}$ returns the unique (see the restriction \emph{Interfaces} in 
Section~\ref{sec.restrictions}) class implementing $I$; and 
\item ${\it mname}(\bar{M})$ returns the sequence of method names in the sequence 
$\bar{M}$ of method declarations.
\end{itemize}

The inference algorithm  %of {\coreABS} 
uses constraints $\mathcal U$, which 
are defined by the following syntax
{\mysyntax{\mathmode}{
\simpleentry \mathcal U [] \etrue 
\bnfor c =c' \bnfor \frr=\frr'
\bnfor \frr(\bar{\frr})\rightarrow\frs \preceq \frr'(\bar{\frr}')\rightarrow\frs' 
[ ]
\oris
 \mathcal U \land \mathcal U
[ ]%[constraint]
}}
where $\etrue$ is the constraint that is always true;
 $\frr=\frr'$ is a classic unification constraint between terms;
 $\frr(\bar{\frr})\rightarrow\frs\preceq\frr'(\bar{\frr}')\rightarrow\frs'$ is a {\em semi-unification} constraint; the constraint  $\mathcal U\land\mathcal U'$ is the conjunction of 
$\mathcal U$ and $\mathcal U'$.
We use {\em semi-unification} constraints~\cite{Henglein:1993} to deal with method invocations: basically, in $\frr(\bar{\frr})\rightarrow\frs\preceq\frr'(\bar{\frr}')\rightarrow\frs'$,
 the left hand side of the constraint corresponds to the method's formal parameter, $\frr$ being the record of $\ethis$, $\vect{\frr}$ being the records of the parameters and $\frr'$ being the record of the returned value,
 while the right hand side corresponds to the actual parameters of the call, and the actual returned value.
The meaning of this constraint is that the actual parameters and returned value must match the specification given by the formal parameters, like in a standard unification:
 the necessity of semi-unification appears when we call several times the same method.
Indeed, there, unification would require that the actual parameters of the different calls must all have the same records, while with semi-unification all method calls are managed independently.

The judgments of the inference algorithm have a typing context $\Gamma$ 
mapping variables to extended future records, future names to future name
values and methods to 
their signatures. 
%\Cosimo{secondo me i metodi non sono mappati!}
They have the following form:
\begin{itemize}
\item[--] $\inferpe{\Gamma}{\obj}{e}{}{\frx}$
for pure expressions $e$ and $\inferpe{\Gamma}{\obj}{f}{}{\frz}$ for
 future names $f$, where 
$\obj$ is the cog name of the object executing the expression and $\frx$ and 
$\frz$ are their inferred values.
%Constraints and contracts are not generated at this stage.
\item[--]
$\inferee{\Gamma}{\obj}{z}{\frr}{\cntc}{\mathcal U}{\Gamma'}$ 
for expressions with side effects $z$, where $\obj$, and $\frx$ are as for pure expressions $e$, $\cntc$ is the contract for $z$ created by the inference rules,
$\mathcal U$ is the generated constraint, and $\Gamma'$ is the environment $\Gamma$ \emph{with updates} of variables and future names. We use the same 
judgment for pure expressions; in this case $\cntc = \pinull$, ${\mathcal U}
= {\tt true}$ and $\Gamma' = \Gamma$.
\item[--]
for statements $s$:
$\inferst{\Gamma}{\obj}{s}{\cntc}{\mathcal U}{\Gamma'}$ where
$\obj$, $\cntc$ and $\mathcal U$ are as before, and
$\Gamma'$ is the environment obtained after the execution of the statement.
The environment may change because of 
variable updates.
\end{itemize}

Since $\Gamma$ is a function, we use the standard predicates $x \in \dom(\Gamma)$ or 
$x \not \in \dom(\Gamma)$. Moreover, given a function $\Gamma$, we define $\Gamma[x \mapsto \frx]$
% and $\Gamma[\ethis.x \mapsto \frx]$
 to be the following function%s
\[
\Gamma[x \mapsto \frx](y) = \left\{ 
	\begin{array}{l@{\quad}l}
	\frx & {\rm if} \; y=x
	\\
	\Gamma(y) & {\rm otherwise}
	\end{array} \right.
\]
We also let $\Gamma |_{\{ x_1, \cdots , x_n \}}$ be the function
\[
{\footnotesize
\Gamma |_{\{ x_1, \cdots , x_n \}}(y) = \left\{ 
	\begin{array}{l@{\quad}l}
	\Gamma(y)  & {\rm if} \; y \in \{x_1, \cdots, x_n\}
	\\
	{\rm undefined} & {\rm otherwise}
	\end{array} \right.}
\]
Moreover, provided that % $\Gamma(x) = \Gamma'(x)$, for every 
$\dom(\Gamma)\cap \dom(\Gamma')=\emptyset$,
the environment
 $\Gamma + \Gamma'$ be defined as follows 
\[
(\Gamma+\Gamma')(x)  \; \eqdef \;
\left\{ \begin{array}{ll}
		\Gamma(x) & {\rm if} \; x \in \dom(\Gamma)
		\\
		\Gamma'(x) & {\rm if} \; x \in \dom(\Gamma')
		\end{array}
		\right.
\]

Finally, we write $\Gamma(\ethis.x) = \frx$ whenever $\Gamma(\ethis) = 
[\cog {:} \obj, x: \frx, \bar{x} : \bar{\frx'}]$ and we let 
\[
\begin{array}{rl}
{\tt Fut}(\Gamma) \eqdef & \{ x
\; | \; \Gamma(x) \; \text{is a future name} \}
\\
\\
\contract{\Gamma} \eqdef & \cntc_1 \rfloor \cdots 
\rfloor \cntc_n
\end{array}
\]
where $\{ \cntc_1, \cdots , \cntc_n\} =
\{ \cntc' \; | \; \mbox{there are} \; f,
\frr : \; \Gamma(f) = (\frr,\cntc') \}$.
%It is worth to observe that the semantics 
%of $\rfloor$ guarantees that $(\cntc \rfloor \cntc_1) \rfloor \cntc_2 =
%(\cntc \rfloor \cntc_2) \rfloor \cntc_1$, see Sections~\ref{sec.contractanalysis}
%and~\ref{sec.mutanalysis}.)

\medskip

The inference rules for expressions and future names are reported in 
Figure~\ref{fig:inf:exp}. They are straightforward, except for \rulename{T-Value} that performs the dereference of variables and return 
the future record stored in the future name of the variable. \rulename{T-Pure}
lifts the judgment of a pure expression to a judgment similar to those for
expressions with side-effects. This expedient allows us to simplify rules for
statements.

\smallskip

Figure~\ref{fig:inf:exp} also reports inference rules for expressions with 
side effects. Rule \rulename{T-Get} 
deals with the $x.$\Abs{get} synchronisation primitive and returns the contract
$\cntc \seqpoint (\obj,\objb)$ $\rfloor \contract{\Gamma}$, where $\cntc$ is 
 stored in the 
future name of $x$ and $(\obj,\objb)$ represents a dependency between the cog of the  object executing 
the expression and the root of the expression. The constraint $\frr=\fRec{\obj'}{X}$ 
is used to extract the root $c'$ of $\frr$. The contract $\cntc$ may have two shapes: either 
(i) $\cntc = \C!\m \, \frr (\bar{\frs}) \rightarrow \frr'$ or (ii) $\cntc = \pinull$.
%In case (ii), 
%$\cntc \seqpoint (\obj,\objb)$ is equal to $(\obj,\objb)$. 
The subterm $\contract{\Gamma}$ lets us collect 
all the contracts in $\Gamma$ 
\emph{that are 
stored in future names that are not check-marked}. In fact, these contracts
correspond to previous asynchronous invocations without any corresponding 
synchronisation 
(\Abs{get} or \Abs{await} operation) in the body. The evaluations of these invocations
may interleave with the evaluation of the expression $x.\get$. For this reason,
the intended meaning of $\contract{\Gamma}$ is that the
dependencies generated by the invocations must be collected \emph{together} with
those generated by $\cntc \seqpoint (\obj,\objb)$.
We also observe that the rule updates the 
environment by check-marking the value of the future name of $x$ and by replacing the
contract with $\pinull$ (because the synchronisation has been already performed). 
This allows subsequent
\Abs{get} (and \Abs{await}) operations on the same future name not to modify the contract
(in fact, in this case they are operationally equivalent to the \Abs{skip} statement)
-- see \rulename{T-Get-Tick}. 

Rule \rulename{T-NewCog} returns a record with a new cog name. This is 
in contrast with \rulename{T-New}, where the cog of the returned record is the
\emph{same} of the  object executing 
the expression~\footnote{\label{footnote.new}It is worth to recall that, in {\coreABS}, 
the creation of an object, either with a 
\Abs{new} or with a \Abs{new cog}, amounts to executing the method \name{init} of 
the corresponding 
class, whenever defined (the \Abs{new} performs a synchronous invocation, 
the \Abs{new cog} performs an asynchronous one). 
In turn, the termination of
 \name{init} triggers the execution of the method
\name{run}, if 
present. The method \name{run} is asynchronously invoked when \name{init}
is absent.
Since \name{init} may be regarded as a method in {\coreABS}, the inference
system in our tool explicitly introduces a synchronous invocation to \name{init}
in case of \Abs{new} and an asynchronous one in case of \Abs{new cog}. However,
for simplicity, we overlook this (simple) issue in the rules \rulename{T-New} and 
\rulename{T-NewCog},
acting as if \name{init} and \name{run} are always absent.}.

Rule \rulename{T-AInvk} derives contracts for asynchronous invocations. Since the
dependencies created by these invocations influence the 
dependencies of the synchronised contract only if a subsequent 
\Abs{get} or \Abs{await} operation is performed, the rule stores the invocation into
a fresh future name of the environment and returns the contract $\pinull$. 
This models {\coreABS} semantics that lets asynchronous invocations be synchronised 
by explicitly getting or awaiting on the corresponding future variable,
see rules \rulename{T-Get} and \rulename{T-Await}. The future name storing the invocation
is returned by the judgment. 
On the contrary, in rule \rulename{T-SInvk}, which deals with synchronous invocations,
the judgement returns a contract that is the invocation (because the corresponding 
dependencies must be added to the current ones) in parallel with the unsynchronised
asynchronous invocations stored in $\Gamma$. %$\contract{\Gamma}$. 

\smallskip

The inference rules for statements are collected in Figure~\ref{fig:inf:stmt}.
The first three rules define the inference of contracts for assignment. There are
two types of assignments: those updating fields and parameters of the \Abs{this}
object and the other ones. For every type, we need to address the cases of updates with values
that are expressions (with side effects) (rules \rulename{T-Field-Record} and 
\rulename{T-Var-Record}),
or future names (rule \rulename{T-Var-Future}). 
Rules for fields and parameters updates enforce that their future records are unchanging, as
discussed in Section~\ref{sec.restrictions}. 
Rule % \rulename{T-Field-Future} and 
\rulename{T-Var-Future}, 
 define the management 
of aliases: future variables are always updated with future names and never with 
future names' values.

Rule \rulename{T-Await} and \rulename{T-AwaitTick} deal with the \Abs{await}
synchronisation when applied to a simple future lookup $x${\tt ?}. They are
similar to the rules  \rulename{T-Get} and \rulename{T-Get-Tick}. 
%\Cosimo{Da rimuovere: In order to correctly associate dependencies
%with each synchronisation, we assume statements of the form
%$\await{x_1 \wedge x_2}$ to be decomposed into $\await{x_1} ; \await{x_2}$.}
% 
%Rule \rulename{T-Await-Bool} deals with \Abs{await} statements on booleans. As
%discussed in Section~\ref{sec.restrictions}, the current release of the 
%inference algorithm
% requires a programmer's notation $[x]$ for these statements. 

Rule \rulename{T-If} defines contracts for conditionals. In this case we collect the
contracts $\cntc_1$ and $\cntc_2$ of the two bran\-ches, with the intended meaning that 
the dependencies defined by $\cntc_1$ and $\cntc_2$ are always kept separated.
As regards the environments, the rule constraints the two environments $\Gamma_1$ and
$\Gamma_2$ produced by typing of the two branches to \emph{be the same} on variables in
$\dom(\Gamma)$ \emph{and} on the values of future names bound to variables
in ${\tt Fut}(\Gamma)$. However, the two branches may have different unsynchronised
invocations that are not bound to any variable. The environment
$
\Gamma_1 + \Gamma_2 |_{ \{ f \; | \; f \notin \Gamma_2({\tt Fut}(\Gamma) )\}}
$
allows us to collect all them.

Rule \rulename{T-Seq} defines the sequential composition of contracts. 
Rule \rulename{Return} constrains the record of {\tt destiny}, which is an identifier
introduced by \rulename{T-Method}, shown in Figure~\ref{fig.meth-class}, for 
storing the return record.

\smallskip

The rules for method and class declarations are defined in Figure~\ref{fig.meth-class}.
Rule \rulename{T-Method}  derives the method contract of 
$\T\;\m\;(\Tbar\; \xbar)\{ \bar{T'} \, \bar{u} ; s\}$ by typing $s$ in an environment 
extended with \name{this}, \Abs{destiny} (that will be set by \Abs{return} statements,
see \rulename{T-Return}),
the arguments $\xbar$, and the local variables $\bar{u}$. In order to deal with alias 
analysis of future variables, we separate fields, parameters, arguments and local
variables with future types from the other ones. In particular, we associate 
future names to the former ones and bind future names to record variables. As discussed
above, the abstract behaviour of the method body is a pair of contracts, which is
$\pairl{\cntc}{ \contract{\Gamma''}}$ for \rulename{T-Method}. 
This
term $\contract{\Gamma''}$ collects all the contracts in $\Gamma''$ 
\emph{that are 
stored in future names that are not check-marked}. In fact, these contracts
correspond to asynchronous invocations without any synchronisation 
(\Abs{get} or \Abs{await} operation) in the body. These invocations will be 
evaluated \emph{after} the termination of the body -- they are the \emph{unsynchronised  contract}.

The rule \rulename{T-Class} yields an \emph{abstract class table} that associates a 
method contract with every method name. It is this abstract class table that is 
used by our analysers in Sections~\ref{sec.contractanalysis} and~\ref{sec.mutanalysis}. 
The rule \rulename{T-Program} derives the contract of a {\coreABS} program by 
typing the main function in the same way as it was a body of a method.

\medskip

\begin{figure*}[t]
\myrules{\mathmode\compactmode}{
\entry[{\scriptsize (T-Method)}]{
  \fields{\C}\cup\parameters{\C} = \bar{T_f\ x}\ \bar{\name{Fut}\arglang{T_f'}\ x'}\\
  \obj, \bar{X}, \bar{X'}, \bar{Y}, \bar{Y'}, \bar{W}, \bar{W'}, \bar{f}, \bar{f'}, \bar{f''}, Z \ {\it fresh}\\\\
  \Gamma'=\bar{y {:} Y} + \bar{y' {:} f'} + \bar{w {:} W} + \bar{w' {:} f''}
    + \bar{f {:} (\bar{X'} , \pinull)} + \bar{f' {:} (\bar{Y'} , \pinull)} + \bar{f'' {:} (\bar{W'} , \pinull)}\\
  \inferst{\Gamma + \this: [\cog {:}\obj , \; \bar{x{:}X}, \bar{x{:}f}] + \Gamma' + \destiny : Z}{\obj}{s}{\cntc}{\mathcal U}{\Gamma''}\\
  \text{$\bar{T},\ \bar{T_f},\ \bar{T_l}$ are not future types}}
{\C,\Gamma \vdash {\T\;\m\;(\bar{T\ y}, \bar{\name{Fut} \arglang{T'}\ y'})\{\bar{T_l\ w}; \, \bar{\name{Fut} \arglang{T_l'}\ w'}; \; s\} } \;:\quad 
  \inferrule{}{[\cog : \obj , \bar{x{:}X}, \, \bar{x'{:}X'}](\bar{Y}, \bar{Y'})\{ \pairl{\cntc}{ \contract{\Gamma''}} \} \;  Z\\\\
       \rhd\; {\mathcal U} \land[\cog : \obj , \bar{x{:}X}, \, \bar{x'{:}X'}](\bar{Y}, \bar{Y'}) \rightarrow Z = \C.\m }
}\and
\entry[{\scriptsize (T-Class)}]{\inferpp{\C,\Gamma}{\bar{M}}{\bar{\mcntc}}{\bar{\mathcal{U}}}}
  {\inferpp{\Gamma}{\class\; \C (\bar{T\ x}) \; \{\bar{T'\ x'};\quad\bar{M}\}}{\bar{\C.{\it mname}(M)\mapsto\mcntc}}{\bar{\mathcal{U}}}}\and
\entry[{\scriptsize (T-Program)}]{\inferpp{\Gamma}{\bar C}{\bar S}{\bar{\mathcal U}}\\\bar{X}, \bar{X'}, \bar{f} \ {\it fresh}\\
    \inferst{\Gamma  + \bar{x{:}X} + \bar{x'{:}f} + \bar{f{:}(X', \pinull)}}{\rm start}{s}{\cntc}{\mathcal U}{\Gamma'}\\
    \text{$\bar T$ are not future types}}
  {\inferpp{\Gamma}{\bar I\ \bar C\ \{\bar{T\ x};\;\bar{\name{Fut}\arglang{T'}\ x'} ;\; s\}}{\bar S,\pairl{\cntc}{\contract{\Gamma'}}}{\mathcal{U} \land \bar{\mathcal U}}}
}
 \caption{Contract rules of method and class declarations and programs.}\label{fig.meth-class}
\end{figure*}

The contract class tables of the classes in a program derived by the rule 
\rulename{Class}, will be
noted $\cct$. We will address the contract of {\tt m} of class {\tt C} by $\cct(\C.\m)$.
 In the following, we assume that every {\coreABS} program is a triple
$(\ct , \{ \vect{T \ x \sseq} s \} , \cct)$, where $\ct$ is the class table, 
$\{ \vect{T \ x \sseq} s \}$ is the main function, and $\cct$ is its contract class 
table. By rule \rulename{Program}, analysing (the deadlock freedom of) a program, 
amounts to verifying the contract of the main function with a record for $\ethis$
which associates a special cog name called ${\rm start}$ to the $\cog$ field 
(${\rm start}$ is intended to be the cog name of the object ${\it start}$).

\begin{example}
The methods of \Abs{Math} in Figure~\ref{fig.Math}
  have the following contracts, once the constraints are solved (we
  always simplify $\cntc \fatsemi \pinull$ into $\cntc$):
  \begin{itemize}
  \item \name{fact\_g} has contract 
  {\footnotesize\[
      \mcontract{}{[\cog {:} \obj](\unit)}{ 
      \pairl{\pinull +
        \mbox{\Abs{Math\!}}\name{fact\_g}\; [\cog {:}
        \obj](\unit)\rightarrow \unit \seqpoint
        (\obj,\obj)}{\pinull}} {\unit}.\]}
    The name $\obj$ in the header refers to the cog name associated
    with \this{} in the code, and binds the occurrences of $\obj$ in
    the body.  The contract body has a recursive invocation to
    \name{fact\_g}, which is performed on an object in the same cog
    $\obj$ and followed by a \Abs{get} operation.  This operation
    introduces a dependency  $(\obj,\obj)$. We observe that, if we
    replace the statement \Abs{Fut<Int> x = this\!fact\_g(n-1)} in
    \name{fact\_g} with \Abs{Math z = new Math() ;} \Abs{Fut<Int> x =
      z\!fact\_g(n-1)}, we obtain the same contract as above because
    the new object is in the same cog as \Abs{this}.

\medskip

  \item \name{fact\_ag} has contract
    {\footnotesize$$\mcontract{}{[\cog {:} \obj](\unit)}{
    \pairl{\pinull +
        \mbox{\Abs{Math\!}}\name{fact\_ag}\; [\cog {:}
        \obj](\unit)\rightarrow \unit \seqpoint
        (\obj,\obj)^\aa}{\pinull}}{\unit}.$$} In this case, the presence of an
    \Abs{await} statement in the method body produces a dependency
   $(\obj,\obj)^\aa$. The subsequent \Abs{get} operation does
    not introduce any dependency  because the future name has a 
    check-marked value in the environment. In fact, in
    this case, the success of \Abs{get} is guaranteed, provided the
    success of the \Abs{await} synchronisation.
    
\medskip

  \item \name{fact\_nc} has contract
    {\footnotesize$$\mcontract{}{[\cog {:} \obj](\unit)}{
    \pairl{\pinull +
        \mbox{\Abs{Math\!}}\name{fact\_nc}\; [\cog {:}
        \obj'](\unit)\rightarrow \unit \seqpoint
        (\obj,\obj')}{\pinull}}{\unit}.$$} This method contract differs from the
    previous ones in that the receiver of the recursive invocation is
    a free name (i.e., it is not bound by $\obj$ in the header).  This
    because the recursive invocation is performed on an object of a
    new cog (which is therefore different from $\obj$). As a
    consequence, the dependency  added by the \Abs{get} relates
    the cog $\obj$ of \Abs{this} with the new cog $\obj'$.
  \end{itemize}
\end{example}

\begin{example}%[The method contracts of class {\rm \Abs{CpxSched}}]
Figure~\ref{fig.cpxsched} displays the contracts of the methods of class 
{\tt CpxSched} in Figure~\ref{fig:exampleC}.
\begin{figure*}
  \begin{itemize}
  \item method \Abs{m1} has contract
    {\footnotesize$$
\begin{array}{l}
      \mcontract{}{[\cog {:} c, \mbox{\Abs{x}}:[\cog {:} c', \mbox{\Abs{x}}:\X]]([\cog {:} c'', \mbox{\Abs{x}}:\Y])}{
    \pairl{\pinull}{\cntc}}{\fRec{c'}{\unit}}.
\\
\\
\; {\rm where} \quad 
\cntc = \mbox{\Abs{CpxSched\!}}\name{m2}\; [\cog {:}c'', \mbox{\Abs{x}}:\Y]([\cog {:} c', \mbox{\Abs{x}}:\X])\rightarrow \fRec{c''}{\unit}\rfloor \mbox{\Abs{CpxSched\!}}\name{m2}\; [\cog {:}c', \mbox{\Abs{x}}:\X]([\cog {:} c'', \mbox{\Abs{x}}:\Y])\rightarrow \fRec{c'}{\unit}
\end{array}
$$}
  \item method \Abs{m2} has contract
    {\footnotesize$$
      \begin{array}{l}
      \mcontract{}{[\cog {:} c, \mbox{\Abs{x}}:\X]([\cog {:} c', \mbox{\Abs{x}}:\Y])}{
    \pairl{ \mbox{\Abs{CpxSched\!}}\name{m3}\; [\cog {:}c', \mbox{\Abs{x}}:\Y](\unit)\rightarrow \unit \seqpoint (c,c')}{\pinull}}{\unit}.
\end{array}$$}
  \item method \Abs{m3} has contract
        {\footnotesize$$\mcontract{}{[\cog {:} c, \mbox{\Abs{x}}:\X](\,)}{\pairl{\pinull}{\pinull}}{\unit}.$$}
  \end{itemize}
\caption{\label{fig.cpxsched} Contracts of {\tt CpxSched}.}
\end{figure*}

According to the contract of the main function, the two invocations of \Abs{m2} are
second arguments of $\rfloor$ operators. This will give rise, in the analysis 
of contracts, to the union of the corresponding cog relations.  
\end{example}

We notice that the inference system of contracts discussed in this section is
modular because, when programs are organised in different modules, it partially supports the separate contract inference of modules with a 
well-founded ordering relation (for example, if there are two modules, classes in the 
second module use definitions or methods in the first one, but not conversely).
In this case, if a module {\tt B} includes a module
{\tt A} then a patch to a class of {\tt B} amounts to inferring contracts for {\tt B} only. 
On the contrary, a patch to a class of {\tt A} may also require a new contract inference
of {\tt B}.

\subsection{Correctness results}

In our system, the ill-typed programs are those manifesting a failure of
the semiunification process, which does not address misbehaviours. In particular,
a program may be well-typed and still manifest a deadlock. In fact, in systems
with \emph{behavioural types}, one usually demonstrates that
\begin{enumerate}
\item
in a well-typed program, 
every configuration $\nt{cn}$ has a behavioural type, let us call it ${\tt bt}(\nt{cn})$;
\item
if  $\nt{cn}\rightarrow \nt{cn}'$ then there is a relationship between
${\tt bt}(\nt{cn})$ and ${\tt bt}(\nt{cn}')$;
\item
the relationship in 2 preserves a given property (in our case, deadlock-freedom).
\end{enumerate}

Item 1, in the context of the inference system of this section, means that the 
program has a contract class table. Its proof needs a contract system for configurations, which we have defined in Appendix~\ref{sec:subjred}. The theorem
corresponding to this item is Theorem~\ref{th:wt_init}.

Item 2 requires the definition of a relation between contracts, called \emph{later stage
relation} in Appendix~\ref{sec:subjred}.
This later stage relation is a syntactic relationship between contracts
whose basic law is that a method invocation is larger than the 
instantiation of its method contract (the other laws, except $\pinull \trianglelefteq
\cntc$ and $\cntc_i \trianglelefteq \cntc_1 + \cntc_2$, are congruence laws). 
%We observe that the later stage relation uses a substitution process that \emph{also performs a 
%pattern matching operation} -- therefore it is partial because the pattern matching may fail. In particular, $\subst{\frs}{\frr}$ (i) extracts the cog names
%and terms $\frs'$ in $\frs$ that corresponds to occurrences of cog names and record
%variables in $\frr$ and (ii) returns the corresponding substitution.

The statement
that relates the later stage relationship to {\coreABS} reduction is Theorem~\ref{th_subjred1}. It is worth to observe that all the theoretical development up-to this point are useless if the 
later stage relation conveyed no relevant property. 
This is the purpose of item 3, which requires the definition of \emph{contract models} and the proof that deadlock-freedom is preserved by the models of contracts in 
later stage relation. The reader can find the proofs of these statements in the
Appendices~\ref{sec:analysis-proof} and~\ref{sec:mutanalysis-proof} (they correspond
to the two analysis techniques that we study).

%%% Local Variables: 
%%% mode: latex
%%% TeX-master: "subm-SoSyM"
%%% End: 

\section{The fix-point analysis of contracts}
\label{sec.contractanalysis}
%!TEX root = SoSyM.tex

The first algorithm we define to analyse contracts uses the standard Knaster-Tarski fixpoint technique. 
We first give an informal introduction of the notion used in the analysis, and start to formally define our algorithm in Subsection~\ref{sec:lam-lamop}
(a simplified version of the algorithm may be found in~\cite{GiachinoL11}, see also Section~\ref{sec:relatedworks}). 

\medskip

Based on a contract class table and a main contract (both produced by the inference system in Section~\ref{sec.FJg-contracts}),
 our fixpoint algorithm generates models that encode the dependencies between cogs that may occur during the program's execution.
These models, called {\em lams} (an acronym for deadLock Analysis Models~\cite{GiachinoL13,GL2014} are sets of relations between cog names, each relation representing a possible configuration of  program's execution. Consider for instance the  main function:
\begin{absexamplen}
{
  I x ; I y ; Fut<Unit> f ;
  x = new cog C() ;
  y = new cog C() ;
  f = x!m() ;
  await f? ;
  f = y!m() ;
  await f? ;  
}
\end{absexamplen}
In this case, the configurations of the program may be represented by two relations: one containing a dependency between the cog name ${\it start}$ and the cog name of $x$
and the other containing a dependency between ${\it start}$ and the cog name of $y$.
This would be represented by the following lam (where $c_x$ and $c_y$ respectively being the cog names of $x$ and $y$):
$$ \slam{(c,c_x)^\aa}, \; \slam{(c,c_y)^\aa} $$
(in order to ease the parsing of the formula, we are representing relations with 
the notation $\slam{\cdot }$ and we have removed the outermost curly brackets).
Our algorithm, being a fixpoint analysis, returns {\em the} lam of a program 
by computing a sequence of approximants. In particular, the algorithm 
performs the following steps:
\begin{enumerate}
\item
compute a new approximant of the lam of every method 
using the abstract class table and the previously computed lams;
\item
reiterate step 1 till a fixed approximant -- say $n$ (if a fixpoint is found before,
go to 4);
\item
when the $n$-th approximant is computed then \emph{saturate}, i.e.~compute the next
approximants by reusing the same cog names (then a fixpoint is eventually found);
\item
replace the method invocations in the main contract with the corresponding values;
  and 
\item
analyse the result of 4 by looking for a circular dependency in one of the relations of the computed lam
   (in such case, a possible deadlock is detected).
\end{enumerate}
The critical issue of our algorithm is the creation of fresh cog names at each step 1,
 because of free names in method contracts (that correspond to new cogs created during 
  method's execution). 
%
\iffalse
The fixpoint algorithm takes as input an abstract class table
and a main contract (both produced by the inference system in 
Section~\ref{sec.FJg-contracts}.
%  These (abstract)
% descriptions are the output of our inference algorithm.
%
Then it (1) computes the fixpoint of the abstract class table using 
the standard  technique,  (2) replaces
the method invocations in the main contract with the corresponding value, and
(3) analyses the result of (2).

The fixpoint algorithm uses models, called \emph{lam} (an acronym for deadLock Analysis Models~\cite{GiachinoL13,GL2014}), where elements are 
relations on cog names.
%
For example, denoting sets $\{b_1, \cdots , b_\ell\}$ as $\slam{b_1, \cdots , b_\ell}$,
we have that:
\begin{itemize}
\item
$\slam{ (c_1,c_2) }, \slam{(c_1,c_2), (c_2,c_3)}$ is a two-elements lam where one 
element contains 
the relation $\{ (c_1,c_2) \}$ and the other one contains $\{ (c_1,c_2), (c_2,c_3)
\}$;
\item
$\slam{(c_1,c_2)^\aa}$ is a singleton lam containing the relation 
$\{ (c_1,c_2)^\aa \}$.
\end{itemize}

The critical issue of our algorithm is the creation of fresh names at each step, 
technically speaking, at \emph{every approximant}, because of free names in method
contracts (that correspond to \Abs{new cog}s). 
%
\fi
%
\begin{figure*}[t]
\[
\begin{array}{l@{\qquad}rcl}
  \mbox{[extension]}  & \lafsa \addition (\obj,\obj')^{[\aa]} &\eqdef &\{ L \cup
    \{(\obj,\obj')^{[\aa]}\} \; | \; L \in \lafsa\}. \\
      \mbox{[parallel]} & \lafsa \parop \lafsa' &\eqdef& \{ L \cup L' \; | \;
     L \in \lafsa \; {\rm and} \; L' \in \lafsa' \}
\\
\\
\mbox{[extension  (on pairs of lams)]}&%\mbox{ Let also } 
    \pairl{\lafsa}{ \lafsa'} \addition (\obj,\obj')^{[\aa]} &\eqdef&
    \pairl{\lafsa\addition (\obj,\obj')^{[\aa]}}{ \lafsa'}
    \\ 
%   \mbox{[parallel (on pairs of lams)]} & \pairl{\lafsa_1 }{\lafsa_1'} \parop \pairl{\lafsa_2}{\lafsa_2' } &\eqdef& \pairl{\pinull}{(\lafsa_1 \cup \lafsa_1') \parop (\lafsa_2 \cup \lafsa_2')}
   \\
   %\pairl{\lafsa_1 \parop \lafsa_2}{\lafsa_1' \parop \lafsa_2'} \\
     \mbox{[parallel (on pairs of lams)]} & \pairl{\lafsa_1 }{\lafsa_1'} \rfloor \pairl{\lafsa_2}{\lafsa_2' } &\eqdef& \pairl{(\lafsa_1 \cup \lafsa_1') \parop (\lafsa_2 \cup \lafsa_2')}{\pinull} 
     \\
     \\ 
   \mbox{[sequence (on pairs of lams)]} &
     \pairl{\lafsa_1 }{\lafsa_1'} \fatsemi \pairl{\lafsa_2}{\lafsa_2' }
     & \eqdef &\left\{ \begin{array}{l@{\qquad}l} \pairl{\lafsa_1 }{
           \lafsa_1' \parop \lafsa_2'} & \mbox{if} \; \lafsa_2 =
        \pinull
         \\
         \\
         \pairl{\lafsa_1 \cup (\lafsa_2 \parop \lafsa_1')}{\lafsa_1'
           \parop \lafsa_2'}	& \mbox{otherwise}
       \end{array} \right.
       \\
       \\
   \mbox{[plus (on pairs of lams)]} & \pairl{\lafsa_1 }{\lafsa_1'} +
     \pairl{\lafsa_2}{\lafsa_2' } & \eqdef & \pairl{\lafsa_1 \cup
       \lafsa_2}{\lafsa_1' \cup \lafsa_2'}.
  \end{array}
\]
  \caption{Lam operations.}
\label{fig:lamOp}
\end{figure*}
For example, consider the contract of \name{Math.fact\_nc}
that has been derived in Section~\ref{sec.FJg-contracts}
$$
\begin{array}{l}
{\tt Math.fact\_nc} \, [\cog {:} \obj](\unit) \{
\\
\qquad 
\pairl{\pinull + \mbox{\Abs{Math\!}}\name{fact\_nc}\; [\cog {:} \obj'](\unit)\rightarrow \unit \seqpoint (\obj,\obj')}{\pinull} \quad \} \unit
\end{array}
$$
According to our definitions, the cog name $c'$ is free. In this case, our fixpoint algorithm
will produce the 
following sequence of lams when computing the
model of ${\tt Math.fact\_nc}[\cog {:} c_0](\unit)$: 
\[
\begin{array}{l@{\qquad}l}
{\it approximant \; 0}: & \pairl{\slam{\varnothing}}{\pinull}
\\
{\it approximant \; 1}: & \pairl{\slam{(c_0,c_1)}}{\pinull}
\\
{\it approximant \; 2}: & \pairl{\slam{(c_0,c_1), (c_1,c_2)}}{\pinull}
\\
\cdots
\\
{\it approximant \; n}: & \pairl{\slam{(c_0,c_1), (c_1,c_2), \cdots (c_{n-1},c_{n})}}{\pinull}
\\
\cdots
\end{array}
\]
While every lam in the above sequence is strictly larger than the previous one, an upper
bound element cannot be obtained by iterating the process.
Technically, the lam model \emph{is not a complete partial 
order} (the ascending chains of lams may have infinite length and no upper bound).
%and it may be the case that the Knaster-Tarski technique does not yield a 
%fixpoint in finite time.  For example, this is the case of 

In order to circumvent this issue and to get a decision on deadlock-freedom in a
finite number of steps, we use 
%another usual method: 
%\begin{enumerate}
%\item[(1)] run the Knaster-Tarski technique
%up-to a \emph{fixed approximant}, let us say $n$, and 
%\item[(2)] resort to 
a 
\emph{saturation argument}. If the $n$-th approximant is not a fixpoint, then
the {$(n+1)$-th} approximant is computed by
\emph{reusing the same cog names used by the 
$n$-th approximant} (no additional cog name is created anymore). Similarly for
the {$(n+2)$-th} approximant till a fixpoint is reached (by straightforward cardinality arguments, the fixpoint does exist, in this
case). This fixpoint is called \emph{the saturated state}. 
%\end{enumerate}

For example, for ${\tt Math.fact\_nc}[\cog {:} c_0](\unit)$, the $n$-th approximant returns the pairs of lams
\[
\pairl{\slam{(c_0,c_1) , \cdots , (c_{n-1},c_n)}}{\pinull} \; .
\]
Saturating at this stage yields the lam
\[
\pairl{\slam{(c_0,c_1) , \cdots , (c_{n-1},c_n), (c_1,c_1)}}{\pinull}
\]
that contains a circular dependency -- the pair $(c_1,c_1)$ -- revealing a 
potential deadlock in the corresponding program. Actually, in this 
case, this circularity is a \emph{false positive} that is introduced
by the (over)approximation: the original code
never manifests a deadlock.

Note finally that a lam is the result of the analysis of one contract.
Hence, to match the structures that are generated during the type inference, our analysis uses three extensions of lams:
 (i) a pair of lams $\pairl{\lafsa}{ \lafsa'}$ for analysing pairs of contracts $\pairl{\cntc}{ \cntc'}$;
 (ii) {\em parameterised} pair of lams $\lambda\vect c.\pairl{\lafsa}{ \lafsa'}$ for analysing methods:
  here, $\vect c$ are the cog names in the header of the method (the $\ethis$ object and the formal parameters), and
  $\pairl{\lafsa}{ \lafsa'}$ is the result of the analysis of the contract pair typing the method;
 and (iii) lam tables $(\cdots,\lambda \vect{c_i}. \pairl{\lafsa_i}{\lafsa_i'}, \cdots )$
 that maps each method in the program to its current approximant.
We observe that $\lambda\vect c.\pairl{\lafsa}{ \lafsa'}$ is $\pairl{\lafsa}{ \lafsa'}$
whenever $\vect{c}$ is empty.

\subsection{\em Lams and lam operations}
\label{sec:lam-lamop}

The following definition formally introduce the notion of lam.

\begin{definition}
A relation on cog names is a set of pairs either of the form $(c_1, c_2)$ or
of the form $(c_1, c_2)^\aa$, generically represented as 
$(c_1,c_2)^{[\aa]}$. We denote such relation by 
$\slam{(c_{i_0},c_{i_1})^{[\aa]} , \cdots , (c_{i_{n-1}},c_{i_n})^{[\aa]}}$.
A \emph{lam}, ranged over $\lafsa$, $\lafsa'$, $\cdots$, is a set of relations on cog 
names. Let $\pinull$ be the lam $\slam{\varnothing}$ and let $\gr{\lafsa}$ be the 
cog names occurring in $\lafsa$.
\end{definition}

The pre-order relation between lam, pair of lams and parameterised pair of lam, noted $\Subset$ is defined below.
This pre-order is central to prove that our algorithm indeed computes a fix point.

\begin{definition}
%{\bf OLD DEF: }
%Let $\lafsa \Subset \lafsa'$ whenever there is a cog name injection $\kappa$ from
%$\gr{\lafsa}$ to $\gr{\lafsa'}$ such that, for every $L \in \lafsa$ there is $L' \in
%\lafsa'$ with $\kappa(L) \subseteq L'$. With an abuse of notation, let
%\begin{itemize}
%\item[--]
%$\pairl{\lafsa_1 }{ \lafsa_1'} \Subset \pairl{\lafsa_2}{\lafsa_2'}$
%whenever  there is a cog name injection $\kappa$ from
%$\gr{\lafsa_1 \cup \lafsa_1'}$ to $\gr{\lafsa_2 \\\cup \lafsa_2'}$ such that, 
%for every $L_1 \in \lafsa_1$ there is $L_2 \in
%\lafsa_2$ with $\kappa(L_1) \subseteq L_2$ and for every $L_1' \in \lafsa_1'$ 
%there is $L_2' \in \lafsa_2'$ with $\kappa(L_1') \subseteq L_2'$ (lifting of 
%the pre-ordering relation to pairs of lams);
%\item[--]
%$\lambda \vect{c}. \pairl{\lafsa_1}{\lafsa_1'} \Subset 
%\lambda \vect{c}. \pairl{\lafsa_2}{\lafsa_2'}$
%whenever  there is a cog name injection $\kappa$ from
%$\gr{\lafsa_1 \cup \lafsa_1'}$ to \\$\gr{\lafsa_2 \cup \lafsa_2'}$ that is 
%the identity on $\vect{c}$ such that $\pairl{\lafsa_1 }{ \lafsa_1'} \Subset \pairl{\lafsa_2}{\lafsa_2'}$.
%\end{itemize} 
%
%
%
%\begin{update}
Let $\lafsa$ and $\lafsa'$ be lams and $\kappa$ be an injective function between cog names. We note $\lafsa \Subset_\kappa \lafsa'$ iff for every 
$L \in \lafsa$ there is $L' \in\lafsa'$ with $\kappa(L) \subseteq L'$.
Let
\begin{itemize}
%\item[--]
%$\pairl{\lafsa_1 }{ \lafsa_1'} \Subset_\kappa \pairl{\lafsa_2}{\lafsa_2'}$ iff $\lafsa_1 \Subset_\kappa \lafsa_2$ and $\lafsa_1' \Subset_\kappa \lafsa_2'$ ;
\item[--]
$\lambda \vect{c}. \pairl{\lafsa_1}{\lafsa_1'} \Subset_\kappa\lambda \vect{c}. \pairl{\lafsa_2}{\lafsa_2'}$
iff $\kappa$ is the identity on $\vect{c}$ and $\pairl{\lafsa_1 }{ \lafsa_1'} \Subset_\kappa \pairl{\lafsa_2}{\lafsa_2'}$.
\end{itemize}
Let also $\Subset$  be the relation
%
%Finally, we define the relation $\Subset$ between lams $\lafsa$, pairs of lams $\pairl{\lafsa}{\lafsa'}$ and between
%parameterised pairs of lams $\lambda \vect{c}. \pairl{\lafsa}{\lafsa'}$ as follow:
\begin{itemize}
\item[--]
$\lafsa \Subset \lafsa'$ iff there is $\kappa$ such that $\lafsa \Subset_\kappa
\lafsa'$;

%\item[--]
%$\pairl{\lafsa_1}{\lafsa_1'} \Subset \pairl{\lafsa_2}{\lafsa_2'}$
% iff there is $\kappa$ such that $\pairl{\lafsa_1}{ \lafsa_1'} \Subset_\kappa \pairl{\lafsa_2}{\lafsa_2'}$;
%
\item[--]
$\lambda \vect{c}. \pairl{\lafsa_1}{\lafsa_1'} \Subset\lambda \vect{c}. \pairl{\lafsa_2}{\lafsa_2'}$ iff there is $\kappa$ such that $\lambda \vect{c}. \pairl{\lafsa_1}{\lafsa_1'} \Subset_\kappa\lambda \vect{c}. \pairl{\lafsa_2}{\lafsa_2'}$.
\end{itemize}
\end{definition}

The set of lams with the $\Subset$ relation is a
pre-order with a bottom element, which is either 
$\pinull$ %or $\pairl{\pinull}{\pinull}$
or $\lambda \vect{c}.\pairl{\pinull}{\pinull}$ or
$( \cdots , \lambda \vect{c_i}. \pairl{\pinull}{\pinull}, \cdots )$ according to the
domain we are considering. 
In Figure~\ref{fig:lamOp} we define a number of basic operations on the lam model that are
 used
in the semantics of contracts.

The relevant property for the following theoretical development is the one below.
We say that an operation is monotone if,
whenever it is applied
to arguments in the pre-order relation $\Subset$, it returns values in the 
same pre-order relation $\Subset$).
The proof is straightforward and therefore omitted.

\begin{proposition}
%An operation {\tt op} is monotone with respect to $\preceq$, if, whenever $\lafsa
%\preceq \lafsa'$ then ${\tt op}(\lafsa) \preceq {\tt op}(\lafsa')$ (similarly 
%for binary operations).
%
The operations of extension, parallel, sequence and plus are monotone with respect to
$\Subset$. Additionally, if $\lafsa \Subset \lafsa'$ then 
$\lafsa \subst{\vect{c'}}{\vect{c}} \Subset \lafsa' \subst{\vect{c'}}{\vect{c}}$.
\end{proposition}

\subsection{\em The finite approximants of abstract method behaviours.}
\label{sec.finiteapproximants}

As explained above, the 
 lam model of a {\coreABS} program is obtained by means of a 
fixpoint technique plus a saturation applied to its contract class table. 
In particular, the lam model of the class table is a
lam table that maps each method $\C.\m$ of the program to $\lambda \vect{c}_{\C.\m}. \pairl{\lafsa_{\C.\m}}{\lafsa_{\C.\m}'}$ where $\vect{c}_{\C.\m} = \extr{\frr, \vect{\frs}}$, with $\frr(\vect{\frs})$ being the header of the method contract of
 $\C.\m$. The definition of $\extr{\frr, \vect{\frs}}$ is given in 
Figure~\ref{fig:extrproc} (we recall that names in headers occur linearly).
\begin{figure*}[t]
\myrules{\mathmode\simplemode}{
    \extr{\unit}  \eqdef  \varepsilon \and
    \extr{\X}  \eqdef  \varepsilon  \and
    \extr{[\cog {:} \obj, x_1 {:} \frr_1, \cdots , x_n {:} \frr_n]} \eqdef  \obj \, \extr{\frr_1} \, \cdots \, \extr{\frr_n} \and
    \extr{\fRec{\obj}{\frr}} \eqdef  \obj \, \extr{\frr}
    \and
    \extr{\vect{\frr}, \vect{\frs}} \eqdef \extr{\vect{\frr}} \extr{ \vect{\frs}}
}\vspace{-1em}
\caption{The extraction process.}\label{fig:extrproc}
\end{figure*}
The following definition presents the algorithm used to compute the next approximant of a method class table.
\begin{figure*}[t]
\myrules{\mathmode\simplemode}{
    \repl{\unit}{\unit} \eqdef \varepsilon \and
    \repl{\frr}{\X} \eqdef \varepsilon \and
    \repl{[\cog {:} \obj', x_1 {:} \frr_1', \cdots , x_n {:} \frr_n']}{[\cog {:} \obj, x_1 {:} \frr_1, \cdots , x_n {:} \frr_n]}
      \eqdef \obj' \, \repl{\frr_1'}{\frr_1} \, \cdots \, \repl{\frr_n'}{\frr_n} \and
    \repl{\fRec{\obj'}{\frr'}}{\fRec{\obj}{\frr}} \eqdef \obj' \, \repl{\frr'}{\frr}
}\vspace{-1em}
\caption{The cog mapping process.}\label{fig:cmproc}
\end{figure*}
\begin{definition}
\label{def.lafsatransf}
Let {\cct} be a contract class table of the form
\[
\bigl(\cdots, \C.\m \mapsto\mcontract{}{\frr_{\C.\m}(\overline{\frs_{\C.\m}})}{\pairl{\cntc_{\C.\m}}{\cntc_{\C.\m}'}}{ \frr_{\C.\m}'}, \cdots \bigr)
\]
\begin{enumerate}
\item
the \emph{ approximant 0} is defined as 
\[
\bigl(\cdots , \lambda(\extr{\frr_{\C.\m}, \vect{\frs_{\C.\m}}}).\pairl{\pinull}{\pinull}, \cdots \bigr) \; ;
\]

\item
let  $\lt = \bigl(\cdots, \lambda\extr{\frr_{\C.\m}, \vect{\frs_{\C.\m}}}.\pairl{\lafsa_{\C.\m}}{\lafsa_{\C.\m}'}, \cdots 
 \bigr)$ be the $n$-th approximant; the $n+1$-th approximant 
is  defined as $\bigl(\cdots, \lambda\extr{\frr_{\C.\m}, \vect{\frs_{\C.\m}}}.\pairl{\lafsa_{\C.\m}''}{\lafsa_{\C.\m}'''}, \cdots 
 \bigr)$ where
$\pairl{\lafsa_{\C.\m}''}{\lafsa_{\C.\m}'''}$ 
$= \cntc_{\C.\m}(\lt)_c \; \fatsemi \; \cntc_{\C.\m}'(\lt)_c$
with $c$ being the cog of $\frr_{\C.\m}$ and the function $\cntc(\lafsa)_c$ being defined by structural induction in Figure~\ref{fig:lamTrans}. 
\end{enumerate}
\end{definition}

\begin{figure*}[t]
\begin{enumerate}
\item[1.]
let $\overline{b_{\C,\m}} = (\gr{\lafsa_{\C,\m}} \cup \gr{\lafsa_{\C,\m}'})  \setminus
\overline{c_{\C,\m}}$. These are the free cog 
names that are replaced by \emph{fresh} cog names at every \emph{transformation step};

\medskip

\item[2.]
the transformation $\cntc_{\C,\m}
(\cdots \lambda \vect{c_{\C,\m}}.\pairl{\lafsa_{\C,\m}}{\lafsa_{\C,\m}'}, \cdots %, \langle \lafsa\fatcomma \lafsa' \rangle
)_c$ is defined inductively as follows:
\\
\\
\begin{tabular}{ll}
--&
$\pairl{\pinull}{\pinull} \addition (c_1,c_2)^{[\aa]}$ \hfill \rulename{L-GAzero}
\\
&\hspace*{.3cm} if
$\cntc_{\C,\m} = (c_1,c_2)^{[\aa]}$;
\\
\\
--&
$(\lambda \vect{c_{\D,\n}}. \pairl{\lafsa_{\D,\n}}{\lafsa_{\D,\n}'} \addition (c,c)^{\aa}
\subst{\overline{b_{\D,\n}'}}{\overline{b_{\D,\n}}}) \repl{\frr'}{\frr_{\D,\n}} \repl{\overline{\frs'}}{\overline{\frs_{\D,\n}}}$
%\Bigr) \subst{\overline{b'}}{\overline{b}}$
\hfill \rulename{L-SInvk}
\\
&\hspace*{.3cm} if
$\cntc_{\C,\m} = \D.\n \; \frr'(\overline{\frs'})\rightarrow\frr''$, $\frr' = 
[\cog {:} c, \bar{x} {:} \bar{\frr'}]$, and $\cct(\D)(\n) =
\mcontract{}{\frr_{\D,\n}(\overline{\frs_{\D,\n}})}{\pairl{\cntc_{\D,\n}}{\cntc_{\D,\n}'}}{\frr_{\D,\n}'}$
\\
&
\hspace*{.3cm} and $\overline{b_{\D,\n}'}$ are fresh cog names;
%\repl{\frr''}{\frr_{\D,\n}'})$;
\\
\\
-- & $(\lambda \vect{c_{\D,\n}}. \pairl{\lafsa_{\D,\n}}{\lafsa_{\D,\n}'}\addition (c,c')
\subst{\overline{b_{\D,\n}'}}{\overline{b_{\D,\n}}}) \repl{\frr'}{\frr_{\D,\n}} \repl{\overline{\frs'}}{\overline{\frs_{\D,\n}}}$
%\Bigr) \subst{\overline{b'}}{\overline{b}}$
\hfill \rulename{L-RSInvk}
\\
&\hspace*{.3cm} if
$\cntc_{\C,\m} = \D.\n \; \frr'(\overline{\frs'})\rightarrow\frr''$, $\frr' = 
[\cog {:} c', \bar{x} {:} \bar{\frr'}]$, and $c \neq c'$  and $\cct(\D)(\n) =
\mcontract{}{\frr_{\D,\n}(\overline{\frs_{\D,\n}})}{\pairl{\cntc_{\D,\n}}{\cntc_{\D,\n}'}}{\frr_{\D,\n}'}$
\\
&
\hspace*{.3cm} and $\overline{b_{\D,\n}'}$ are fresh cog names;
%\repl{\frr''}{\frr_{\D,\n}'})$;
\\
\\
--&
%$\Bigl(  
%$\lambda \vect{c_{\C,\m}}. 
$\pairl{\pinull}{\lafsa_{\D,\n}'' \cup 
\lafsa_{\D,\n}'''}$
%\Bigr) \subst{\overline{b'}}{\overline{b}}$
\hfill \rulename{L-AInvk}
\\
&\hspace*{.3cm} if
$\cntc_{\C,\m} = \D!\n \; \frr'(\overline{\frs'})\rightarrow\frr''$ and $\cct(\D)(\n) =
\mcontract{}{\frr_{\D,\n}(\overline{\frs_{\D,\n}})}{\pairl{\cntc_{\D,\n}}{\cntc_{\D,\n}'}}{\frr_{\D,\n}'}$
\\
&\hspace*{.3cm} 
and $\pairl{\lafsa_{\D,\n}''}{\lafsa_{\D,\n}'''} =
(\lambda \vect{c_{\D,\n}}. \pairl{\lafsa_{\D,\n}}{\lafsa_{\D,\n}'}
\subst{\overline{b_{\D,\n}'}}{\overline{b_{\D,\n}}}) \repl{\frr'}{\frr_{\D,\n}} \repl{\overline{\frs'}}{\overline{\frs_{\D,\n}}}$
%\\
%& \hspace*{.3cm} 
and $\overline{b_{\D,\n}'}$ are fresh cog names;
%\repl{\frr''}{\frr_{\D,\n}'})$;
\\
\\
--&
%$\Bigl( 
%$\lambda \vect{c_{\C,\m}}. (
$(\lambda \vect{c_{\D,\n}}. \pairl{\lafsa_{\D,\n}}{\lafsa_{\D,\n}'} \addition (c_1,c_2)^{[\aa]}
\subst{\overline{b_{\D,\n}'}}{\overline{b_{\D,\n}}}) \repl{\frr'}{\frr_{\D,\n}} \repl{\overline{\frs'}}{\overline{\frs_{\D,\n}}}$
%)$ 
%\Bigr) \subst{\frr'}{\frr_{\D,\n}} \subst{\overline{\frs'}}{\overline{\frs_{\D,\n}}}
%\subst{\frr''}{\frr_{\D,\n}'}\subst{\overline{b'}}{\overline{b}}$
\hfill \rulename{L-GAinvk}
\\
&\hspace*{.3cm} if
$\cntc_{\C,\m} = \D!\n \; \frr'(\overline{\frs'}) \rightarrow\frr''\seqpoint (c_1,c_2)^{[\aa]}$ and $\cct(\D)(\n) =
\mcontract{}{\frr_{\D,\n}(\overline{\frs_{\D,\n}})}{\pairl{\cntc_{\D,\n}}{\cntc_{\D,\n}'}}{\frr_{\D,\n}'}$
%$\cct(\D)(\n) =
%\mcontract{}{\frr_{\D,\n}}{\overline{\frs_{\D,\n}}}{\frr_{\D,\n}'}$
and $\overline{b_{\D,\n}'}$ are fresh cog names;
\\
\\
--&
$\cntc_{\C,\m}'
(\cdots , \lambda \vect{c_{\C,\m}}.\pairl{\lafsa_{\C,\m}}{\lafsa_{\C,\m}'}, \cdots
)_c 
 \fatsemi 
\cntc_{\C,\m}''
(\cdots ,\lambda \vect{c_{\C,\m}}.\pairl{\lafsa_{\C,\m}}{\lafsa_{\C,\m}'}, \cdots
)_c$ 
%
%$\frr_{\C,\m}(\overline{\frs_{\C,\m}}).\cntc_{\C,\m}'
%(\cdots \langle \lafsa_{\C,\m}\fatcomma \lafsa_{\C,\m}' \rangle, \cdots ) 
%\fatsemi \frr_{\C,\m}(\overline{\frs_{\C,\m}}).\cntc_{\C,\m}''
%(\cdots \langle \lafsa_{\C,\m}\fatcomma \lafsa_{\C,\m}' \rangle, \cdots ) $
\hfill \rulename{L-Seq}
\\
&\hspace*{.3cm} if
$\cntc_{\C,\m} = \cntc_{\C,\m}' \fatsemi \cntc_{\C,\m}''$;
\\
\\
--&
$\cntc_{\C,\m}'
(\cdots , \lambda \vect{c_{\C,\m}}.\pairl{\lafsa_{\C,\m}}{\lafsa_{\C,\m}'}, \cdots
)_c \; + \; 
\cntc_{\C,\m}''
(\cdots , \lambda \vect{c_{\C,\m}}.\pairl{\lafsa_{\C,\m}}{\lafsa_{\C,\m}'}, \cdots
)_c$ 
%
%$\frr_{\C,\m}(\overline{\frs_{\C,\m}}).\cntc_{\C,\m}'
%(\cdots \langle \lafsa_{\C,\m}\fatcomma \lafsa_{\C,\m}' \rangle, \cdots ) 
%\fatsemi \frr_{\C,\m}(\overline{\frs_{\C,\m}}).\cntc_{\C,\m}''
%(\cdots \langle \lafsa_{\C,\m}\fatcomma \lafsa_{\C,\m}' \rangle, \cdots ) $
\hfill \rulename{L-Plus}
\\
&\hspace*{.3cm} if
$\cntc_{\C,\m} = \cntc_{\C,\m}' + \cntc_{\C,\m}''$;
\\
\\
%--&
%$\cntc_{\C,\m}'
%(\cdots , \lambda \vect{c_{\C,\m}}.\pairl{\lafsa_{\C,\m}}{\lafsa_{\C,\m}'}, \cdots
%)_c \; \parop\; 
%\cntc_{\C,\m}''
%(\cdots ,\lambda \vect{c_{\C,\m}}.\pairl{\lafsa_{\C,\m}}{\lafsa_{\C,\m}'}, \cdots
%)_c$ 
%%
%$\frr_{\C,\m}(\overline{\frs_{\C,\m}}).\cntc_{\C,\m}'
%(\cdots \langle \lafsa_{\C,\m}\fatcomma \lafsa_{\C,\m}' \rangle, \cdots ) 
%\fatsemi \frr_{\C,\m}(\overline{\frs_{\C,\m}}).\cntc_{\C,\m}''
%(\cdots \langle \lafsa_{\C,\m}\fatcomma \lafsa_{\C,\m}' \rangle, \cdots ) $
%\hfill \rulename{L-Par}
%\\
%&\hspace*{.3cm} if
%$\cntc_{\C,\m} = \cntc_{\C,\m}' \parop \cntc_{\C,\m}''$.\\
%\\
--&
$\cntc_{\C,\m}'
(\cdots , \lambda \vect{c_{\C,\m}}.\pairl{\lafsa_{\C,\m}}{\lafsa_{\C,\m}'}, \cdots
)_c \; \rfloor\; 
\cntc_{\C,\m}''
(\cdots , \lambda \vect{c_{\C,\m}}.\pairl{\lafsa_{\C,\m}}{\lafsa_{\C,\m}'}, \cdots
)_c$ 
%
%$\frr_{\C,\m}(\overline{\frs_{\C,\m}}).\cntc_{\C,\m}'
%(\cdots \langle \lafsa_{\C,\m}\fatcomma \lafsa_{\C,\m}' \rangle, \cdots ) 
%\fatsemi \frr_{\C,\m}(\overline{\frs_{\C,\m}}).\cntc_{\C,\m}''
%(\cdots \langle \lafsa_{\C,\m}\fatcomma \lafsa_{\C,\m}' \rangle, \cdots ) $
\hfill \rulename{L-Par}
\\
&\hspace*{.3cm} if
$\cntc_{\C,\m} = \cntc_{\C,\m}' \rfloor \cntc_{\C,\m}''$.
\end{tabular}

\end{enumerate}
  \caption{ Lam transformation of \cct.}\label{fig:lamTrans}
\end{figure*}

It is worth to notice that there are two rules for synchronous invocations in 
Figure~\ref{fig:lamTrans}: \rulename{L-SInvk} dealing with  synchronous 
invocations on the same cog name of the caller -- the index $c$ of the transformation --,
\rulename{L-RSInvk} dealing with  synchronous invocations on different cog names.

%
%The \emph{lam transformation} of Definition~\ref{def.lafsatransf}
%is initially applied to the first approximant (of the so-called \emph{abstract
%class table})
Let
$$\Bigl( \cdots , \lambda \vect{\obj_{\C,\m}}.\pairl{{\lafsa_{\C,\m}}^0}{{\lafsa_{\C,\m}'}^0} , \cdots
\Bigr) = \Bigl( \cdots , \lambda \vect{\obj_{\C,\m}}.\pairl{\pinull}{\pinull} \rangle , \cdots
\Bigr),$$ 
and let
\[
\begin{array}{l}
\Bigl( \cdots ,\lambda \vect{\obj_{\C,\m}}.\pairl{{\lafsa_{\C,\m}}^0}{{\lafsa_{\C,\m}'}^0} , \cdots
\Bigr) , \\
\Bigl( \cdots ,\lambda \vect{\obj_{\C,\m}}.\pairl{{\lafsa_{\C,\m}}^1}{{\lafsa_{\C,\m}'}^1} , \cdots
\Bigr),\\
\Bigl( \cdots ,\lambda \vect{\obj_{\C,\m}}.\pairl{{\lafsa_{\C,\m}}^2}{{\lafsa_{\C,\m}'}^2} , \cdots
\Bigr), \cdots
\end{array}
\]
be the sequence obtained by the algorithm of Definition~\ref{def.lafsatransf}
(this is the standard Knaster-Tarski technique). This sequence is non-decreasing (according to $\Subset$) because it is defined 
as a composition of monotone operators, see 
Proposition~\ref{def.lafsatransf}. 
%
%In the above sequence,
%the lam function $\lambda \vect{\obj_{\C,\m}}.\pairl{{\lafsa_{\C,\m}}^i}{{\lafsa_{\C,\m}'}^i}$ is the 
%\emph{i-th finite approximant} of {\C.\m}.
%
Because of the
creation of new cog names at each iteration, the fixpoint of the above sequence may not exist. We have already discussed the example of ${\tt Math.fact\_nc}$.
%For example, the method contract $\cct(\C)(\m)={}$
%%contract $\C.\m \; [\cog {:} c](~)\rightarrow [\cog {:} c] \seqpoint (c',c) $, 
%%where  
%$$\mcontract{}{[\cog {:} c](~)}{\{ \C.\m \; [\cog {:} c'](~)\rightarrow [\cog {:} c'] \seqpoint (c,c') \} }{[\cog {:} c']},$$
%% $\obj() \{ 
%% \C.\m^{\gg}(\obj') \}~\obj$
%yields the infinite sequence 
%\[
%\begin{array}{l}
%\lambda c. \pairl{\pinull}{\pinull}
%\\
%\lambda c. \pairl{\slam{ (c,c') }}{\pinull}
%\\
%\lambda c. \pairl{\slam{ (c,c'), (c',c_1) }}{\pinull}
%\\
%\lambda c. \pairl{\slam{ (c,c'), (c',c_2), (c_2,c_3) }}{\pinull}
%\\
%\cdots
%\end{array}
%\]
%%where a set of pairs $W$ represents the \emph{one-state/no-transition 
%%lafsa} $(\{ \la \},\emptyset,\la, \la)$. 
%%
In order to let our analysis terminate, after a given approximant, we 
run the Knaster-Tarski technique \emph{using a different semantics for the operations
\rulename{L-SInvk}, \rulename{L-RSInvk}, \rulename{L-AInvk}, and \rulename{L-GAinvk}} 
(these are the rules where cog names may be created).
In particular, when these operations are used at approximants larger than $n$, the
renaming of free variables is disallowed. That is, the substitutions $\subst{\overline{b_{\D,\n}'}}{\overline{b_{\D,\n}}}$ in Figure~\ref{fig:lamTrans} are removed. It is straightforward to verify that
these operations are still monotone. It is also straightforward to demonstrate by a 
simple cardinality argument the
existence of fixpoints in the lam domain 
by running the Knaster-Tarski technique with this different
semantics.
This method is called a \emph{saturation technique at} $n$. 

For example, if we compute the third approximant of
$$
\begin{array}{l}
{\tt Math.fact\_nc} \mapsto
\\
\quad [\cog {:} \obj](\unit) \{ 
\pairl{\pinull + \mbox{\Abs{Math\!}}\name{fact\_nc}\; [\cog {:} \obj'](\unit)\rightarrow \unit \seqpoint (\obj,\obj')}{\pinull}
\\
\qquad \qquad \qquad \} \unit
\end{array}
$$
% $\obj() \{ 
% \C.\m^{\gg}(\obj') \}~\obj$
we get the sequence 
\[
\begin{array}{l}
\lambda c. \pairl{\pinull}{\pinull}
\\
\lambda c. \pairl{\slam{ (c,c_0) }}{\pinull}
\\
\lambda c. \pairl{\slam{ (c,c_0), (c_0,c_1) }}{\pinull}
\\
\lambda c. \pairl{\slam{ (c,c_0), (c_0,c_1), (c_1,c_2) }}{\pinull}
\end{array}
\]
and, if we saturate a this point, we obtain
\[
\begin{array}{l}
\lambda c. \pairl{\slam{(c,c_0) (c_0,c_0), (c_0,c_1), (c_1,c_2) }}{\pinull}
\\
\lambda c. \pairl{\slam{(c,c_0) (c_0,c_0), (c_0,c_1), (c_1,c_2) }}{\pinull} \qquad 
\mbox{fixpoint}
\end{array}
\]

\begin{definition}
Let $(\ct , \{ \vect{T \ x \sseq} s \} , \cct)$ be a {\coreABS} program and let
$\Bigl( \cdots \lambda \vect{c_{\C,\m}}. 
\pairl{{\lafsa_{\C,\m}}^{n+h}}{{\lafsa_{\C,\m}'}^{n+h}} , \cdots \Bigr)$ be  the 
fixpoint (unique up to renaming of cog names)
obtained by the saturation technique at $n$. The \emph{abstract class table at} $n$, written $\act_{[n]}$, is a map that takes $\C.\m$ and returns 
$\lambda \vect{c_{\C,\m}}. \pairl{{\lafsa_{\C,\m}}^{n+h}}{{\lafsa_{\C,\m}'}^{n+h}}$.
\end{definition}

Let $(\ct , \{ \vect{T \ x \sseq} s \}, \cct)$ be a {\coreABS} program and
$$\inferst{\Gamma 
}{{\rm start}}{\{ 
\vect{T \ x \sseq} s \}}{\pairl{\cntc}{\cntc'}}{\mathcal U}{\Gamma} \; .$$
%
%\tfJ{\Gamma + \this: [\cog {:} {\rm start}]}{{\rm start}}{\{ 
%\vect{T \ x \sseq} s \}}{(\T',\frr)}{\cntc}{}$ 
%where
%$\Gamma$ maps $\C.\m$ to the interface of $\cct(\C)(\m)$ and 
%``${\rm start}$'' is a special cog name corresponding to the object ${\it start}$, 
%and let
%$\act_{[n]}$ be the corresponding abstract class table at $n$.
The \emph{abstract
semantics saturated at} $n$ of $(\ct , \{ \vect{T \ x \sseq} s \},$ $\cct)$ is 
%{\small
%\[  
%\pairl{\cntc \Bigl( \cdots ,\act_{[n]}(\C.\m) , 
%\cdots \Bigr)_{\rm start}}{\cntc'\Bigl( \cdots ,\act_{[n]}(\C.\m) , 
%\cdots \Bigr)_{\rm start}}   \; .
%\]}
$(\cntc (\act_{[n]})_{\rm start})\cseq (\cntc'(\act_{[n]})_{\rm start})$.

As an example, in Figure~\ref{fig:satEx} we compute the abstract semantics 
saturated at 2 of 
the
class \name{Math} in Figure~\ref{fig.Math}. 

\begin{figure*}[t] 
{\footnotesize
\begin{tabular}{l@{\quad}l@{\quad}l@{\quad}l@{\quad}l}
method & approx. 0 & approx. 1 & approx.2 & saturation
\\
\hline
\\
\mbox{\Abs{Math.fact\_g}} & $\lambda c. \pairl{\pinull}{\pinull}{}$ 
& $\lambda c. \pairl{\slam{ (c,c)}}{\pinull}{}$
& $\lambda c. \pairl{\slam{(c,c)}}{\pinull}{}$
\\
\mbox{\Abs{Math.fact\_ag}} & $\lambda c. \pairl{\pinull}{\pinull}{}$
& $\lambda c. \pairl{\slam{(c,c)^\aa}}{\pinull}{}$
& $\lambda c. \pairl{\slam{(c,c)^\aa }}{\pinull}{}$
\\
\mbox{\Abs{Math.fact\_nc}} & $\lambda c. \pairl{\pinull}{\pinull}{}$
& $\lambda c. \pairl{\slam{(c,c')}}{\pinull}{}$
& $\lambda c. \pairl{\slam{(c,c'), (c',c'')}}{\pinull}{}$
& $\lambda c. \pairl{\slam{(c,c'), (c',c'), (c',c'')}}{\pinull}{}$
\end{tabular}
}
\caption{Abstract class table computation for class \name{Math}.}
\label{fig:satEx}
\end{figure*}

\subsection{\em Deadlock analysis of lams.}
\label{sec.deadlock_analysis_lams}

\begin{definition}
\label{def.lam-circularity}
Let $L$ be a relation on cog names (pairs in $L$ are either of the form $(\obj_1, \obj_2)$ or of the form
$(\obj_1, \obj_2)^\aa$. 
%The {\get}\emph{-closure} of
%$L$, noted $L^\get$, is the least set such that
%$$\begin{array}{c}
%L \in L^\get 
%\qquad 
%\bigfract{(c_1,c_2) \in L^\get \quad (c_2,c_3)^{[\aa]} \in L^\get}{(c_1,c_3) \in L^\get}
%%\\ 
%%\\
%%\bigfract{(c_1,c_2) \in L^\get \quad (c_2,c_3)^\aa \in L^\get}{(c_1,c_3) 
%%\in L^\get}
%\end{array}$$
$L$ \emph{contains a circularity} if $L^{\get}$ has a pair $(\obj,\obj)$ (see Definition~\ref{def.obj-circularity}). Similarly, $\pairl{\lafsa}{\lafsa'}$ (or 
$\lambda \vect{c}. \pairl{\lafsa}{\lafsa'}$)
has a circularity if
there is $L \in \lafsa \cup \lafsa'$ that contains a circularity.

A {\coreABS} program with an abstract class table saturated at $n$ is \emph{deadlock-free} if its abstract semantics 
 $\pairl{\lafsa}{\lafsa'}$ does not contain a circularity.
\end{definition}

\medskip
\noindent
The fixpoints for \mbox{\Abs{Math.fact\_g}} and \mbox{\Abs{Math.fact\_ag}} are
found at the third iteration. According to the above definition of deadlock-freedom,
\mbox{\Abs{Math.fact\_g}} yields a deadlock, whilst \mbox{\Abs{Math.fact\_ag}} is deadlock-free because $\{ (c,c)^\aa \}^\get$ does not contain any circularity.
As discussed before, there exists no fixpoint for \mbox{\Abs{Math.fact\_nc}}. If
we decide to stop at the approximant 2 and saturate, we get 
$$\lambda c. \pairl{\slam{(c,c'), (c',c') , (c',c'')}}{\pinull},$$
which contains a  circularity that is a false positive.

Note that saturation might even start at the approximant 0
(where every method is $\lambda c. \pairl{\pinull}{\pinull}{}$). In this case,
for
\mbox{\Abs{Math.fact\_g}} and \mbox{\Abs{Math.fact\_ag}},
we get the same answer and the same pair of lams as the above third approximant. 
For \mbox{\Abs{Math.fact\_nc}}
we get $$\lambda c. \pairl{\slam{(c,c'), (c',c')}}{\pinull}{},$$ which contains 
a circularity. 

\medskip

In general, in techniques like the one we have presented, 
it is possible to augment the precision of the analysis 
by delaying the saturation. However, assuming that pairwise different method contracts
have disjoint free cog names (which is a reasonable assumption), we have not found 
any sample {\coreABS} code where saturating at 1 gives a better precision than saturating at 0. While this issue is left open, the current version of our 
tool {\DFfABS} allows one to specify the
saturation point; the default saturation point is 0.

\begin{figure*}[t] 
{\footnotesize
\begin{tabular}{l@{\quad}l@{\quad}l@{\quad}l}
method & approx. 0 & approx. 1 & approx.2
\\
\hline
\\
\mbox{\Abs{CpxSched.m1}} & $\lambda c,c',c''. \pairl{\pinull}{\pinull}{}$ 
& $\lambda c,c',c''. \pairl{\pinull}{\slam{(c',c''),(c'',c')}}{}$
& $\lambda c,c',c''. \pairl{\pinull}{\slam{(c',c''),(c'',c')}}{}$
\\
\mbox{\Abs{CpxSched.m2}} & $\lambda c,c'. \pairl{\pinull}{\pinull}{}$
& $\lambda c,c'. \pairl{\slam{(c,c')}}{\pinull}{}$
& $\lambda c,c'. \pairl{\slam{(c,c')}}{\pinull}{}$
\\
\mbox{\Abs{CpxSched.m3}} & $\lambda c. \pairl{\pinull}{\pinull}{}$
& $\lambda c. \pairl{\pinull}{\pinull}{}$
\end{tabular}
}
\caption{Abstract class table computation for class \name{CpxSched}.}
\label{fig:satExCpx}
\end{figure*}

The computation of the abstract class table for class \name{CpxSched} does not need
any saturation, all methods are non-recursive and encounter their fixpoint by iteration 2. (See Figure~\ref{fig:satExCpx}.)
The abstract class table shows a circularity for method \name{m1}, manifesting the presence of a deadlock.

The correctness of the fixpoint analysis of contracts discussed in this section 
is demonstrated in Appendix~\ref{sec:analysis-proof}.
%
%our technique is affirmed in the following theorem.
%Given a configuration $\nt{cn}$ of a program with  $\act_{[n]}$ as its abstract class table at $n$, and such that 
%$\Delta\vdash_R \nt{cn} : \pairl{\cntc_1}{\cntc_2}$, 
%we write $\abstractsemantics{\nt{cn}}{\act_{[n]}}$  
%%$\abstractsemantics{\pairl{\cntc_1}{\cntc_2}}{\act_{[n]}}$ 
%meaning 
%{\small
%\[
%\pairl{\cntc_1 \Bigl( \cdots ,\act_{[n]}(\C.\m) , \cdots \Bigr)_{\rm start}}{\cntc_2\Bigl( \cdots ,\act_{[n]}(\C.\m) , \cdots \Bigr)_{\rm start}} \; .
%\] }
%
%\begin{theorem}\label{thm:analysis-correct}
%Let $(\ct , \{ \vect{T \ x \sseq} s \}, \cct)$ be a \coreABS{} program, $\act_{[n]}$ be its abstract class table at $n$, and $\nt{cn}$ be a configuration of its operational semantics.
%  \begin{enumerate}
%  \item If $\nt{cn}$ has a circularity, then at least one relation in $\abstractsemantics{\nt{cn}}$, has a circularity;
%\item if $\nt{cn}\to \nt{cn}'$ and $\abstractsemantics{\nt{cn}'}$ has a circularity, then a circularity is already present in $\abstractsemantics{\nt{cn}}$.
%  \end{enumerate}
%\end{theorem}
%\begin{proof}
%  See Appendix~\ref{sec:analysis-proof}.
%\end{proof}
%% \begin{theorem}
%% A {\coreABS} program is
%% deadlock-free  if its  abstract semantics computed with an 
%% abstract class table saturated at the $n$-th approximant, for some $n$, 
%% is deadlock-free.
%% \end{theorem}
%
We remark that this technique 
%the fixpoint analysis of contracts discussed in this section 
is
as modular as the inference system: once the contracts of a module have been computed,
one may run the fixpoint analysis and attach the corresponding abstract values to 
the code. Analysing a program reduces to computing the lam of the main function. 

%
%\begin{itemize}
%\item  
%The lam  for the contract
%$
%\mcontract{\Abs{Math}.\name{fact\_g}}{\obj[\,](\unit)}{\Abs{Math}.\name{fact\_g}\; \obj[\,](\unit)\rightarrow \unit \seqpoint (\obj,\obj)}{\unit}.
%$  of \name{fact\_g} is computed as follows
%
%
%%
% 
%\item  
%\name{fact\_ag} has contract
%$\mcontract{\Abs{Math}.\name{fact\_ag}}{\obj[\,](\,)}{\Abs{Math}.\name{fact\_ag}\; \obj[\,](\unit)\rightarrow \unit \seqpoint (\obj,\obj)^\aa}{\unit}$.
%
%
%
%\item  
%\name{fact\_g} has contract 
%$\mcontract{\Abs{Math}.\name{fact\_nc}}{\obj[\,](\unit)}{\Abs{Math}.\name{fact\_ng}\; \objb[\,](\unit)\rightarrow \unit \seqpoint (\obj,\objb)}{\unit}$.
%
%
%
%\end{itemize}
%
%
%
%
%Figure~\ref{fig1}(i) illustrates the (single state)  lam of the 
%contracts $\C.\n~\objb[\f:Y] (\,)\rightarrow\objb'[\f:Y] {\seqpoint} (\obj,\objb)$,
%assuming that the contract of {\tt C.n} is $\pinull$. According to our
%inference system, this contract is associated to a {\coreABS} code 
%such as
%\begin{absexamplen}
%    Fut<D> x = b!n(); 
%    x.get; 
%\end{absexamplen}
%Figure~\ref{fig1}(ii)  displays the (single state)  lam of the 
%contract $\C.\n~\objb[\f:Y] (\,)\rightarrow\objb'[\f:Y] {\seqpoint} (\obj,\objb)^\aa$ (which corresponds to an \Abs{await} operation).
%
%The models in Figure~\ref{fig1}(i) and (ii) do not manifest 
%any problematic dependency between object names as long as 
%the values of $a$ and $b$ are different. 
%However, a critical pair appears if $\obj = \objb$ -- an 
%\emph{object-circularity}.
%

%%% Local Variables:
%%% mode: latex
%%% TeX-master: "SoSyM"
%%% End:

\section{The model-checking analysis of contracts}
\label{sec.mutanalysis}
%!TEX root = SoSyM.tex

The second analysis technique for the contracts of 
Section~\ref{sec.FJg-contracts}
consists  of computing contract models \emph{by expanding} their 
invocations. We therefore begin this section by introducing a semantics of contracts that 
is alternative to the one of Section~\ref{sec.contractanalysis}.

\subsection{\em Operational semantics of contracts.}

%$$\CP[~] \quad ::= \quad [~] \qquad | \qquad \CP[~] \sparop\cntc 
%\qquad | \qquad \CP[~]\seqmut\cntc  $$
%
%As usual $\CP[\cntc]$ is the contract where the hole of $\CP[~]$ is replaced by $\cntc$.
%
%The operational semantics uses two auxiliary functions on contracts \syn{} and \asyn{}, which given a contract extracts 
%its synchronous (resp. asynchronous) behaviour. They corresponds to the pair of lams produced by the fixed-point semantics of Section~\ref{sec.contractanalysis}.
%
%\begin{definition}[\syn{} and \asyn{}]
%\ \\
%$\begin{array}{rcl}
%  \syn(\cntc) &=& 
%  \begin{cases}
%    \pinull & \mbox{ if } \cntc = \C!\m~\frr(\bar{\frr})\rightarrow\frr'; \\
%    %\C!\m~\frr(\bar{\frr})\rightarrow\frr'\sparop(c,c')^{[\aa]} & \mbox{ if } \cntc =  \C!\m~\frr(\bar{\frr})\rightarrow\frr'\seqpoint(c,c')^{[\aa]}\\
%    \syn(\cntc_1) + \syn(\cntc_2) & \mbox{ if } \cntc = \cntc_1 + \cntc_2;\\
%    \syn(\cntc_1) \fatsemi \syn(\cntc_2) & \mbox{ if } \cntc = \cntc_1 \fatsemi \cntc_2;\\
%    \cntc &\mbox{ otherwise}
%  \end{cases}\\
%\\
%  \asyn(\cntc) &=& 
%  \begin{cases}
%    \cntc & \mbox{ if } \cntc = \C!\m~\frr(\bar{\frr})\rightarrow\frr'; \\
%    \asyn(\cntc_1) + \asyn(\cntc_2) & \mbox{ if } \cntc = \cntc_1 + \cntc_2;\\
%    \asyn(\cntc_1) \fatsemi \asyn(\cntc_2) & \mbox{ if } \cntc = \cntc_1 \fatsemi \cntc_2;\\
%    \pinull &\mbox{ otherwise}
%  \end{cases}
%\end{array}$
%\end{definition}

\begin{figure*}
contract pairs
\[
\cntcp \; ::= \quad \cntc %\quad | \quad \pairl{\cntc}{\cntc} 
\quad | \quad \pairl{\cntcp}{\cntcp}_c
\quad | \quad \cntcp \addition (c,c')^{[\aa]} \quad | \quad \cntcp + \cntcp 
\quad | \quad \cntcp \fatsemi \cntcp  %\quad | \quad \cntcp \parop \cntcp 
\quad | \quad \cntcp \rfloor \cntcp 
\]

\medskip

contract pairs contexts $(\sharp \in \{ +, \fatsemi, %\parop, 
\rfloor \})$
\[
\begin{array}{r@{\quad}l}
\CD[\,] \; ::= & [\, ] \quad | \quad \CD[\,] \addition (c',c'')^{[\aa]} 
\quad | \quad \CD[\,] \sharp \cntcp 
\quad | \quad \cntcp \sharp \CD[\,] 
\\
\\
\CP[\,]_c \; ::= & \pairl{\CD[\,]}{\cntcp}_c \quad | \quad  \pairl{\cntcp}{\CD[\,]}_c \quad | \quad\pairl{\CP[\,]_{c}}{\cntcp}_{c'}
\quad | \quad  \pairl{\cntcp}{\CP[\,]_{c}}_{c'}
\quad | \quad \CP[\,]_c \addition (c',c'')^{[\aa]} 
\quad | \quad \CP[\,]_c \sharp \cntcp 
\quad | \quad \cntcp \sharp \CP[\,]_c
%\\
%\\
%| &
%\CP[\,]_c \fatsemi \cntcp  
%\quad | \quad \cntcp \fatsemi \CP[\,]_c
%\quad | \quad 
% \CP[\,]_c \parop \cntcp 
%\quad | \quad \cntcp \parop  \CP[\,]_c
%\quad | \quad \CP[\,]_c \rfloor \cntcp 
%\quad | \quad \cntcp \rfloor \CP[\,]_c
\end{array}
\]

\medskip

reduction relation
\[
\begin{array}{c}
\mathrule{Red-SInvk}{
	\begin{array}{c}
	\C.\m = \frs(\bar{\frs}) \{ \pairl{\cntc}{\cntc'} \} \frs'
	\qquad \frr = [\cog{:}c, \bar{x} {:} \bar{\frr''}]
	\\ 
	\gr{\pairl{\cntc}{\cntc'}} \setminus \gr{\frs,\bar{\frs}}=\wt{z} 
	\\
	\wt{w} \mbox{ are fresh} 
    \qquad
	\pairl{\cntc}{\cntc'}\subst{\wt{w}}{\wt{z}} 
		\subst{\frr,\bar{\frr}}{\frs,\bar{\frs}} = \pairl{\cntc''}{\cntc'''}
    \end{array}
  }{
    \CP[\C.\m \;\frr(\bar{\frr})\rightarrow \frr']_c
    \lred{} 
   	\CP[\pairl{\cntc''}{\cntc'''}_c]_c
  }
\qquad
\mathrule{Red-RSInvk}{
	\begin{array}{c}
	\C.\m = \frs(\bar{\frs}) \{ \pairl{\cntc}{\cntc'} \} \frs'
	\qquad \frr = [\cog{:}c', \bar{x} {:} \bar{\frr''}] \qquad c \neq c'
	\\ 
	\gr{\pairl{\cntc}{\cntc'}} \setminus \gr{\frs,\bar{\frs}}=\wt{z} 
	\\
	\wt{w} \mbox{ are fresh} 
    \qquad
	\pairl{\cntc}{\cntc'}\subst{\wt{w}}{\wt{z}} 
		\subst{\frr,\bar{\frr}}{\frs,\bar{\frs}} = \pairl{\cntc''}{\cntc'''}
    \end{array}
  }{
    \CP[\C.\m \;\frr(\bar{\frr})\rightarrow \frr']_c
    \lred{} 
   	\CP[\pairl{\cntc''}{\cntc'''}_{c'} \addition (c,c')]_c
  }
\\
\\
\mathrule{Red-AInvk}{
    \begin{array}{c}
      \C.\m = \frs(\bar{\frs}) \{ \pairl{\cntc}{\cntc'} \} \frs'
      \qquad \frr = [\cog{:}c', \bar{x} {:} \bar{\frr''}]
      \\ 
      \gr{\pairl{\cntc}{\cntc'}} \setminus \gr{\frs,\bar{\frs}}=\wt{z} 
      \\
      \wt{w} \mbox{ are fresh} 
      \qquad 
      \pairl{\cntc}{\cntc'}\subst{\wt{w}}{\wt{z}} \subst{\frr,\bar{\frr}}{\frs,\bar{\frs}} = \pairl{\cntc''}{\cntc'''}
    \end{array}
  }{
    \CP[\C!\m \;\frr(\bar{\frr})\rightarrow \frr']_c
    \lred{} 
    \CP[\pairl{\cntc''}{\cntc'''}_{c'}]_c
  }
\qquad
\mathrule{Red-GAInvk}{
    \begin{array}{c}
      \C.\m = \frs(\bar{\frs}) \{ \pairl{\cntc}{\cntc'} \} \frs'
      \qquad \frr = [\cog{:}c', \bar{x} {:} \bar{\frr''}]
      \\ 
      \gr{\pairl{\cntc}{\cntc'}} \setminus \gr{\frs,\bar{\frs}}=\wt{z} 
      \\ 
      \wt{w} \mbox{ are fresh} 
      \qquad
      \pairl{\cntc}{\cntc'}\subst{\wt{w}}{\wt{z}} \subst{\frr,\bar{\frr}}{\frs,\bar{\frs}} = \pairl{\cntc''}{\cntc'''}
    \end{array}
  }{
   \CP[\C!\m \;\frr(\bar{\frr})\rightarrow \frr'\seqpoint (c'',c''')^{[\aa]}]_c    \lred{} 
    \CP[\pairl{\cntc''}{\cntc'''}_{c'}\addition (c'',c''')^{[\aa]}]_c
  }
%  \\
%  \\
%  \mathax{Red-Choice}{
%  }{
%    \CP[\cntc + \cntc']
%    \lred{} 
%    \CP[\cntc\seqmut \cntc']
%  }
  \\
  \\  
\end{array}
\]
\caption{Contract reduction rules.}
\label{fig:contred}
\end{figure*}

The operational semantics of a contract is defined as a reduction relation between 
terms that are \emph{contract pairs} $\cntcp$, whose syntax is defined in 
Figure~\ref{fig:contred}. These contract pairs highlight (in the operational semantics) 
the fact
that every contract actually represents two collections of relations on cog names: those corresponding 
to the \emph{present states} and those corresponding to
 \emph{future states}. We have discussed this dichotomy in Section~\ref{sec.FJg-contracts}.

In Figure~\ref{fig:contred} we have also defined the \emph{contract pair contexts}, noted 
$\CP[~]_c$, which are indexed contract pairs with a hole. The index 
$c$ indicates that the hole is immediately enclosed by $\pairl{\cdot}{\cdot}_{c}$.

The reduction relation that defines the evaluation of contract pairs 
$\pairl{\cntcp_1}{\cntcp_1'}_c \lred{} \pairl{\cntcp_2}{\cntcp_2'}_c$
is defined in 
Figure~\ref{fig:contred}. There are four reduction rules: \rulename{Red-SInvk}
for synchronous invocation on the same cog name of the caller (which is stored in 
the index of the enclosing pair), 
\rulename{Red-RSInvk} for synchronous invocations on different cog name, \rulename{Red-AInvk} for asynchronous invocations,
 and  \rulename{Red-GAInvk} for asynchronous invocations with synchronisations.
We observe that every evaluation step amounts to expanding method invocations 
by replacing free cog names in method contracts with fresh names and
\emph{without modifying the syntax tree of contract pairs}. 

To illustrate the operational semantics of contracts we discuss 
three
examples:
\begin{enumerate}
\item
Let 

{\small
${\tt F.f} \; = \; [\cog : c](x: [\cog : c'], y: [\cog : c'']) \{ $
\\
\qquad $\pairl{ 
({\tt F.g} \, [\cog : c'](x: [\cog : c'']) \rightarrow \unit).(c,c')  \; +  \; 
\pinull.(c',c'')}{\pinull}$
\\
\qquad $  \} \unit$}
\\
and 
\\
{\small
${\tt F.g} \; = \; \, [\cog : c](x: [\cog : c']) \{ \pairl{\pinull.(c,c') + \pinull.(c',c)}{\pinull} \} \unit$
}

Then

\smallskip
{\scriptsize
\[
\begin{array}{l}
\pairl{{\tt F!f}\, [\cog : c](x: [\cog : c'], y: [\cog : c'']) \rightarrow \unit 
}{\pinull}_{\rm start}
\\
\lred{} \; 
\pairl{\pairl{\pinull}{({\tt F.g} \, [\cog : c'](x: [\cog : c'']) \rightarrow \unit).(c,c')  \; +  \; 
\pinull.(c',c'')}_{c}}{\pinull}_{\rm start}
\\
\lred{} \; 
\pairl{\pairl{\pinull}{\pairl{\pinull.(c',c'') + \pinull.(c'',c')}{\pinull}_{c'}\addition (c,c')  \; +  \; 
\pinull.(c',c'')}_c}{\pinull}_{\rm start}
\end{array}
\]}

\smallskip

The contract pair in the final state \emph{does not contain method invocations}. This 
is because the above main function is not recursive.
Additionally, the evaluation of ${\tt F.f}\, [\cog : c](x: [\cog : c'], y: [\cog : c''])$ 
\emph{has not created names}. This 
is because names in the bodies of 
{\tt F.f} and {\tt F.g} are bound.

\item
Let
{\small
\[
{\tt F.h}\; = \; [\cog : c]( \unit ) \{ \pairl{\pinull}{({\tt F.h}[\cog : c'] (\unit) 
\rightarrow \unit) \parop \pinull.(c,c') } \} \unit
\]}
Then

\smallskip

{\scriptsize
$\begin{array}{l}
\pairl{{\tt F!h} \, [\cog : c]( \unit ) \rightarrow \unit}{\pinull}_{\rm start}
\\
\lred{} \, 
\pairl{\pairl{\pinull}{({\tt F.h}[\cog : c'] (\unit) 
\rightarrow \unit) \parop \pinull.(c,c') }_c}{\pinull}_{\rm start}
\\
\lred{} \,
\pairl{\pairl{\pinull}{\pairl{\pinull}{({\tt F.h}[\cog : c''] (\unit) 
\rightarrow \unit) \parop \pinull.(c',c'') }_{c'} \parop \pinull.(c,c') }_c}{\pinull}_{\rm start} 
\\
\lred{}
\cdots  
\end{array}$}

\smallskip

\noindent where, in this case, the contract pairs grow in the number of dependencies as the
evaluation progresses. This growth \emph{is due to the presence of a free name}
in the definition of {\tt F.h} that, as said, corresponds to generating a fresh
name at every recursive invocation.

\item Let
{\small
\[
{\tt F.l} \, = \, [\cog : c](~) \{ \pairl{ \pinull.(c,c') \; \fatsemi \; 
( \pinull.(c,c') \rfloor {\tt F!l} \, [\cog : c](~) \rightarrow \unit)}{\pinull} \} \unit
\]}
Then

\smallskip

{\scriptsize
\[
\begin{array}{l}
\pairl{{\tt F!l} \,[\cog : c](~) \rightarrow \unit}{\pinull}_{\rm start}
\\
\lred{} \;
\pairl{\pairl{\pinull}{\pinull.(c,c') \; \fatsemi \; 
( \pinull.(c,c') \rfloor {\tt F!l} \, [\cog : c](~) \rightarrow \unit)}_c}{\pinull}_{\rm start}
\\
\lred{} \;
\langle \langle {\pinull} \fatcomma \pinull.(c,c') \; \fatsemi \; 
( \pinull.(c,c') \rfloor 
 \\  
\qquad  \pairl{\pinull}{\pinull.(c',c'') \; \fatsemi
( \pinull.(c',c'') \rfloor {\tt F!l} \, [\cog : c''](~) \rightarrow \unit)}_{c'}
)\rangle_c \fatcomma {\pinull}\rangle_{\rm start}
\\
\lred{}  \cdots 
\end{array}
\]}

\smallskip

\noindent 
In this case, the contract pairs grow in the number of ``$\fatsemi$''-terms, 
which become
larger and larger as the evaluation progresses.
\end{enumerate}

It is clear that, in presence of recursion and of free cog names in method contracts,
a technique that analyses contracts by expanding method invocations is fated to fail 
because the system is  infinite state. However,
it is possible to stop the expansions at suitable points without losing any relevant
information about dependencies. In this section we highlight the technique we
have developed in~\cite{GL2014} that has been prototyped for {\coreABS} in
{\DFfABS}.

\subsection{\em Linear recursive contract class tables.}

Since contract pairs models may be infinite-state, instead of resorting to a saturation technique, which introduces  inaccuracies, we exploit a generalisation of permutation theory
that let us decide when stopping the evaluation with the guarantee that if no circular
dependency has been found up to that moment then it will not appear afterwards.
That stage corresponds to the \emph{order} of an associated permutation.
It turns out that this technique is suited for so-called \emph{linear recursive} 
contract class
tables.

\begin{definition}
\label{def.linearrecursion}
A contract class table is \emph{linear recursive} if
(mutual) recursive invocations in bodies of methods have \emph{at most one recursive invocation}.
\end{definition}
It is worth to observe that a {\coreABS} program may be linear recursive while the 
corresponding contract class table is not.
For example,
consider the following method {\tt foo} of class {\tt Foo} that prints integers by
invoking a printer service and awaits for the termination of the printer task and
for its own termination: 
\begin{absexamplen}
Void foo(Int n, Print x){
    Fut<Void> u, v ;
    if (n == 0) return() ; 
    else { u = this!foo(n-1, x) ;
           v = x!print(n) ;
           await v? ;
           await u? ;
           return() ;
     }
}
\end{absexamplen}
While {\tt foo} has only one recursive invocation, its contract written in Figure~\ref{fig.foocontract} is not.
\begin{figure*}[t]
{\small
\[
\begin{array}{rl}
{\tt Foo.foo} = &  [\cog : c] (\unit, x:[\cog : c']) \{ 
\; \pairl{\pinull \; + 
\; \Bigl( \bigl( \cntc .(c,c')^{\aa} 
\rfloor \cntc' \bigr)
\fatsemi \cntc'.(c,c)^{\aa} \Bigr)}{\pinull}
\; \} \rightarrow \unit 
\\
{\rm where} \; \; \cntc = & {\tt Print!print} \, [\cog : c'](\unit) \rightarrow \unit
\\
\cntc' = & {\tt Foo!foo} \, [\cog : c] (\unit, x:[\cog : c']) \rightarrow \unit
\end{array}
\]
}
\caption{\label{fig.foocontract} Method contract of {\tt Foo.foo}.}
\end{figure*}
That is, the contract of {\tt Foo.foo} displays \emph{two} recursive invocations because,
in correspondence of the {\tt await v?} instruction, we need to collect all the effects
produced by the previous unsynchronised asynchronous invocations (see rule 
\rulename{T-Await})~\footnote{
It is possible to define sufficient conditions on {\coreABS} programs that entail 
linear recursive contract class tables. For example, two such conditions are
that, in (mutual) recursive methods, recursive invocations are either (i) synchronous  
or (ii) asynchronous followed by a \Abs{get} or \Abs{await} synchronisation on the future value, without any other \Abs{get} or \Abs{await} synchronisation or synchronous 
invocation in between.}.

\subsection{\em Mutations and flashbacks}

The idea of our technique is to consider the patterns of cog names in the formal parameters 
and the (at most unique) recursive invocation of
method contracts and to study the changes. For example, the above method contracts of
{\tt F.h} and {\tt F.l} transform the pattern of cog names in the formal parameters, 
written $(c)$ into the pattern of recursive invocation $(c')$. We write this
transformation as
\[
(c) \; \leadsto \; (c') \; .
\]
In general, the transformations we consider are called \emph{mutations}.

\begin{definition}
A \emph{mutation} is a transformation of tuples of (cog) names, written
\[
(x_1, \cdots , x_n) \; \leadsto \; (x_1', \cdots , x_n') 
\]
where $x_1, \cdots , x_n$ are pairwise different and $x_i'$ \emph{ may not occur in} $\{ x_1, \cdots , x_n \}$. 

Applying a mutation $(x_1, \cdots , x_n) \; \leadsto \; (x_1', \cdots , x_n')$
to a tuple of cog names (that may contain duplications) 
$(c_1, \cdots , c_n)$ gives a tuple $(c_1', \cdots , c_n')$ where
\begin{itemize}
\item[--] 
$c_i' = c_j$ if $x_i' = x_j$;
\item[--]
$c_i'$ is a fresh name if $x_i' \not \in \{ x_1, \cdots , x_n \}$;
\item
$c_i' = c_j'$ if they are both fresh and $x_i' = x_j'$.
\end{itemize}
We write $(c_1, \cdots , c_n) \redmut (c_1', \cdots , c_n')$ when 
$(c_1', \cdots , c_n')$ is obtained by applying a mutation (which is kept implicit)
to $(c_1, \cdots , c_n)$.
\end{definition}

For example, given the mutation 
\begin{eqnarray}
\label{mut.uno}
(x,y,z,u) \leadsto (y,x,z',z')
\end{eqnarray}
we obtain the following sequence of tuples:
%\[
\begin{eqnarray}
\label{mut.red}
(c,c',c'',c''') \; \redmut & (c',c,c_1,c_1) 
\\
\redmut & (c,c',c_2,c_2) \nonumber
\\
\redmut & (c',c,c_3,c_3) \nonumber
\\
\redmut & \cdots \nonumber
\end{eqnarray}
%\]
When a mutation $(x_1, \cdots , x_n ) \leadsto (x_1', \cdots , x_n')$ is such that
$\{ \x_1,  \cdots , x_n\} = \{ x_1', \cdots , x_n'\}$ then 
the mutation is a \emph{permutation}~\cite{Comtet}. In this case, 
the permutation theory guarantees that, by repeatedly applying the same permutation 
to a tuple of names, at some point, one obtains the initial tuple. 
This point, which is known as the \emph{order of the permutation}, allows one to define the following algorithm for linear recursive method contracts whose mutation is a permutation:
\begin{enumerate}
\item
compute the order of the permutation associated to the recursive method contract and
\item
correspondingly unfold the term to evaluate.
\end{enumerate}
It is clear that, when method contract bodies have no free cog names,
 further unfoldings of the recursive method contract cannot add new dependencies. Therefore the evaluation, as far as dependencies are concerned, may stop.

When a mutation is not a permutation, as in the example above, it is not possible
to get again an old tuple by applying the mutation because of the presence of fresh names.
However, it is possible to demonstrate that a tuple is equal to an old one \emph{up-to} a
suitable map, called flashback.

\begin{definition}
A tuple of cog names $(c_1, \cdots , c_n)$ is equivalent to $(c_1', \cdots, c_n')$, written
$(c_1, \cdots , c_n) \approx (c_1', \cdots,$ $ c_n')$, if there
is an injection $\imath$ called \emph{flashback} such that:
\begin{enumerate}
\item
$(c_1, \cdots, c_n) = (\imath( c_1'), \cdots, \imath(c_n'))$

\item
$\imath$ is the identity on ``old names'', that is,  if $c_i' \in  
\{ c_1, \cdots, c_n \}$ then $\imath(c_i') = c_i'$.
\end{enumerate}
\end{definition}
For example, in the sequence of transitions~(\ref{mut.red}),
%\[
%(c,c',c'',c''') \; \redmut \; (c',c,c_1,c_1)
%\; \redmut \; (c,c',c_2,c_2)
%\; \redmut \;  (c',c,c_3,c_3)
%\]
there is a flashback from the last tuple to the second one and there is  
(and there will be, by applying the mutation (\ref{mut.uno})) no tuple that is equivalent to 
the initial tuple.
  
It is possible to generalize the result about permutation orders:

\begin{theorem}[\cite{GL2014}]
\label{thm.mainthm}
Let $(x_1, \cdots, x_n) \leadsto (x_1', \cdots, x_n')$ be a mutation
and let 
{\small \[
(c_1, \cdots , c_n) \redmut (c_{n+1}, \cdots , c_{2n})
\redmut (c_{2n+1}, \cdots , c_{3n}) \redmut \cdots 
\]}
be a sequence of applications of the mutation. Then there are  $0 \leq h < k$ such that 
\[
(c_{hn+1}, \cdots , c_{(h+1)n}) \approx (c_{kn+1}, \cdots , c_{(k+1)n})
\]
The value $k$ is called \emph{order of the mutation}.
\end{theorem}

For example, the order of the mutation (\ref{mut.uno}) is 3.

\subsection{Evaluation of the main contract pair}

The generalisation of permutation theory in Theorem~\ref{thm.mainthm} allows us 
to define the notion of \emph{order of the contract of the main function} in a 
linear recursive contract class table. This order is the length of the evaluation of the 
contract obtained
\begin{enumerate}
\item 
by unfolding every recursive function as many times as  
\emph{twice its ordering}~\footnote{The interested reader may find in~\cite{GL2014}
the technical reason for unfolding recursive methods as many times as twice the length of
the order of the corresponding mutation.};
\item
by iteratively applying 1 to every invocation of recursive function that has been 
produced during the unfolding.
\end{enumerate}

In order to state the theorem of the correctness of our analysis technique, we
need to define the lam of a contract pair. The following functions
will do the job.

Let  $\sem{\cdot}$ be a map taking a contract pair and returning a pair of lams
that is defined by
\[
\begin{array}{rl}
\sem{\C.\m \;\frr(\bar{\frr})\rightarrow \frr'} \, = & \pairl{\pinull}{\pinull}
\\
\sem{\C!\m \;\frr(\bar{\frr})\rightarrow \frr'} \, = & \pairl{\pinull}{\pinull}
\\
\sem{\C!\m \;\frr(\bar{\frr})\rightarrow \frr'. (c,c')^{[\aa]}} \, = & 
\pairl{\slam{(c,c')^{[\aa]}}}{\pinull}
\\
\sem{\pairl{\cntcp}{\cntcp'}_c} \; = & \pairl{\sem{\cntcp}}{\sem{\cntcp'}}
\end{array}
\]
and it is homomorphic with respect to the operations $+$, $\fatsemi$, %$\parallel$,
$\rfloor$ (whose definition on pairs of lams is in Figure~\ref{fig:lamOp}).
Let $\mathbbm{t}$ be terms of the following syntax
\[
\mathbbm{t} \; ::= \quad \lafsa \quad | \quad \pairl{\mathbbm{t}}{\mathbbm{t}}
\]
and let $(\lafsa)^\flat = \lafsa$ and $(\pairl{\mathbbm{t}}{\mathbbm{t}'})^\flat = 
(\mathbbm{t})^\flat , \, (\mathbbm{t}')^\flat$.

\begin{theorem}[\cite{GL2014}]
\label{thm.invariance}
Let $\pairl{\cntcp_1}{\cntcp_1'}$ be a main function contract and let 
{\footnotesize\[
\begin{array}{@{\!\!}l}
\pairl{\cntcp_1}{\cntcp_1'}_{\rm start} \lred{} \pairl{\cntcp_2}{\cntcp_2'}_{\rm start}
\lred{} \pairl{\cntcp_3}{\cntcp_3'}_{\rm start} \lred{} \cdots
\end{array}
\]}
be its evaluation.
Then there is a $k$, which is the order of $\pairl{\cntcp_1}{\cntcp_1'}_{\rm start}$
 such that if a circularity occurs in 
$(\sem{\pairl{\cntcp_{k+h}}{\cntcp_{k+h}'}_{\rm start}})^\flat$, for every $h$, then it also occurs in
$(\sem{\pairl{\cntcp_{k}}{\cntcp_{k}'}_{\rm start}})^\flat$.
\end{theorem}

\begin{example}
The reduction of the contract of method \name{Math.fact\_nc} is as in Figure~\ref{fig:ex.red.fact}.
The theory of mutations provide us with an order for this evaluation.
In particular, the mutation associated to \name{Math.fact\_nc} is $c \leadsto c'$, with order 1, such that after one step we can encounter a flashback to a previous state of the mutation.
Therefore, we need to reduce our contract for a number of steps corresponding to twice the ordering of  \name{Math.fact\_nc}:
after two steps we find the flashback associating the last generated pair $(c',c'')$ to the one produced in the previous step $(c,c')$, by mapping $c'$ to $c$ and $c''$ to $c'$.

The flattening and the evaluation of the resulting contract are shown in Figure~\ref{fig:ex.eval.mut} and produce the pair of lams $ \pairl{ \slam{(c',c'') ,(c,c')}} {\pinull}$ which does not present any deadlock. Thus, differently from the fixpoint analysis for the same example, with this operational analysis we get a precise answer instead of a false positive. (See Figure~\ref{fig:satEx} and Section~\ref{sec.deadlock_analysis_lams}.)

\begin{figure*} 
\centering
  %{\scriptsize
    $\begin{array}{l} \pairl{{\tt
          Math!fact\_nc} \, [\cog : c]( \unit ) \rightarrow
        \unit}{\pinull}_{\rm start}
      \\
      \lred{} \quad
      \pairl{ \pairl{\pinull} {\pinull + {\tt Math!fact\_nc}[\cog :
          c'] (\unit)\rightarrow \unit \seqpoint(c,c') }_c}
      {\pinull}_{\rm start}
      \\
      \lred{} \quad
      \pairl{ \pairl{\pinull} {\pinull + \pairl{\pinull} {\pinull +
            {\tt Math!fact\_nc}[\cog : c''] (\unit)\rightarrow \unit
            \seqpoint(c',c'') }_{c'} \seqpoint(c,c') }_c}
      {\pinull}_{\rm start}
      % \\
      % \lred{}\quad
      % \pairl{ \pairl{\pinull} {\pinull + \pairl{\pinull} {\pinull +
      %       \pairl{\pinull} {\pinull + {\tt Math!fact\_nc}[\cog :
      %         c'''] (\unit)\rightarrow \unit \seqpoint(c'',c''')
      %       }_{c''} \seqpoint(c',c'') }_{c'} \seqpoint(c,c') }_c}
      % {\pinull}_{\rm start}

      % \\
      % \lred{}\quad
      % \pairl{ \pairl{\pinull} {\pinull + \pairl{\pinull} {\pinull +
      %       \pairl{\pinull} {\pinull + \pairl{\pinull} {\pinull + {\tt Math!fact\_nc}[\cog :
      %         c''''] (\unit)\rightarrow \unit \seqpoint(c''',c'''')
      %       }_{c'''}  \seqpoint(c'',c''')
      %       }_{c''} \seqpoint(c',c'') }_{c'} \seqpoint(c,c') }_c}
      % {\pinull}_{\rm start}
    \end{array}$
%}
  \caption{Reduction for contract of method \name{Math.fact\_nc}.}
\label{fig:ex.red.fact}
\end{figure*}

\begin{figure*}
  \centering
$\begin{array}{l}
 (\sem{\pairl{ \pairl{\pinull} {\pinull + \pairl{\pinull} {\pinull +
            {\tt Math!fact\_nc}[\cog : c''] (\unit)\rightarrow \unit
            \seqpoint(c',c'') }_{c'} \seqpoint(c,c') }_c}
      {\pinull}_{\rm start}})^\flat\\
=
 (\pairl{ \pairl{\pinull} {\pinull + \pairl{\pinull} {\pinull +
            \pairl{\slam{(c',c'')}}{\pinull} } \seqpoint(c,c') }}
      {\pinull})^\flat\\
=
 \pairl{\pinull +  \pinull + \slam{(c',c'')}\addition (c,c')} {\pinull}\\
=
 \pairl{ \slam{(c',c'') ,(c,c')}} {\pinull}\\
  \end{array}$
% $\begin{array}{l}
%  (\sem{\pairl{ \pairl{\pinull} {\pinull + \pairl{\pinull} {\pinull +
%             \pairl{\pinull} {\pinull + \pairl{\pinull} {\pinull + {\tt Math!fact\_nc}[\cog :
%               c''''] (\unit)\rightarrow \unit \seqpoint(c''',c'''')
%             }_{c'''}  \seqpoint(c'',c''')
%             }_{c''} \seqpoint(c',c'') }_{c'} \seqpoint(c,c') }_c}
%       {\pinull}_{\rm start}})^\flat\\
% =
%  (\pairl{ \pairl{\pinull} {\pinull + \pairl{\pinull} {\pinull +
%             \pairl{\pinull} {\pinull + \pairl{\pinull} {\pinull + \pairl{\slam{(c''',c'''')}}{\pinull}
%             }  \seqpoint(c'',c''')
%             }\seqpoint(c',c'') } \seqpoint(c,c') }}
%       {\pinull})^\flat\\
% =
%  \pairl{  {\pinull + {\pinull +
%             {\pinull +  {\pinull + {\slam{(c''',c'''')}}
%             }   \addition (c'',c''')
%             } \addition (c',c'') }  \addition (c,c') }}
%       {\pinull}\\
% =
%  \pairl{ \slam{(c''',c''''),(c'',c'''),(c',c'') ,(c,c')}} {\pinull}\\
%   \end{array}$
  \caption{Flattening and evaluation of resulting contract of method \name{Math.fact\_nc}.}
\label{fig:ex.eval.mut}
\end{figure*}

\end{example}

The correctness of the technique based on mutations is demonstrated in 
Appendix~\ref{sec:mutanalysis-proof}.

%%% Local Variables:
%%% mode: latex
%%% TeX-master: "SoSyM"
%%% End:

 \section{The {\DFfABS} tool and its application to the case study} \label{sec:sdatool}
 %!TEX root = SoSyM.tex

{\coreABS} (actually full {\ABS}~\cite{johnsen10fmco})
comes with a suite~\cite{WongAMPSS12} that offers
 a compilation framework,
 a set of tools to analyse the code,
 an Eclipse IDE plugin and Emacs mode for the language.
We extended this suite with an implementation of our deadlock analysis
framework (at the time of writing the suite has only the fixpoint analyser, the 
full framework is available at \url{http://df4abs.nws.cs.unibo.it}).
The {\DFfABS} tool is built upon the abstract syntax tree (AST) of the 
{\coreABS} type checker, which allows us to exploit the type information stored in 
every node of the tree. This simplifies the implementation of several contract inference rules.
The are four main modules that comprise {\DFfABS}:

\medskip

(1) \emph{Contract and Constraint Generation}. This is performed in three steps:
(i) the tool first parses the classes of the program and generates a map 
between interfaces and classes, required for the contract inference of method calls;
(ii) then it parses again all classes of the program to generate the initial environment $\Gamma$ that maps methods to the corresponding method signatures;
 and (iii) it finally parses the AST and, at each node, it applies 
 the contract inference 
rules in Figures~\ref{fig:inf:exp},~\ref{fig:inf:stmt}, and
 \ref{fig.meth-class}.

\medskip

(2) \emph{Constraint Solving} is done by a generic semi-unification solver implemented in Java, following the algorithm defined in~\cite{Henglein:1993}.
%The implementation of that solver is available at
%\url{http://proton.inrialpes.fr/~mlienhar/semi-unification}.
When the solver terminates (and no error is found), it produces a substitution that 
satisfies the input constraints.
Applying this substitution to the generated contracts produces the abstract class table and the contract of the main function of the program.
%
%is constraint-based and works in three steps:
%\begin{enumerate}
%\item
%given a program, it generates \emph{method contract skeletons} and 
%\emph{method interface skeletons} for each method and interface of the program, plus 
%a set of constraints that the contracts must satisfy;
%\item
%the constraints are then solved by a dedicated solver, that generate a substitution validating them;
%\item
%the substitution is applied to the contract skeletons, generating the actual contracts of the methods.
%\end{enumerate}
%
%The

\medskip

(3)
\emph{Fixpoint Analysis} uses dynamic structures to store lams of every method
contract (because lams become larger and larger as the
analysis progresses). At each iteration of the analysis, a number of fresh cog names
is created and the states are updated according to what is prescribed by the
contract.
%A basic operation of the analyser is the renaming, which is used when computing every
%approximant. 
%
At each iteration, the tool checks whether a fixpoint has been reached.
Saturation starts when the number of iterations reaches a maximum value
 (that may be customised by the user).
In this case, since the precision of the algorithm degrades, the 
tool signals that the answer may be imprecise.
To detect whether a relation in the fixpoint lam contains a circular 
dependency, we run Tarjan algorithm~\cite{Tarjan72} for connected components of graphs and we stop
the algorithm when a circularity is found.

\medskip

(4)
\emph{Abstract model checking}
algorithm for 
deciding the circularity-freedom problem in linear recursive contract class tables
 performs the following steps.
(\emph{i}) \emph{Find (linear) recursive methods}: by parsing the 
contract class table we create
a graph where nodes are function names and, for every invocation of {\tt D.n} in the 
body of {\tt C.m},  there is an edge from {\tt C.m} to {\tt D.n}. Then a standard depth first search
associates to every node a path of (mutual) recursive invocations (the paths starting and ending at 
that node, if any). The contract class table is linear recursive if every node has at most one associated path.
(\emph{ii}) \emph{Computation of the orders}: given the list of recursive 
methods, we compute the corresponding mutations.
(\emph{iii}) \emph{Evaluation process}: the contract pair corresponding to the 
main function is evaluated till every recursive function invocation has been 
unfolded up-to twice the corresponding order. 
(\emph{iv}) \emph{Detection of circularities}: this is performed with the
same algorithm of the fixpoint analysis.

\medskip

As regards the computational complexity, the contract inference system is polynomial time with respect to the length of the program in most of the 
cases~\cite{Henglein:1993}. The fixpoint analysis is 
is exponential in the number of cog names in a contract class table (because lams
may double the size at every iteration). However, this exponential effect actually
bites in practice.
%
% in a program. This is the reason
%why, in the above table, increasing the number of iterations (from {\tt 2} to {\tt 3}) 
%causes the runtime to increase by a factor of 6. We
%remark that in most cases, the precision of the SDA tool does not
%enhance at iterations higher than {\tt 1}.
%
%
%linear up-to the saturation 
%point (with respect to the length of the program) 
%and it has polynomial cost for reaching a saturated fixpoint.
The abstract model checking is linear with respect to the length of the program
as far as steps (i) and (ii) are concerned. Step (iv) is linear 
with respect to the size of the final lam. The critical step is (iii), which may 
be exponential with respect to the length of the program. Below, there is an
overestimation of the computational complexity. 
 Let 
\begin{description}
\item[$\ordermut{{\it max}}$] be
the largest order of a recursive method contract (without loss of generality,
we assume there is no mutual recursion). 
\item[$m_{\it max}$] be the maximal number of function invocations
in a body or in the contract of the
main function. 
\end{description}
An upper bound to the length of the evaluation till the saturated state is
$$
\sum_{0 \leq i \leq \ell}(2 \times \ordermut{{\it max}} \times m_{\it max})^i ,
$$
where $\ell$ is the number of methods in the program.
Let $k_{\it max}$ be the maximal number of dependency pairs in a body. Then
the size of the saturated state is $O(k_{\it max} \times 
{(\ordermut{{\it max}} \times m_{\it max})^\ell})$, which is also the computational
complexity of the abstract model checking.

\subsection{\em Assessments}

We tested {\DFfABS} on a number of 
medium-size programs written for
benchmarking purposes by {\coreABS} programmers and on 
an industrial case study based on the Fredhopper Access Server (FAS)%
\footnote{Actually, the FAS module has been written in {\ABS}~\cite{FAS}, %
 and so, we had to adapt it in order to conform with {\coreABS} restrictions (see Section~\ref{sec.restrictions}). %
This adaptation just consisted of purely syntactic changes, and only took half-day work (see also the comments in~\cite{GL2013a}).}
 developed by SDL Fredhopper~\cite{FAS}, which provides search and merchandising IT services to e-Commerce companies.
The (leftmost two columns of the) table in Figure~\ref{fig:assessments} 
reports the experiments: for every program we display the
number of lines, whether the analysis has reported a deadlock ({\tt D}) or not
($\checkmark$), the time in seconds required for the analysis.
Concerning time, we only report the time of the analysis of {\DFfABS} (and not 
the one taken by the inference) 
when they run on a QuadCore 2.4GHz and Gentoo (Kernel 3.4.9).

\begin{figure*}
{\small
\begin{center}
\begin{tabular}{||l||c|@{\quad}c@{\quad}|@{\quad}c@{\quad}|@{\quad}c@{\quad}||}
\hline
\quad program \quad & lines &  \begin{tabular}{c} {\DFfABS}/fixpoint \\ result \ time 
\end{tabular} 
& \begin{tabular}{c} {\DFfABS}/model-check \\ result \ time \end{tabular} & \begin{tabular}{c} {\tt DECO} \\
result \ time \end{tabular}
\\ \hline
{\tt PingPong} & 61 &  \checkmark \quad 0.311 & \checkmark \quad 0.046 & \checkmark \quad 1.30 
\\ \hline
{\tt MultiPingPong} & 88  & {\tt D} \quad 0.209 & {\tt D} \quad 0.109 & {\tt D} \quad 1.43 
\\ \hline
{\tt BoundedBuffer}  \quad &  103  & \checkmark \quad 0.126  & \checkmark \quad 0.353 & \checkmark \quad 1.26
\\ \hline
{\tt PeerToPeer}    &   185 & \checkmark \quad 0.320 & \checkmark \quad 6.070 & \checkmark \quad 1.63
\\ \hline \hline
{\tt FAS Module} & 2645 & \checkmark \quad 31.88 & \checkmark \quad 39.78 & \checkmark \quad 4.38
\\
\hline
\end{tabular}
\end{center}
}
\caption{\label{fig:assessments} 
Assessments of {\DFfABS}.}
\end{figure*}

The rightmost column of the table in Figure~\ref{fig:assessments} reports the results of another
tool that has also been developed for the deadlock analysis of {\coreABS} 
programs: {\tt DECO}~\cite{Antonio2013}. 
This technique integrates a point-to analysis with 
an analysis returning (an over-approximation of) program points that may be running 
in parallel. 
As highlighted by the above table, the three tools return the results as regards deadlock analysis,
 but are different as regards performance.
In particular the fixpoint and model-checking analysis of {\DFfABS} are comparable on small/mid-size programs,
 {\tt DECO} appears less performant (except for {\tt PeerToPeer}, where our model-checking analysis is quite slow
 because of the number of dependencies produced by the underlying algorithm). 
On the {\tt FAS module}, our two analysis are again comparable, while {\tt DECO} has a better performance
 ({\tt DECO} worst case complexity is cubic in the size of the input).

\smallskip

Few remarks  about the precision of the techniques follow.
{\DFfABS}/model-check is the most powerful tool we are aware of for linear recursive contract class table.
For instance, it correctly detect the deadlock-freedom of the method {\tt Math.fact\_nc} (previously defined in Figure~\ref{fig.Math})
 while {\DFfABS}/fixpoint signals a false positive.
Similarly, {\tt DECO} signals a false positive deadlock for the following 
program, whereas {\DFfABS}/model-check returns its deadlock-freedom. 
\begin{absexamplen}
class C implements C {
	Unit m(C c){ C w ;
		w = new cog C() ;
		w!m(this) ;
		c!n(this) ; 
        }
	Unit n(C a){  Fut<Unit> x ; 	
		x = a!q() ;
        x.get ;
        }
	Unit q(){ }
}
{		C a; C b ;
   Fut<Unit> x ;
   a = new cog C() ;
   b = new cog C() ;
   x = a!m(b) ;
}
\end{absexamplen}

However, {\DFfABS}/model-check is not defined on non-linear recursive contract class tables.
Non-linear recursive contract class tables can easily be defined, as shown with the following two contracts:
$$\begin{array}{rl}
\C.\m \; = & [\cog : c]\ ()\ \{\pairl{\pinull}{(\C!\m \, [\cog: c]() \rightarrow \unit).(c,c')\\
&+\ \C!{\tt n} \, [\cog: c'']([\cog:c]) \rightarrow \unit}\} \rightarrow \unit
\\
\C.{\tt n} \; = & [\cog: c]\ ([\cog:c'])\ \\&\{ \pairl{ (\C!\m \, [\cog: c]() 
\rightarrow \unit).(c,c')}{\pinull} \} \rightarrow \unit
\end{array}$$
Here, {\DFfABS}/model-check  fails to analyse $\C.\m$ while {\DFfABS}/fixpoint and {\tt DECO}
successfully recognise as dead\-lock-free%
\footnote{In~\cite{GL2014}, we have defined a source-to-source transformation %
 taking nonlinear recursive contract class tables and returning linear recursive ones. %
This transformation introduces fake cog dependencies that returns a false positive %
 when applying {\DFfABS}/model-check on the example above.}.
We conclude this section with a remark about the proportion between 
programs with linear 
recursive contract class tables and those with nonlinear ones. 
While this proportion is hard to assess, our preliminary 
analyses strengthen the claim that nonlinear recursive programs are rare.
We have parsed the three case-studies developed in the
European project HATS~\cite{FAS}. The case studies are the FAS module, 
a Trading System (TS) modelling a supermarket handling sales, and a Virtual Office 
of the Future (VOF) where office workers are enabled to perform their office tasks seamlessly independent of their current location. 
FAS has 2645 code-lines, TS has 1238 code-lines, and VOF has 429 code-lines. In none of
them we found a nonlinear recursion in the corresponding contract class table, 
TS and VOF have respectively 2 and 3 
linear recursive method contracts (there are
recursions in functions on data-type values that have nothing to do with locks and
control). This substantiates the usefulness of our technique in these programs;
the analysis of a wider range of programs is matter of future work.

%%% Local Variables: 
%%% mode: latex
%%% TeX-master: "subm-SoSyM"
%%% End: 

 \section{Related works}
 \label{sec:relatedworks}
 %!TEX root = SoSyM.tex

A preliminary theoretical study was undertaken in~\cite{GiachinoL11}, 
where (\emph{i}) the considered language is a functional subset of {\coreABS}; 
(\emph{ii}) contracts are not inferred, they are provided by the programmer and 
type-checked;
(\emph{iii}) the deadlock analysis is less precise because it is not iterated as in this 
contribution, but stops at the first approximant, and (\emph{iv}), more importantly,
method contracts are not pairs of lams, which led it to discard dependencies
(thereby causing the analysis, in some cases, to erroneously yield false negatives).
This system has been improved in~\cite{GL2013a} by modelling method contracts as
pairs of lams, thus supporting a more precise fixpoint technique.
The contract inference system of~\cite{GL2013a} has been extended in this contribution with the management of aliases of futures and with the dichotomy of present contract
and future contract in the inference rules of statements.

The proposals in the literature that statically analyse deadlocks 
are largely based on (behavioural) types. In~\cite{Boyapati2002,Flanagan03,Abadi2006,Vasconcelos2009}
a type system is defined that computes a partial order of 
the locks in a program and a subject reduction theorem demonstrates that 
tasks follow this order. Similarly  to these techniques, the tool 
{\tt Java PathFinder}~\cite{JavaPathFinder} computes a tree of lock orders for
 every method
and searches for mismatches between such orderings.
On the contrary, our technique does not
compute any ordering of locks during the inference of contracts, thus being more 
flexible: a computation may acquire two locks in different order at different 
stages, being correct in our case, but incorrect with the other techniques. 
The Extended Static Checking for Java~\cite{ESC} is an automatic tool for contract-based programming: annotation are used to specify loop invariants, pre and post conditions, and
to catch deadlocks. The tool warns the programmer if the annotations cannot be validated. 
This techniques requires that annotations are explicitly provided by the programmer, while
they are inferred in {\DFfABS}.

A well-known deadlock analyser is {\sc TyPiCal}, a tool that
has been developed for pi-calculus by Kobayashi~\cite{Typicaltool,Kobayashi1998,Kobayashi2004,Kobayashi06}. 
{\sc TyPiCal} uses a clever technique for deriving 
inter-channel dependency information and
is able to deal with several recursive behaviours and the creation of new 
channels without committing to any pre-defined order of channel names.
Nevertheless, since 
{\sc TyPiCal} is based on an inference system, there are recursive behaviours
that escape its accuracy.
For instance, it returns false positives when recursion create networks with arbitrary numbers of nodes. To illustrate the issue we consider the following deadlock-free 
program computing factorial 
\begin{absexamplen}
class Math implements Math {
   Int fact(Int n, Int r){
   	Math y ;
	Fut<Int> v ;
	if (n == 0) return r ; 
        else { y = new cog Math() ;
               v = y!fact(n-1, n*r) ;
               w = v.get ;
               return w ;
       }
   }
}
{
	Math x ; Fut<Int> fut ; Int r ;
	x = new cog Math();
	fut = x!fact(6,1);
	r = fut.get ;
}
\end{absexamplen}
that is a variation of the method {\tt Math.fact\_ng} in Figure~\ref{fig.Math}. This
code is deadlock free according to {\DFfABS}/ model-check, however, its implementation in
pi-calculus~\footnote{The pi-calculus factorial program is
\\
{\tt *factorial?(n,(r,s)).
\\
\hspace*{.5cm}  if n=0 then r?m. s!m else new t in 
\\
\hspace*{1cm}                        (r?m. t!(m*n)) | factorial!(n-1,(t,s)) 
}
\\
In this code, {\tt factorial} returns the value (on the channel {\tt s}) by 
\emph{delegating}
this task to the recursive invocation, if any. In particular,
the initial invocation of {\tt factorial}, which is {\tt r!1 | factorial!(n,(r,s))}, 
performs a synchronisation between {\tt r!1} and the input {\tt r?m} in 
the continuation of {\tt factorial?(n,(r,s))}. In turn, this may delegate 
the computation of the factorial to a subsequent synchronisation on a 
new channel {\tt t}.
{\sc TyPiCal} signals a deadlock on the two inputs {\tt r?m} because it fails in connecting
the output {\tt t!(m*n)} with them. 
} is not deadlock-free according to {\sc TyPiCal}. The extension of {\sc TyPiCal}
with a technique similar to the one in Section~\ref{sec.mutanalysis}, but covering
the whole range of lam programs, has been recently defined in~\cite{GKL2014}.

Type-based deadlock analysis has also been 
studied in~\cite{PuntigamP01}. In this contribution, types define objects' states
and  can express acceptability of messages.
The exchange of messages modifies the state of the objects.
In this context, a deadlock is avoided by setting an ordering on types.
With respect to our technique, \cite{PuntigamP01} uses a deadlock prevention
approach, rather than detection, and no inference system for types is
provided.

In~\cite{PunPhD}, the author proposes two approaches for a type and effect-based deadlock analysis for a concurrent extension of ML.
The first approach, like our ones, uses a type and effect inference algorithm, followed by an analysis to verify deadlock freedom.
However, their analysis approximates infinite behaviours with a chaotic behaviour that
non-deterministically acquires and releases locks, thus becoming imprecise.
For instance, the previous example should be considered a potential deadlock in their approach.
The second approach is an initial result on a technique for reducing deadlock analysis to data race analysis.

Model-theoretic techniques for deadlock analysis have also been investigated.
defined. In~\cite{CM97}, circular dependencies among processes are detected as
erroneous configurations, but dynamic creation of names is not
treated. Similarly in~\cite{deBoerBGLSZ12} (see the discussion
below). 

Works that specifically tackle the problem of deadlocks for languages with the
same concurrency model as that of {\coreABS} are the following:
\cite{WestNM10} defines an approach for deadlock prevention (as opposed to
our deadlock detection) in SCOOP, an Eiffel-based concurrent language.
Different from our approach, they annotate classes with the used
\emph{processors} (the analogue of cogs in {\coreABS}), while this information is
inferred by our technique. Moreover each method
exposes preconditions representing required lock ordering of processors (processors obeys an order in
which to take locks), and this information must be provided by the
programmer.
\cite{deBoerBGLSZ12} studies a Petri net based analysis, reducing deadlock
detection to a reachability problem in Petri nets.
This technique is more precise in that it is thread based and not just object
based.
Since the model is finite, this contribution does not address the feature of
object creation and it is not clear how to scale the technique.
We plan to extend our analysis in order to consider
finer-grained thread dependencies instead of just object dependencies.
\cite{KerfootMT09} offers a design pattern methodology for CoJava to
obtain deadlock-free programs. CoJava, a {Java} dialect where data-races and
data-based deadlocks are avoided by the type system, prevents threads from
sharing mutable data. Deadlocks are excluded by  a programming style based on
ownership types and \emph{promise} (i.e. future) objects. The main differences
with our technique are (\emph{i}) the needed information must be provided by the
programmer, (\emph{ii}) deadlock freedom is obtained through ordering and
timeouts, and (\emph{iii}) no guarantee of deadlock freedom is provided by the
system.

The relations with the work by Flores-Montoya 
\emph{et al.}~\cite{Antonio2013} has been largely discussed in 
Section~\ref{sec:sdatool}. Here we remark that,
as regards the design, {\tt DECO} is a monolithic code written in Prolog. 
On the contrary,
{\DFfABS} is a highly modular  Java code. Every module may 
be replaced by another; for instance one may rewrite the inference system for another language
and plug it easily in the tool, or one may use a different/refined contract analysis 
algorithm, in particular one used in {\tt DECO} (see Conclusions).

 \section{Conclusions}
\label{sec:conclusion}
We have developed a framework for  detecting deadlocks in {\coreABS} programs.
 The technique uses 
(i) an inference algorithm to extract abstract descriptions of methods, called
contracts, (ii) an evaluator of contracts, which computes an over-approximated 
fixpoint semantics, (iii) a model checking algorithm that evaluates contracts by
unfolding method invocations. 

This study can be extended in several directions. 
As regards the prototype, the next release will provide indications about \emph{how} 
deadlocks have been produced by pointing out the elements in the code that
generated the detected circular dependencies.
This way, the programmer will be able to check whether or not the detected
circularities are actual deadlocks, fix the problem in case it is a verified
deadlock, or be assured that the program is deadlock-free.

{\DFfABS}, being modular, may be integrated with other
analysis techniques. In fact, in collaboration with Kobayashi~\cite{GKL2014},
we have recently defined a variant of the model checking 
algorithm that has no linearity restriction. 
For the same reason, another direction of research is to analyse contracts with 
the point-to analysis
technique of {\tt DECO}~\cite{Antonio2013}. We expect that such analyser will
be simpler than {\tt DECO} because, after all, contracts are simpler than
{\coreABS} programs. 

Another direction of research is the application of our inference system to other languages featuring asynchronous method invocation, possibly after removing or 
adapting or adding rules. One such language that we are currently studying is
 {\tt ASP}~\cite{ASP}. 
While we think that our framework and its underlying theory are
robust enough to support these applications, we observe that a necessary condition for demonstrating 
the results of correctness of the framework is that the language has a formal 
semantics.

 \bibliographystyle{spmpsci}      % mathematics and physical sciences
 \bibliography{main}   % name your BibTeX data base

\appendix

\section{Properties of Section~\ref{sec.FJg-contracts}}\label{sec:subjred}
%!TEX root = SoSyM.tex

\def\futures{F}
\def\rewritet{\Rightarrow}

The \emph{initial configuration} of a well-typed {\coreABS} program is
\[
{\it ob}({\it start}, \varepsilon, \{ [\text{destiny} \mapsto f_{\it start}, 
\bar{x} \mapsto \undefined] \, | \, s \}, \varnothing) \, \cog(\text{start},{\it start})
\]
where the activity $\{ [\text{destiny} \mapsto f_{\it start}, 
\bar{x} \mapsto \undefined] \, | \, s \}$ corresponds to the activation of the main function. 
A \emph{computation} is a sequence of reductions starting at the initial 
configuration according to the operational semantics. 
We show in this appendix that such computations keep configurations well-typed;
in particular, we show that the sequence of contracts corresponding to the 
configurations of the computations is in
 the {\em later-stage relationship} (see Figure~\ref{fig:sr:lr}).

\paragraph{Runtime contracts.}
In order to type the configurations we use a \emph{runtime type system}. To this aim
 we extend the syntax of contracts 
in Figure~\ref{fig:synCon} and define \emph{extended futures} $\futures$, \emph{extended 
contracts} that, with an abuse of notation, we still denote  $\cntc$ and
\emph{runtime contracts} $\concntc$ 
as follows:
\mysyntax{\mathmode}{
\simpleentry \futures [] f\bnfor \i_f [{}]
\\
\\
\simpleentry \cntc []
  {\mbox{{\em as in Figure~\ref{fig:synCon}}}}
   \bnfor f\bnfor f.(c,c') \bnfor f.(c,c')^{\aa} 
   \bnfor \mpair{c}{\cntc}{\cntc}
   [{}]
\\ \\
\simpleentry \concntc [] \pinull \bnfor \cproc{f}{c}{\cntc}{\cntc}\bnfor \cinvoc{f}{\ccall{\C!\m}{\frr}{\vect{\frr}}{\frr}} \bnfor \concntc \parallel \concntc
 [{}]
}
As regards $\futures$, they are introduced for distinguishing two kind of future names:
 i) $f$ that has been used in the contract inference system as a {\em static time} representation of a future,
  but is now used as its {\em runtime} representation;
 ii) $\i_f$ now replacing $f$ in its role of {\em static time} future
 (it's typically used to reference a future that isn't created yet).

As regards $\cntc$ and $ \concntc$, the extensions are motivated by the fact that, at runtime, the informations 
about contracts are scattered in all the configuration.
However, when we plug all the parts to type the whole configuration,
%we can merge the different informations to get back an extended contract $\cntc$
% (which turns out to be a parallel $\mpair{c_1}{\cntc_1}{\cntc_n'} \parallel \cdots \parallel \mpair{c_n}{\cntc_n}{\cntc_n'}$ where $c_i$ is the cog in which $\pair{\cntc_i}{\cntc_i'}$ executes).
we can merge the different informations to get a runtime contract $\rcntc'$ such that every contract $\cntc\in\rcntc'$ does not contain any reference to futures anymore.
This  merging is done using a set of rewriting rules $\rewritet$ defined in 
Figure~\ref{def.fsimpl} that let one replace the occurrences of runtime futures in runtime contracts $\concntc$ with the corresponding contract of the future.
We write $f \in \names{\concntc}$ whenever $f$ occurs in $\concntc$ not as an index.
The substitution $\concntc \subst{\cntc}{f}$ replaces the occurrences of $f$ in contracts $\cntc''$ 
of $\concntc$ (by definition of our configurations, in these cases 
$f$ can never occur as index in $\concntc$).
It is easy to demonstrate that the merging process always terminates and is confluent for non-recursive contracts and,
in the following, we let $\fsimpl{\concntc}$ be the {\em normal form} of 
$\concntc$ with respect to $\rewritet$:
\begin{definition}
A runtime contract $\rcntc$ is {\em non-recursive} if:
\begin{itemize}
\item all futures $f\in\names{\rcntc}$ are declared once in $\rcntc$
\item all futures $f\in\names{\rcntc}$ are not recursive, i.e. for all $\mf{f}{\mpair{c}{\cntc}{\cntc'}}\in\rcntc$, we have $f\not\in\names{\mf{f}{\mpair{c}{\cntc}{\cntc'}}}$
\end{itemize}
\end{definition}
\begin{figure*}
%\[
%\begin{array}{rcl@{\qquad \qquad}l}
%  \concntc \fpar \mf{f}{\mpair{c}{\cntc}{\cntc'}}
%  & \rewritet & 
%    \left\{ 
%    \begin{array}{l@{\qquad}l}
%    \concntc\subst{\mpair{c}{\cntc}{\cntc'}}{f} & {\rm if} \; f \in \names{\concntc}
% \\
% \\
%  \concntc \fpar  \mf{f}{\mpair{c}{\cntc}{\cntc'}}  & {\rm otherwise} 
%  \end{array} \right.
%\end{array}
%%
%\hspace{2.5em}
%%
%\begin{array}{rcl@{\qquad \qquad}l}
%  \concntc \fpar \cinvoc{f}{\ccall{\C!\m}{\frr}{\vect{\frr}}{\frr}}
%  &\rewritet& 
%  \left\{ 
%     \begin{array}{l@{\qquad}l} 
%     \concntc\subst{\ccall{\C!\m}{\frr}{\vect{\frr}}{\frr}}{f}
%  & {\rm if} \; f \in \names{\concntc}
% \\
% \\
%  \concntc \fpar \cinvoc{f}{\ccall{\C!\m}{\frr}{\vect{\frr}}{\frr}}
%  & {\rm otherwise}
%  \end{array} \right.   
%\end{array}
%\]
\myrules{\mathmode\scriptmode\simplemode}{
\entry{f \in \names{\concntc}}{\concntc \fpar \mf{f}{\mpair{c}{\cntc}{\cntc'}} \rewritet \concntc\subst{\mpair{c}{\cntc}{\cntc'}}{f}} \and
\entry{f \in \names{\concntc}}{\concntc \fpar \cinvoc{f}{\ccall{\C!\m}{\frr}{\vect{\frr}}{\frr}} \rewritet \concntc\subst{\ccall{\C!\m}{\frr}{\vect{\frr}}{\frr}}{f}}
}
\caption{\label{def.fsimpl} Definition of $\rewritet$}
\end{figure*}

\paragraph{Typing Runtime Configurations.}
The typing rules for the runtime configuration are given in Figures~\ref{fig:rntyp}, \ref{fig:inf:expR} and~\ref{fig:inf:stmtR}.
Except for few rules (in particular, those in Figure~\ref{fig:rntyp} which type the runtime objects of a configuration),
 all the typing rules have a corresponding one in the contract inference system
 defined in Section~\ref{sec.FJg-contracts}.
Additionally, the typing judgments are identical to the corresponding one in the
inference system, with three minor differences:
\begin{itemize}
\item[ i)] the typing environment, that now contains a reference to the contract 
class table and mappings object names to pairs $(\C,\frr)$, is called $\Delta$;
\item[ ii)] the typing rules do not collect constraints;
%\item[ iii)] the $\rtcontract{\cdot}$ function on environments $\Delta$ is similar to $\contract{\cdot}$ in Section~\ref{sec.FJg-contracts},
% except that it now grabs values all futures $\i_f$ and $f$ from $\Delta$. More precisely
%\[\begin{array}{rl}
%\rtcontract{\Delta} \eqdef & \cntc_1 \parallel \cdots \parallel \cntc_n\parallel f_1\parallel \cdots\parallel f_m
%\end{array}\]
%where $\{ \cntc_1, \cdots , \cntc_n\} = \{ \cntc' \sht \exists\i_f,\frr : \; \Delta(\i_f) = (\frr,\cntc') \}$
% and $\{f_1,\cdots ,f_m\} = \{f'\sht f'\in\dom(\Delta)\}$.
\item[ iii)] the $\rtcontract{\cdot}$ function on environments $\Delta$ is similar to $\contract{\cdot}$ in Section~\ref{sec.FJg-contracts},
 except that it now grabs all $\i_f$ and all futures $f'$ that was created by the current thread $f$. More precisely
\[\begin{array}{rl}
\rtcontract{\Delta,f} \eqdef & \cntc_1 \parallel \cdots \parallel \cntc_n\parallel f_1\parallel \cdots\parallel f_m
\end{array}\]
where $\{ \cntc_1, \cdots , \cntc_n\} = \{ \cntc' \sht \exists\i_f,\frr : \; \Delta(\i_f) = (\frr,\cntc') \}$
 and $\{f_1,\dots,f_m\}=\{f'\sht\Delta(f')=(\rec',f)\}$.
\end{itemize}

\medskip

\noindent
Finally, few remarks about the auxiliary functions:
\begin{itemize}
\item[--]
$\textit{init}(\C,o)$ is supposed to return the init activity of the class
 \C. However, we have assumed that these activity is always empty, see Footnote~\ref{footnote.new}. Therefore the corresponding contract will be $\pairl{\pinull}{\pinull}$.

\item[--]
 $\textit{atts}(\C,\many{v}, o, c)$ returns a substitution
provided that $\many{v}$ have records $\many{\frr}$ and $o$ and $c$ are object and cog identifiers, respectively.

\item[--] $\textit{bind}(o,f,\m,\many{v'},\C)$ returns the activity corresponding 
to the method $\name{C}.\m$
  with the parameters $\many{v'}$ provided that
   $f$ has type $\fRec{c}{\frr}$ and $\many{v'}$ have the types $\many{\frr'}$.
%
%\item[--]
%$\ff(s)$ (\emph{runtime future names}) returns the set of runtime future names occurring in the statement $s$.
%This function is obviously not necessary for typing programs at static time.
\end{itemize}

\begin{figure*}[t]
\myrules{\mathmode\compactmode}{
\entry[{\scriptsize (TR-Future-Tick)}]{\Delta(f)=(\fRec{c}{\frr},\cntc)^{\checkmark}\\
\Delta \vdash \nt{val}: \frr}
  {\Delta\vdash_R \textit{fut}(f,\nt{val})\typed\pinull}
\and
\entry[{\scriptsize (TR-Future)}]{\Delta(f)=(\fRec{c}{\frr},\cntc) }
  {\Delta\vdash_R \textit{fut}(f,\bot)\typed\pinull}
\and
\entry[{\scriptsize (TR-Invoc)}]{\Delta(f)=(\fRec{c}{\frr'},\cntc)\\\\%\Delta(f)=(\fRec{c}{\frr'},\,\C!\m~[\cog: c, \many{x{:}\frr}](\many{\frs})\rightarrow\frr') \\\\
    \Delta\vdash_R\many{v}=\many{\frr} \\  \Delta(o)=[\cog: c, \many{x{:}\frr}]}
  {\Delta\vdash_R \textit{invoc}(o,f,m,\many{v})\typed\cinvoc{f}{\ccall{\C!\m}{[\cog: c, \many{x{:}\frr}]}{\many{\frs}}{\frr'}}}
\and
\entry[{\scriptsize (TR-Object)}]{\Delta(o)=[\cog: c, \many{x{:}\frr}]\\\Delta\vdash_R^{c,o} \many{\nt{val}}\typed \many{\frr}\\\\
    \Delta \vdash^{c,o}_R p:\concntc\\
    \Delta \vdash^{c,o}_R \bar{p}\typed\bar{\concntc}}
  {\Delta\vdash_R ob(o, [\cog\mapsto c;\many{x \mapsto\nt{val}}], p, \bar{p}):\concntc\fpar\bar{\concntc}}
\and
\entry[{\scriptsize (TR-Process)}]{
	\Delta\vdash_R^{c,o} \many{\nt{val} : \frx} 
	\\
	\Delta(f)=(\fRec{c}{\frr'},\vect{f'})^{[\checkmark]}
	\\\\
	\Delta[\text{destiny} \mapsto f, \many{x\mapsto\frx}]\vdash^{c,o}_R \nt{s}\typed\cntc\;|\;	\Delta''}
  {\Delta\vdash^{c,o}_R \{\text{destiny} \mapsto f, \many{\ x \mapsto \nt{val}} \; |\,\nt{s}\}: \cproc{f}{c}{\cntc}{\rtcontract{\Delta'',f}}}
\and
%\textcolor{red}{\entry[{\scriptsize (TR-Process)}]{
%	\Delta\vdash_R^{c,o} o : \frr
%	\\
%	\Delta\vdash_R^{c,o} \many{\nt{val} : \frx} 
%	\\
%	\Delta(f)=(\fRec{c}{\frr'},\cntc')^{[\checkmark]}
%	\\\\
%	\Delta'=\Delta|_{\ff(\nt{s})}
%	\\
%	\Delta'[o \mapsto \frr, \text{destiny} \mapsto f, \many{x\mapsto\frx}]
%	\vdash^{c,o}_R \nt{s}\typed\cntc\;|\;	\Delta''}
%  {\Delta\vdash^{c,o}_R \{\text{destiny} \mapsto f, \many{\ x \mapsto \nt{val}} \; |\,\nt{s}\}: \cproc{f}{c}{\cntc}{\rtcontract{\Delta''}}}
%  }
\and
\simpleentry[{\scriptsize (TR-Idle)}]{\Delta\vdash^{c,o}_R \key{idle}\typed\pinull}
\and
\entry[{\scriptsize(TR-Parallel)}]{\Delta\vdash_R cn_1\typed\concntc_1\\\Delta\vdash_R cn_2\typed\concntc_2}{\Delta\vdash_R cn_1\ cn_2\typed\concntc_1\parallel\concntc_2}
}
\caption{The typing rules for runtime configurations.}
\label{fig:rntyp}
\end{figure*}

\begin{figure*}[t]
runtime expressions
\myrules{\mathmode\compactmode}{
\entry[{\scriptsize (TR-Obj)}]{\Delta(o) = (\C,\frr)}{\inferpeR{\Delta}{\obj,o}{o}{}{\frr}} \and
\entry[{\scriptsize (TR-Fut)}]{\Delta(\futures)=\frz}{\inferpeR{\Delta}{\obj,o}{\futures}{}{\frz}} \and
\entry[{\scriptsize (TR-Var)}]{\Delta(x)=\frx}{\inferpeR{\Delta}{\obj,o}{x}{x}{\frx}} \and
\entry[{\scriptsize (TR-Field)}]{x\not\in\dom(\Delta)\\\Delta(o.x)=\frr}{\inferpeR{\Delta}{\obj,o}{\x}{x}{\frr}} \and
\entry[{\scriptsize (TR-Value)}]{\inferpeR{\Delta}{\obj,o}{e}{x}{\futures}\\
    \inferpeR{\Delta}{\obj,o}{\futures}{x}{(\frr,\cntc)^{[\checkmark]}}}{\inferpeR{\Delta}{\obj,o}{e}{x}{\frr}} \and
\entry[{\scriptsize (TR-Val)}]{e \quad  \textit{primitive value or arithmetic-and-bool-exp}}
  {\inferpeR{\Delta}{\obj,o}{e}{}{\unit}} \and
\entry[{\scriptsize (TR-Pure)}]{\inferpeR{\Delta}{\obj,o}{e}{}{\frr}}{\infereeR{\Delta}{\obj,o}{e}{}{\frr}{\pinull}{\Delta}}
}
expressions with side-effects
\myrules{\mathmode\compactmode}{
\entry[{\scriptsize (TR-Get)}]{\inferpeR{\Delta}{\obj,o}{x}{}{\i_f}\\\inferpeR{\Delta}{\obj,o}{\i_f}{}{(\fRec{\objb}{\frr'}, \cntc)}\\\\
    \Delta[{\it destiny}]=f\\\Delta'=\Delta[\i_f \mapsto (\frr, \pinull)^\checkmark]}
  {\infereeR{\Delta}{\obj,o}{x.\nget}{}{\frr'}{\cntc\seqpoint(\obj,\objb) \rfloor\rtcontract{\Delta',f}}{\Delta'}} 
\and
\entry[{\scriptsize (TR-Get-Runtime)}]{\inferpeR{\Delta}{\obj,o}{x}{}{f}\\\inferpeR{\Delta}{\obj,o}{f}{}{(\fRec{\objb}{\frr'}, \cntc)}\\\\
    \Delta[{\it destiny}]=f'\\\Delta'=\Delta[f \mapsto (\frr, \pinull)^\checkmark]}
  {\infereeR{\Delta}{\obj,o}{x.\nget}{}{\frr'}{f\seqpoint(\obj,\objb) \rfloor\rtcontract{\Delta',f'}}{\Delta'}} 
\and
\entry[{\scriptsize (TR-Get-tick)}]{\inferpeR{\Delta}{\obj,o}{x}{}{\futures}\\
    \inferpeR{\Delta}{\obj,o}{\futures}{}{(\fRec{\objb}{\frr'}, \cntc)^\checkmark}}
  {\infereeR{\Delta}{\obj,o}{x.\nget}{}{\frr'}{\pinull}{\Delta}}
\and
\entry[{\scriptsize (TR-NewCog)}]
  {\inferpeR{\Delta}{\obj,o}{\bar{e}}{\bar{n}}{\bar{\frr}}\\\parameters{\C}= \bar{T\ x}\\\fields{\C}=\bar{T'\ x'}\\\obj' \fresh}
  {\infereeR{\Delta}{\obj,o}{\nnew\ \ncog\ \C(\bar{e})}{\_}{[\cog{:}\obj', \bar{x{:}\frr}, \bar{x'}{:}\bar{\frr'}]}{\pinull}{\Delta}}
\and
\entry[{\scriptsize (TR-New)}]
  {\inferpeR{\Delta}{\obj,o}{\bar{e}}{\bar{n}}{\bar{\frr}}\\\parameters{\C}= \bar{T\ x}\\\fields{\C}=\bar{T'\ x'}}
  {\infereeR{\Delta}{\obj,o}{\nnew\ \C(\bar{e})}{\_}{[\cog{:}\obj, \bar{x{:}\frr}, \bar{x'}{:}\bar{\frr'}]}{\pinull}{\Delta}}
\and
\entry[{\scriptsize (TR-AInvk)}]{
    \inferpeR{\Delta}{\obj,o}{e}{n}{ [\cog {:} \obj' , \bar{x{:}\frr}] } \\ \fclass{\types{e}}=\C \\
    \inferpeR{\Delta}{\obj,o}{\bar{e}}{\bar{n}}{\bar{\frs}}\\ \fields{\C},\parameters{\C} = \bar{T\ x}\\
    \Delta(\C.\m)=\frr' (\bar{\frs'})\{\pairl{\cntc}{\cntc'}\} \frr''\\
    \bar{c'} = \gr{ \frr''}\setminus \gr{\frr',\bar{\frs'}}
    \\ \bar{c},\i_f \fresh\\
    \frs''=\frr''\subst{\bar{c}}{\bar{c'}}\subst{\frr,\bar{\frs}}{\frr',\bar{\frs'}}}
  {\infereeR{\Delta}{\obj,o}{e!\m(\bar{e})}{}{\i_f}{\pinull}{}
    {\Delta[\i_f\mapsto(\fRec{\objb}{\frs''}, \,\C!\m~\frr(\bar{\frs}) \rightarrow \frs'')]}} 
\and
\entry[{\scriptsize (TR-SInvk)}]{
    \inferpeR{\Delta}{\obj,o}{e}{n}{ [\cog {:} \obj' , \bar{x{:}\frr}]} \\ \fclass{\types{e}}=\C \\
    \inferpeR{\Delta}{\obj,o}{\bar{e}}{\bar{n}}{\bar{\frs}}\\ \fields{\C},\parameters{\C} = \bar{T\ x}\\
    \Delta(\C.\m)=\frr' (\bar{\frs'})\{\pairl{\cntc}{\cntc'}\} \frr''\\
    \bar{c'} = \gr{ \frr''}\setminus \gr{\frr',\bar{\frs'}}
    \\ \bar{c} \fresh\\
    \frs''=\frr''\subst{\bar{c}}{\bar{c'}}\subst{\frr,\bar{\frs}}{\frr',\bar{\frs'}}}
  {\infereeR{\Delta}{\obj,o}{e.\m(\bar{e})}{}{\frs''}{\C.\m~\frr(\bar{\frs}) \rightarrow \frs''\rfloor\rtcontract{\Delta}}{}{\Delta}} 
}
\caption{Runtime typing rules for expressions}
\label{fig:inf:expR}
\end{figure*}

\begin{figure*}[t]
statements
\myrules{\mathmode\compactmode}{
\entry[{\scriptsize (TR-Var-Record)}]{\inferpeR{\Delta}{\obj,o}{x}{n}{\frx}\\\infereeR{\Delta}{\obj,o}{z}{n}{\frx'}{\cntc}{\Delta'}}
  {\inferstR{\Delta}{\obj,o}{x=z}{\_}{\cntc}{}{\Delta'[x \mapsto\frx']}} 
\and
\entry[{\scriptsize (TR-Field-Record)}]
  {x \not\in\dom(\Delta)\\\Delta(\ethis.x)=\frr\\\infereeR{\Delta}{\obj,o}{z}{n}{\frr}{\cntc}{\Delta'}}
  {\inferstR{\Delta}{\obj,o}{\x=z}{\_}{\cntc}{}{\Delta'}}
\and
\entry[{\scriptsize (TR-Var-Future)}]{\inferpeR{\Delta}{\obj,o}{x}{n}{\futures}}
  {\inferstR{\Delta}{\obj,o}{x=f}{}{\pinull}{}{\Delta[x \mapsto f]}} 
\and
\entry[{\scriptsize (TR-Await)}]
  {\inferpeR{\Delta}{\obj,o}{x}{n}{\i_f}\\\inferpeR{\Delta}{\obj,o}{\i_f}{}{(\fRec{\obj'}{\frr},\cntc)}\\\\
    \Delta[{\it destiny}]=f\\\Delta'=\Delta [\i_f \mapsto (\fRec{\objb}{\frr},\pinull)^\checkmark]}
  {\inferstR{\Delta}{\obj,o}{\nwait\ x?}{n}{\cntc\seqpoint(\obj,\obj')^{\aa}\rfloor\rtcontract{\Delta',f}}{}{\Delta'}}
\and
\entry[{\scriptsize (TR-Await-Runtime)}]
  {\inferpeR{\Delta}{\obj,o}{x}{n}{f}\\\inferpeR{\Delta}{\obj,o}{f}{}{(\fRec{\obj'}{\frr},\cntc)}\\\\
    \Delta[{\it destiny}]=f'\\\Delta'=\Delta [f \mapsto (\fRec{\objb}{\frr},\pinull)^\checkmark]}
  {\inferstR{\Delta}{\obj,o}{\nwait\ x?}{n}{f\seqpoint(\obj,\obj')^{\aa}\rfloor\rtcontract{\Delta',f'}}{}{\Delta'}}
\and
\entry[{\scriptsize (TR-Await-Tick)}]
  {\inferpeR{\Delta}{\obj,o}{x}{n}{\futures}\\\inferpeR{\Delta}{\obj,o}{\futures}{}{(\fRec{\objb}{\frr},\cntc)^\checkmark}}
  {\inferstR{\Delta}{\obj,o}{\nwait\ x?}{n}{\pinull}{}{\Delta }}
\and
\entry[{\scriptsize (TR-If)}]{
	\inferpeR{\Delta}{\obj,o}{e}{n}{\Bool}
	\\
	\inferstR{\Delta}{\obj,o}{s_1}{n_1}{\cntc_1}{\mathcal U_1}{\Delta_1}
	\\
	\inferstR{\Delta}{\obj,o}{s_2}{n_2}{\cntc_2}{\mathcal U_2}{\Delta_2}
	\\\\
        x\in\dom(\Delta) \Longrightarrow \Delta_1(x)=\Delta_2(x)
        \\\\
        x\in{\tt Fut}(\Delta) \Longrightarrow \Delta_1(\Delta_1(x)) = \Delta_2(\Delta_2(x))
        \\\\
                \Delta' = \Delta_1 + (\Delta_2 \setminus (\dom(\Delta) \cup \{ \Delta_2(x) \; | 
        \; x \in {\tt Fut}(\Delta_2)\}))
	}{
	\inferstR{\Delta}{\obj,o}{\eif{e}{s_1}{s_2}}{\_}{\cntc_1+\cntc_2}{
		\mathcal U_1\land \mathcal U_2}{\Delta'}
	}
\and
\simpleentry[{\scriptsize (TR-Skip)}]{\inferstR{\Delta}{\obj,o}{\eskip}{}{\pinull}{\etrue}{\Delta}} \and
\entry[{\scriptsize (TR-Seq)}]{
	\inferstR{\Delta}{\obj,o}{s_1}{n_1}{\cntc_1}{\mathcal U_1}{\Delta_1}
	\\\\
	\inferstR{\Delta_1}{\obj,o}{s_2}{n_2}{\cntc_2}{\mathcal U_2}{\Delta_2}
	}{
  	\inferstR{\Delta}{\obj,o}{s_1;s_2}{\_}{\cntc_1\fatsemi\cntc_2}{
		\mathcal U_1\land \mathcal U_2}{\Delta_2}
	}
\and
\entry[{\scriptsize (TR-Return)}]{
	\inferpeR{\Delta}{\obj,o}{e}{n}{\frr}
	\\\\
	\Delta(\text{destiny})=f \\ \Delta(f)=(\fRec{c}{\frr},\cntc)
	}{
	\inferstR{\Delta}{\obj,o}{\nreturn\ e}{\_}{\pinull}{}{\Delta}
	} \and
\entry[{\scriptsize (TR-Cont)}]{\Delta(f)=\frz}{\inferstR{\Delta}{\obj,o}{\key{cont}(f)}{\_}{\pinull}{}{\Delta}}
}
\caption{Runtime typing rules for statements}
\label{fig:inf:stmtR}
\end{figure*}

\medskip

\begin{theorem}\label{th:wt_init}
Let $P=\bar{I}\ \bar{C}\ \{\many{T\ x}; s\}$ be a \coreABS\  program and let
$\inferpp{\Gamma}{P}{\cct,\,\pair{\cntc}{\cntc'}}{\mathcal{U}}$.
Let also $\sigma$ be a substitution satisfying $\mathcal{U}$ and
\[
\Delta=\sigma(\Gamma+\cct)+start:[\cog: {\rm start}]+f_{\it start}:(\fRec{\rm start}{\unit},\pinull)
\]
Then
\[
\Delta\typep_R ob(start,\varepsilon,\{ l \mid s\},\emptyset) \; \cog({\rm start}, {\it start}) \typed\sigma(\pairl{\cntc}{\cntc'})_{f_{\it start}}^{\rm start}
\]
where $l = [\text{destiny} \mapsto f_{\it start}, \, \many{x \mapsto \undefined}]$.
\end{theorem}

\begin{proof}
By ({\sc TR-Configuration}) and ({\sc TR-Object}) we are reduced to prove:
\begin{eqnarray}
\label{eq.inituno}
\Delta\typep_R^{{\rm start}, {\it start}}\{\text{destiny} \mapsto f_{\it start}, \many{x \mapsto \undefined}|s\}\typed \sigma(\pairl{\cntc}{\cntc'})_{f_{\it start}}^{\rm start}
\end{eqnarray}
To this aim, let $\bar{X}$ be the variables used in the inference rule of \rulename{T-Program}.

To demonstrate (\ref{eq.inituno}) we use \rulename{TR-Process}. Therefore we need to prove:
$$\Delta[\text{destiny} \mapsto f_{\it start}, \bar{x\mapsto\sigma(X)}] 
\typep_R^{{\rm start}, {\it start}} s\typed \sigma(\cntc)\;|\;\Delta'$$
 with $\rtcontract{\Delta'}=\sigma(\cntc'')$.
This proof is done by a standard induction on $s$, using a derivation tree identical 
to the one used for the inference (with the minor exception of replacing the $f$s 
used in the inference with corresponding $\i_f$s). This is omitted because straightforward.
\qed
\end{proof}

\begin{definition}
A runtime contract $\rcntc$ is {\em well-formed} if it is non recursive and if futures and method calls in $\rcntc$ are placed as described by the typing rules: i.e. in a sequence $\cntc_1\cseq\dots\cseq\cntc_n$, they are present in all $\cntc_i$, $i_1\leq i\leq i_k$ with $\cntc_{i_1}$ being when the method is called, and $\cntc_{i_k}$ being when the method is synchronised with.
Formally, for all $\mf{f}{\mpair{c}{\cntc}{\cntc'}}\in\rcntc$, we can derive $\emptyset\typep\cntc\typed \cntc'$ with the following rules:
\myrules{\mathmode\scriptmode\simplemode}{
\simpleentry{\pinull\typep\pinull\typed \pinull} \and
\simpleentry{\pinull\typep\ccall{\C.\m}{\rec}{\vect{\rec}}{\rec'}\typed\pinull} \and
\entry{\cntc'=\pinull\lor\cntc'=f}{\cntc'\typep f\typed f}  \and
\entry{\cntc=\ccall{\C!\m}{\rec}{\vect{\rec}}{\rec'}\\\\\cntc'=\pinull\lor\cntc'=\cntc}{\cntc'\typep\cntc\typed\cntc} \and
\entry{\cntc'=\pinull\lor\cntc'=f}{\cntc'\typep\csync{f}{\cdepn{c}{c'}}\typep\pinull} \and
\entry{\cntc=\ccall{\C!\m}{\rec}{\vect{\rec}}{\rec'}\\\\\cntc'=\pinull\lor\cntc'=\cntc}{\cntc'\typep\csync{\cntc}{\cdepn{c}{c'}}\typed\pinull} \and
\entry{\cntc'\typep\cntc_1\typed \cntc''\\\cntc''\typep\cntc_2\typed \cntc'''}{\cntc'\typep\cntc_1\cseq\cntc_2\typed \cntc'''} \and
\entry{\cntc'\typep\cntc_1\typed \cntc''}{\cntc'\typep\cntc_1\cseq\pinull\typed \cntc''} \and
\entry{\cntc'\typep\cntc_1\typed \cntc''\\\cntc'\typep\cntc_2\typed \cntc''}{\cntc'\typep\cntc_1\cor\cntc_2\typed \cntc''} \and
\entry{\cntc_1'\typep\cntc\typed \cntc_1''\\\cntc_2'\typep\cntc'\typed \cntc_2''}{\cntc_1'\cpar\cntc_2'\typep\cntc\cpar\cntc'\typed \cntc_1''\cpar \cntc_2''}
}
\end{definition}
\begin{lemma}\label{lem:rtkwf}
If $\Delta\typep\nt{cn}\typed\rcntc$ is a valid statement, then $\rcntc$ is well-formed.
\end{lemma}
\begin{proof}
The result is given by the way $\rtcontract{\cdot}$ is used in the typing rules.\qed
\end{proof}

In the following theorem we use the so-called \emph{later-stage relation} $\cls$ 
that has been defined in Figure~\ref{fig:sr:lr} on  runtime contracts.
\begin{figure*}
the substitution process
  \[
  \begin{array}{rcl}
    \subst{\unit}{\unit} & \eqdef & \varepsilon
    \\
    \subst{\frr}{\X} & \eqdef & \subst{\frr}{\X}
    \\
    \subst{[\cog {:} \obj', x_1 {:} \frr_1', \cdots , x_n {:} \frr_n']
    }{[\cog {:} \obj, x_1 {:} \frr_1, \cdots , x_n {:} \frr_n]}
    & \eqdef & \subst{\obj'}{\obj} \, \subst{\frr_1'}{\frr_1} \, \cdots \, 
    \subst{\frr_n'}{\frr_n}
    \\
    \subst{\fRec{\obj'}{\frr'}}{\fRec{\obj}{\frr}} 
    & \eqdef & \subst{\obj'}{\obj} \, \subst{\frr'}{\frr}
  \end{array}
  \]

the later-stage relation is the least congruence with respect to runtime contracts that contains the rules
\myrules{\mathmode\compactmode}{
\entry[LS-Global]{\clsstmtm{\concntc_1}{\concntc_1'}\\\clsstmtm{\concntc_2}{\concntc_2'}}{\clsstmtm{\concntc_1\parallel\concntc_2}{\concntc_1'\parallel\concntc_2'}} \and
\entry[LS-Bind]
  {\Delta(\C.\m) = \mcontract{\rec_\nthis}{(\vect{\rec_\nthis})}{\pair{\cntc}{\cntc'}}{\rec_\nthis'}
  \\    \vect{c} = \fn(\pair{\cntc}{\cntc'})\setminus\fn(\rec_\nthis,\vect{\rec_\nthis}, \rec_\nthis') 
  \\\\
	\rec_\nparam =  [\cog: c, \many{x{:}\frr}] \\
 \vect{c'} \cap \fn(\rec_\nparam,\vect{\rec_\nparam}, \rec_\nparam') = \emptyset}
  {\clsstmtm{\mf{f}{\mccall{\C!\m}{\rec_\nparam}{\vect{\rec_\nparam}}{\rec_\nparam'}}}%
    {\mf{f}{\mpair{c}{\cntc}{\cntc'}}\subst{\vect{c'}}{\vect{c}}\subst{\rec_\nparam,\vect{\rec_\nparam},\rec_\nparam'}{\rec_\nthis,\vect{\rec_\nthis},\rec_\nthis'}}}
\and
\entry[LS-AInvk]{f'\in\fn(\pair{\cntc}{\cntc'}) }
  {\clsstmtm{\mf{f}{\mpair{c}{\cntc}{\cntc'}}\subst{\ccall{\C!\m}{\rec_\nparam}{\vect{\rec_\nparam}}{\rec_\nparam'}}{f'}}%
    {\mf{f}{\mpair{c}{\cntc}{\cntc'}}\parallel\mf{f'}{\mccall{\C!\m}{\rec_\nparam}{\vect{\rec_\nparam}}{\rec_\nparam'}}}}
\and
%
%\entry[LS-SInvk]{f'\in\fn(\pair{\cntc}{\cntc'})\\
%\rec_\nparam =  [\cog: c, \many{x{:}\frr}] 
%\\
%\textcolor{red}{f'\fresh  ?? da togliere??}
%}
%  {\clsstmtm{\mf{f}{\mpair{c}{\cntc}{\cntc'}}\subst{\ccall{\C.\m}{\rec_\nparam}{\vect{\rec_\nparam}}{\rec_\nparam'}}{\csync{f'}{\cdepa{c}{c}}}}%
%    {\mf{f}{\mpair{c}{\cntc}{\cntc'}}\parallel\mf{f'}{\mccall{\C!\m}{\rec_\nparam}{\vect{\rec_\nparam}}{\rec_\nparam'}}}} \and
%
\entry[LS-SInvk]{f'\in\fn(\pair{\cntc}{\cntc'})\\\rec_\nparam =  [\cog: c, \many{x{:}\frr}] }
  {\clsstmtm{\mf{f}{\mpair{c}{
  ( \ccall{\C.\m}{\rec}{\vect{\frs}}{\rec'}
    \parallel \cntc) \cseq \cntc'}{\cntc''}}}%
    {\mf{f}{\mpair{c}{(f'.(c,c)^{\aa} \parallel \cntc) \cseq \cntc'}{\cntc''}}\parallel\mf{f'}{\mccall{\C!\m}{\rec_\nparam}{\vect{\rec_\nparam}}{\rec_\nparam'}}}}
\and
\entry[LS-RSInvk]{f'\in\fn(\pair{\cntc}{\cntc'})\\\rec_\nparam =  [\cog: c', \many{x{:}\frr}]  \\ c' \neq c}
  {\clsstmtm{\mf{f}{\mpair{c}{
  ( \ccall{\C.\m}{\rec}{\vect{\frs}}{\rec'}
    \parallel \cntc) \cseq \cntc'}{\cntc''}}}%
    {\mf{f}{\mpair{c}{(f'.(c,c') \parallel \cntc) \cseq \cntc'}{\cntc''}}\parallel\mf{f'}{\mccall{\C!\m}{\rec_\nparam}{\vect{\rec_\nparam}}{\rec_\nparam'}}}}
\and
%\entry[LS-Inner --]{\clsstmtc{c}{\cntc_1}{\cntc_1'}\\\clsstmtc{c}{\cntc_2}{\cntc_2'}}{\clsstmtm{\mf{f}{\mpair{c}{\cntc_1}{\cntc_2}}}{\mf{f}{\mpair{c}{\cntc_1'}{\cntc_2'}}}}
%\and
%\simpleentry[LS-DepNull \textcolor{red}{BRUTTA}]{\clsstmtm{\mf{f}{\mpair{c}{\cntc}{\cntc'}}\cpar\mf{f'}{\mpair{c'}{\pinull}{\pinull}}}{\mf{f}{\mpair{c}{\cntc\subst{\pinull}{\csync{f'}{\cdepn{c_1}{c_2}}}}{\cntc'\subst{\pinull}{\csync{f'}{\cdepn{c_1}{c_2}}}}}\cpar\mf{f'}{\mpair{c'}{\pinull}{\pinull}}}}
\simpleentry[LS-DepNull]{\clsstmtm{\mf{f}{\mpair{c}{\cntc}{\cntc'}}\cpar\mf{f'}{\mpair{c'}{\pinull}{\pinull}}}{\mf{f}{\mpair{c}{\cntc\subst{\pinull}{f'}}{\cntc'\subst{\pinull}{f'}}}\cpar\mf{f'}{\mpair{c'}{\pinull}{\pinull}}}}
\\\\
%
%\entry[LS-RSInvk -- ]{\rec=\rrec{c'}{\vect{x}:\vect{\rec'}} \\ c\neq c'}
%  {\clsstmtc{c}{\ccall{\C.\m}{\rec}{\vect{\rec}}{\rec'}}{\csync{\ccall{\C!\m}{\rec}{\vect{\rec}}{\rec'}}{\cdeps{c}{c'}}}} \and
%%
%\simpleentry[LS-Refl --]{\clsstmtc{c}{\cntc}{\cntc}}
\and
\simpleentry[LS-Fut]{
	\clsstmtc{}{f}{\pinull}
}
\and
\simpleentry[LS-Empty]{
	\clsstmtc{}{0.(c,c')^{[\aa]}}{\pinull}
}
\and
\simpleentry[LS-Delete]{\clsstmtc{}{\pinull\cseq\cntc}{\cntc}}\and
\simpleentry[LS-Plus]{\clsstmtc{}{\cntc_1 + \cntc_2}{\cntc_i}}\and
%
%\entry[LS-Trans--]{\clsstmtc{c_\nthis}{\cntc_1}{\cntc_2}\\\clsstmtc{c_\nthis}{\cntc_2}{\cntc_3}}
%  {\clsstmtc{c_\nthis}{\cntc_1}{\cntc_3}} \and
%\entry[LS-Seq--]{\clsstmtc{c_\nthis}{\cntc_1}{\cntc_3}\\\clsstmtc{c_\nthis}{\cntc_2}{\cntc_4}}
%  {\clsstmtc{c_\nthis}{\cntc_1\cseq\cntc_2}{\cntc_3\cseq\cntc_4}} \and
%\entry[LS-Par--]{\clsstmtc{c_\nthis}{\cntc_1}{\cntc_3}\\\clsstmtc{c_\nthis}{\cntc_2}{\cntc_4}}
%  {\clsstmtc{c_\nthis}{\cntc_1\parallel\cntc_2}{\cntc_3\parallel\cntc_4}}  \and
%
%
%\\
%\\
%\simpleentry[LS-Ext1]{\clsstmtm{\mf{f}{\mpair{c}{\cntc\parallel f'}{\cntc'\parallel f'}}}{\mf{f}{\mpair{c}{\cntc}{\cntc'}}} }\and
%\simpleentry[LS-Ext2]{\clsstmtc{}{(\cntc\parallel f)\cseq(\cntc'\parallel f)}{(\cntc\cseq\cntc')\parallel f}}\and
%\simpleentry[LS-Ext3]{\clsstmtc{}{(\cntc\parallel f)\cseq\pinull}{(\cntc\cseq\pinull)\parallel f}}\and
%\simpleentry[LS-Ext4]{\clsstmtc{}{(\cntc\parallel f)\cor(\cntc'\parallel f)}{(\cntc\cor\cntc')\parallel f}}
%\and
%\simpleentry[LS-GC]{\clsstmtm{\mf{f}{\mpair{c}{\pinull}{\pinull}}}{\pinull}
%\textcolor{red}{?? manca indice??}}
}
\caption{The later-stage relation}\label{fig:sr:lr}
\end{figure*}

We observe that the later-stage relation uses a substitution process that \emph{also performs a 
pattern matching operation} -- therefore it is partial because the pattern matching may fail. In particular, $\subst{\frs}{\frr}$ (i) extracts the cog names
and terms $\frs'$ in $\frs$ that corresponds to occurrences of cog names and record
variables in $\frr$ and (ii) returns the corresponding substitution.

\medskip

\begin{theorem}[Subject Reduction]\label{th_subjred1}
Let $\Delta\typep_R cn\typed \concntc$ and $cn\rightarrow cn'$.
Then there exist $\Delta'$,  $\concntc'$, and an injective renaming of 
cog names $\imath$ such that 
\begin{itemize}
\item[--] $\Delta' \typep_R cn'\typed \concntc'$ and
\item[--] $\imath(\concntc)\cls\concntc'$.
\end{itemize}
\end{theorem}
\begin{proof}
The proof is a case analysis on the reduction rule used in $cn\rightarrow cn'$ and
we assume that the evaluation of an expression $\feval{e}{\sigma}$
%of a well-typed expression gives a well-typed value
% and that the operator $\feval{\cdot}{\sigma}$ on pure expression finishes.
always terminates. We focus on the most interesting cases. We remark that the injective
renaming $\imath$ is used to identify fresh cog names that are created by the static
analysis with fresh cog names that are created by the operational semantics. In fact,
the renaming is not the identity only in the case of cog creation (second case below).

\begin{itemize}
\item {\em Skip Statement.}\vspace{-.8em}
$$\nredrule{Skip}{
	ob(o,a,\{l \mid \key{skip};s\},q) 
	\to {ob(o,a,\{l \mid s\},q)}}$$
By \rulename{TR-Object}, \rulename{TR-Process}, \rulename{TR-Seq} and \rulename{TR-Skip}, there exists $\Delta''$ and $\cntc$ such that
  $\inferstR{\Delta''}{c,o}{\key{skip};s}{}{\pinull\fatsemi\cntc}{}{\Delta''}$.
It is easy to see that $\inferstR{\Delta''}{c,o}{s}{}{\cntc}{}{\Delta''}$.
Moreover, by \rulename{LS-Delete}, we have $\clsstmtc{\ncog(o)}{\pinull\cseq\cntc}{\cntc}$ which proves that $\concntc\cls\concntc'$.

\item {\em Object creation.}\vspace{-.8em}
  $$\ntyperule{New-Object}{
   o' = \text{fresh}(\name{C}) \quad
   p=\text{init}(\name{C},o')
   \quad
   a'=\text{atts}(\name{C},\feval{\many{e}}{(a+l)}, c)
   }{
   ob(o,a,\{l \mid x=\key{new}\ \name{C}(\many{e});s\},q)\ cog(c, o)
   \\
   \to ob(o,a,\{l \mid x=o';s\},q)\ cog(c, o) \quad
   ob(o',a',\text{idle},\{p\})}$$
By \rulename{TR-Object} and \rulename{TR-Process},
there exists $\Delta''$ that extends $\Delta$ such that
  $\Delta''\typep_R^{c,o}\key{new}\ \name{C}(\many{e})\typed\frr,\,\pinull\mid\Delta''$.
Let  $\Delta'=\Delta[o'\mapsto\frr]$.
The theorem follows by the assumption that $p$ is empty (see Footnote~\ref{footnote.new}).

\item {\em Cog creation.}\vspace{-.8em}
  $$\ntyperule{New-Cog-Object}{
   c' = \text{fresh}(\,) \quad o' = \text{fresh}(\name{C}) \quad  
   p=\text{init}(\name{C},o') 
   \\  a'=\text{atts}(\name{C},\feval{\many{e}}{(a+l)},c')
   }{
   ob(o,a,\{l \mid x=\key{new}\ \key{cog}\ \name{C}(\many{e});s\},q)
   \\
   \to ob(o,a,\{l \mid x=o';s\},q)\quad ob(o',a',p,\emptyset)\quad cog(c', o')
   }$$
By \rulename{TR-Object} and \rulename{TR-Process},
there exists $\Delta''$ that extends $\Delta$ such that
  $\Delta''\typep_R^{c,o}\key{new}\ \name{C}(\many{e})\typed[cog:c'',\bar{x:\frr}],\,\pinull\mid\Delta''$
  for some $c''$ and records $\bar{\frr}$.
Let $\Delta'=\Delta[o'\mapsto[cog:c',\bar{x:\frr}],c'\mapsto \textit{cog}]$
 and $\imath(c'')=c'$, where $\imath$ is an injective renaming on cog names.
The theorem follows by the assumption that $p$ is empty (see Footnote~\ref{footnote.new}).

\item {\em Asynchronous calls.}\vspace{-.8em}
 $$\ntyperule{Async-Call}{
  o'=\feval{e}{(a+l)}
  \quad \many{v}=\feval{\many{e}}{(a+l)}\quad f = \text{fresh}(~)}
  {ob(o,a,\{l \mid x=e \async \m(\many{e});s\},q) \\\to  ob(o,a,\{l \mid x=f;s\},q) \
  \invoc(o',f,\m,\many{v})\ \fut(f,\bot)}$$
By \rulename{TR-Object} and \rulename{TR-Process}, 
there exist $\bar{\frr}$, $\Delta_1'$, $\cntc$ and $\concntc''$ such that (let $f'=l(\text{destiny})$)

\smallskip

\begin{itemize}
\item  $\concntc=\cproc{\ncog(o)}{f'}{\cntc}{\rtcontract{\Delta_1',f'}}\fpar\concntc''$ 
\item  $\Delta\vdash_{R}^{c,o} \bar{v} :\bar{\frx}$ (with $l=[\bar{y \mapsto v}]$)
\item  $\Delta\vdash_{R}^{c,o} q :\concntc''$ 
\item  $\Delta[\bar{y\mapsto\frr}]\vdash_R^{c,o} x=e!m(\many{e});s\typed\cntc\mid\Delta_1'$
\end{itemize}
Let $\Delta_1=\Delta[\bar{y\mapsto\frr}]$: by either \rulename{TR-Var-Record} or \rulename{TR-Field-Record} and {\sc (TR-AInvk)},
there exist $\frr=\fRec{c'}{\frr'}$ (where $c'$ is the cog of the record of $e$), $\i_f$ and $\cntc_{\i_f}$ such that
  $\Delta_1\typep_R^{c,o} e!m(\many{e})\typed \i_f,\pinull\mid \Delta_1[\i_f\mapsto(\frr,\cntc_{\i_f})]$.
%Let $\Delta(f')=(\frs,\cntc')$, $\Delta_k=[f'\mapsto(\frs,\cntc'\cpar f)][f\mapsto(\frr,\pinull)]$.
By construction of the type system (in particular, the rules \rulename{TR-Get$^{*}$} and \rulename{TR-Await$^{*}$}),
 there exists a term $t$ such that $\cntc=t\subst{\cntc_{\i_f}}{\i_f}$ and such that $\inferstR{\Delta_1[f\mapsto(\frr,f')]}{c,o}{x=f;s}{}{t\subst{f}{\i_f}}{}{\Delta_2'}$
 (with $\Delta_2'\definedas \Delta_1'\setminus\{\i_f\}[f\mapsto(\frr,f')^{[\checkmark]}]$ and $[\checkmark]=\checkmark$ iff $\Delta_1'(\i_f)$ is checked).
By construction of the {\it rt\_unsync} function, there exist a term $t'$ such that 
  $\rtcontract{\Delta_1'}=t'\subst{\cntc_{\i_f}}{\i_f}$ and $\rtcontract{\Delta_2'}=t'\subst{f}{\i_f}$.

Finally, if we note $\Delta'\definedas\Delta[f\mapsto(\frr,f')]$, we can type the invocation message with $\cinvoc{f}{\cntc_{\i_f}}$ (as $c'$ is the cog of the record of $\ethis$ in $\cntc'$), we have
\begin{itemize}
\item $\Delta'\typep_R cn'\typed\cproc{c}{f'}{t\subst{f}{\i_f}}{t'\subst{f}{\i_f}}\cpar\cinvoc{f}{\cntc_{\i_f}}\cpar\concntc''$
\item the rule \rulename{LS-AInvk} gives us that $$\concntc\cls\cproc{c}{f'}{t\subst{f}{\i_f}}{t'\subst{f}{\i_f}}\cpar\cinvoc{f}{\cntc_{\i_f}}\cpar\concntc''$$
\end{itemize}

 \item \emph{Method instantiations}.\vspace{-.8em}
\[
\ntyperule{Bind-Mtd}{
	\{ l \mid s \} =\text{bind}(o,f,\m,\many{v},\text{class}(o))
	}{
	ob(o,a,p,q)\  invoc(o,f,\m,\many{v})\to ob(o,a,p,q \cup \{ l \mid s \} )
	}
\]

 By assumption and rules \rulename{TR-Parallel} and \rulename{R-Invoc} we have  $\Delta(o)=(\C,\frr)$, $\Delta(f)=(\fRec{c}{\frr'},\pinull)$, $c=\ncog(\frr)$
 and $\concntc=\cinvoc{f}{C!\m\ \frr(\many{\frr})\rightarrow\frr'}\cpar \concntc'$ with 
   $\Delta\vdash_R \nt{invoc}(o,f,m,\many{v}):\cinvoc{f}{\C!\m~\frr(\many{\frr})\rightarrow\frr'}$ and
   $\Delta\vdash_R ob(o,a,p,q):\concntc'$.
 Let $\many{x}$ be the formal parameters of {\m} in $\C$.
The auxiliary function $\textit{bind}(o,f,m,\many{v},\C)$ returns a process $\{[\text{destiny}\mapsto f, \many{x}\mapsto \many{v}] \; | \; s\}$.
It is possible to demonstrate that  $\Delta\vdash_R^{c,o} \{l[\text{destiny}\mapsto f, \many{x}\mapsto \many{v}]|s\}:\cproc{c}{f}{\cntc_{\m}}{\cntc'_{\m}}$, where $\Delta(\C.\m)
 = \frs(\bar{\frs}) \{\pairl{\cntc_{0}}{\cntc'_{0}}\} \frs'$ and 
 $\cntc_{\m}=
 \cntc_0\subst{\bar{c}}{\bar{c'}}\subst{\frr,\bar{\frr}}{\frs,\bar{\frs}}$ and 
$\bar{c'}\in \frs'\setminus(\frs\cup\bar{\frs})$ with $\bar{c} \fresh$ and $\cntc'_{\m}=\cntc'_0\subst{\bar{c}}{\bar{c'}}\subst{\frr,\bar{\frr}}{\frs,\bar{\frs}}$.

By rules \rulename{TR-Process} and \rulename{TR-Object}, it follows
that $\Delta\vdash_R ob(o,a,p,q\cup\{\textit{bind}(o,f,m,\many{v},\C)\}): \concntc'\fpar\cproc{c}{f}{\cntc_{\m}}{\cntc'_{\m}}$.
Moreover, by applying the rule \rulename{LS-Bind}, we have that $\cinvoc{f}{\C!\m~\frr(\many{\frr})\rightarrow\frr'}\cls \cproc{c}{f}{\cntc_{\m}}{\cntc'_{\m}}$
 which implies with the rule \rulename{LS-Global} that $\concntc\cls\concntc'$.

 \item {\em Getting the value of a future}.
$$\ntyperule{Read-Fut}{
	f = \feval{e}{(a+l)}  \quad v \neq \bot
	}{ob(o,a,\{l \mid x=e.\key{get};s\},q) \ \textit{fut}(f,v)
\to\\ ob(o,a,\{l \mid x=v;s\},q)\ \textit{fut}(f,v)}
$$

By assumption and rules \rulename{TR-Parallel}, \rulename{TR-Object} and \rulename{TR-Future-Tick},
 there exists $\Delta''$, $\cntc$, $\concntc''$ such that (let $f'=$l[\text{destiny}])
\begin{itemize}
\item[--] $\Delta\vdash_R^{\ncog(o),o}\{l|x=e.\key{get};s\}\typed \cproc{\ncog(o)}{f'}{\cntc}{\rtcontract{\Delta'',f'}}$
\item[--] $\Delta\vdash_R^{\ncog(o),o}q\typed\concntc''$
\item[--] $\Delta\vdash_R \textit{fut}(f,v):\pinull$
\item[--] $\concntc=\cproc{\ncog(o)}{f'}{\cntc}{\rtcontract{\Delta'',f'}}\cpar \concntc''$, and 
\item[--] $\feval{e}{a\circ l}=f$. 
\end{itemize}
Moreover, as $\textit{fut}(f,v)$ is typable and contains a value, we know that $\cntc=\pinull\fatsemi\cntc'$ ($e.\key{get}$ has contract $\pinull$).
With the rule \rulename{TR-Pure}, have that
$\Delta\vdash_R^{\ncog(o),o}\{l|x=v;s\}\typed \cproc{\ncog(o)}{f'}{\cntc}{\rtcontract{\Delta'',f'}}$,
and with $\concntc'=\concntc$, we have the result.

\item \emph{Remote synchronous call}.  Similar to the cases of asynchronous
   call with a \Abs{get}-synchronisation. The result follows, in particular, from  rule \rulename{LS-RSInvk} of Figure~\ref{fig:sr:lr}.
\item \emph{Cog-local synchronous call}. 
  Similar to case of asynchronous call. The result follows, in particular, from rules \rulename{LS-SimpleNull} of Figure~\ref{fig:sr:lr} and from the Definition of $\C.\m\ \frr(\bar{\frr})\rightarrow\frs$.
\item \emph{Local Assignment}.
$$\ntyperule{Assign-Local}{
	x \in \dom(l)  \quad  v=\feval{e}{(a+ l)}
	}{
	ob(o,a,\{l \mid x=e;s\},q) \\
	\to {ob(o,a,\{l[x\mapsto v] \mid s\},q)}}
$$
By assumption and rules \rulename{TR-Object}, \rulename{TR-Process}, \rulename{TR-Seq}, \rulename{TR-Var-Record} and \rulename{TR-Pure},
 there exists $\Delta''$, $\cntc$, $\concntc''$ such that
 (we note $\Delta_1$ for $\Delta[\vect{y:\frx}]$ and $f$ for $l[\text{destiny}]$)
\begin{itemize}
\item[--] $\Delta\vdash_R^{\ncog(o),o}\{l|x=e;s\}\typed \cproc{\ncog(o)}{f}{\pinull\fatsemi\cntc}{\rtcontract{\Delta'',f}}$
\item[--] $\Delta\vdash_R^{\ncog(o),o}q\typed\concntc''$
\item[--] $\concntc=\cproc{\ncog(o)}{f}{\cntc}{\rtcontract{\Delta'',f}}\cpar \concntc''$, and 
\item[--] $\feval{e}{a\circ l}=v$. 
\end{itemize}
We have
$$\Delta\vdash_R^{\ncog(o),o}\{l[x\mapsto\feval{e}{(a+ l)}]|s\}\typed \cproc{\ncog(o)}{f}{\cntc}{\rtcontract{\Delta'',f}}$$
 which gives us the result with $\rcntc'=\mf{f}{\mpair{c}{\cntc}{\rtcontract{\Delta'',f}}}\cpar\rcntc'$.
%We have three cases.
%Either $l(x)$ is not a future, in which case we have 
%$$\Delta\vdash_R^{\ncog(o),o}\{l[x\mapsto\feval{e}{(a+ l)}]|s\}\typed \cproc{\ncog(o)}{f}{\cntc}{\rtcontract{\Delta''}}$$
% which gives us the result with $\rcntc'=\mf{f}{\mpair{c}{\cntc}{\rtcontract{\Delta''}}}\cpar\rcntc'$.
%Either $l(x)$ is a future $f'$ and $f'\in \Im(l[x\mapsto\feval{e}{(a+l)}])\cup\ff(s)$: $\nt{cn'}$ is then typed as in the previous case.
%Finally, let suppose that $l(x)$ is a future $f'$ and $f'\not\in \Im(l[x\mapsto\feval{e}{(a+l)}])\cup\ff(s)$.
%Note $\Delta_1'$ for $\Delta_{|(\vect{\frx}\setminus\{f'\})\cup\ff(s)}[\vect{y:\frx}]$:
% we have that $\rtcontract{\Delta_1'}=\rtcontract{\Delta_1}\setminus\{f'\}$: $s$ is now typed with a variant $\cntc'$ of $\cntc$,
% where all references to $f'$ have been deleted.
%By construction, $\cntc$ cannot contain a synchronisation with $f'$, and because it is well-formed, we can apply the rules \rulename{LS-Ext$^{*}$} to get $\cntc'$, which gives us the result.

\iffalse
\item \emph{Return Statement}.
$$\ntyperule{Return}{
	v = \feval{e}{(a+l)}  \quad f = l(\text{destiny}) }
	{ob(o,a,\{l \mid \key{return}\ e\},q) \ \textit{fut}(f,\bot)\\
	\to ob(o,a,\nidle,q) \ \textit{fut}(f,v)
	}
$$
\fi
\end{itemize}
\qed
\end{proof}

%%% Local Variables:
%%% mode: latex
%%% TeX-master: "SoSyM"
%%% End:

\section{Properties of Section~\ref{sec.contractanalysis}}\label{sec:analysis-proof}
%!TEX root = SoSyM.tex

% 1. Semantics of a contract (in normal form, i.e. the contract must be well-formed).
In this section, we will prove that the statements given in Section~\ref{sec.contractanalysis} are correct, i.e. that the fixpoint analysis does detect deadlocks.
To prove that statement, we first need to define the dependencies generated by the runtime contract of a running program.
Then, our proof works in three steps:
 i) first, we show that our analysis (performed at static time) contains all the dependencies of the runtime contract of the program;
 ii) second that the dependencies in a program at runtime are contained in the dependencies of its runtime contract;
 and finally iii) when $\nt{cn}$ (typed with $\rcntc$) reduces to $\nt{cn}'$ (typed with $\rcntc'$), we prove that the dependencies of $\rcntc'$ are contained in $\rcntc$.
Basically, we prove that the following diagram holds:\vspace{.3em}
\newcommand{\nodedistance}{1.8em}
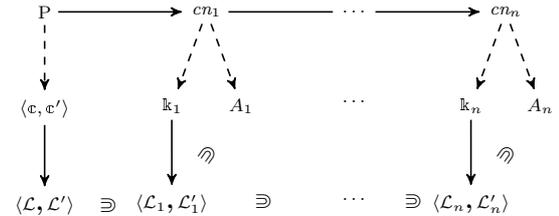
\hspace*{2em}\begin{tikzpicture}[->,>=stealth',shorten >=1pt,auto,scale=0.8, every node/.style={scale=0.8},node distance=.9cm, semithick]
\node (P) {P};
\node (CN1) [right=of P,xshift=3em]   {$\nt{cn}_1$};
\node (CNd) [right=of CN1,xshift=2em] {$\dots$};
\node (CNN) [right=of CNd,xshift=2em] {$\nt{cn}_n$};

\node (CC)  [below=of P] {$\pair{\cntc}{\cntc'}$};
\node (K1)  [below=of CN1, xshift=-\nodedistance] {$\rcntc_1$};
\node (A1)  [below=of CN1, xshift=\nodedistance] {$A_1$};
\node (KAd) [below=of CNd,yshift=-.3em] {$\dots$};
\node (KN)  [below=of CNN, xshift=-\nodedistance] {$\rcntc_n$};
\node (AN)  [below=of CNN, xshift=\nodedistance] {$A_n$};

\node (LL0) [below=of CC,yshift=.2em]  {$\pairl{\LM}{\LM'}$};
\node (LL1) [below=of K1]  {$\pairl{\LM_1}{\LM_1'}$};
\node (LLd) [below=of KAd,yshift=-.7em] {$\dots$};
\node (LLN) [below=of KN]  {$\pairl{\LM_n}{\LM_n'}$};

\draw[->] (P) -- (CN1);
\draw[-] (CN1) -- (CNd);
\draw[->] (CNd) -- (CNN);

\draw[->,dashed] (P) -- (CC);
\draw[->] (CC) -- (LL0);

\draw[->,dashed] (CN1) -- (K1);
\draw[->,dashed] (CN1) -- (A1);
\draw[->] (K1) -- (LL1);

\draw[->,dashed] (CNN) -- (KN);
\draw[->,dashed] (CNN) -- (AN);
\draw[->] (KN) -- (LLN);

\node[scale=1.3] (INC1) at ($(LL0) !.5! (LL1)$) {$\Supset$};
\node[scale=1.3] (INCd) at ($(LL1) !.5! (LLd)$) {$\Supset$};
\node[scale=1.3] (INCN) at ($(LLd) !.5! (LLN)$) {$\Supset$};

\path[-,white] (LL1) edge node [sloped,anchor=center,black,scale=1.3] {$\Supset$} (A1)
               (LLN) edge node [sloped,anchor=center,black,scale=1.3] {$\Supset$} (AN);

\end{tikzpicture}

Hence, the analysis $\pair{\LM}{\LM'}$ contains all the dependencies $A_i$ that the program can have at runtime, and thus, if the program has a deadlock, the analysis would have a circularity.

\medskip

In the following, we introduce how we compute the dependencies of a runtime contract.
This computation is difficult in general, but in case the runtime contract is as we constructed it in the subject-reduction theorem, then the definition is very simple.
First, let say that a contract $\cntc$ that does not contain any future is {\em closed}.
It is clear that we can compute $\cntc(\act_{[n]})$ when $\cntc$ is closed.
\begin{proposition}
Let $\Delta\typep \nt{cn}\typed \rcntc$ be a typing derivation constructed as in the proof of Theorem~\ref{th_subjred1}.
Then $\rcntc$ is well formed and $\fsimpl{\rcntc}=\mf{f_{\it start}}{\mpair{{\rm start}}{\cntc}{\cntc'}}$ where $\cntc$ and $\cntc'$ are closed.
\end{proposition}
\begin{proof}
The first property is already stated in Lemma~\ref{lem:rtkwf}.
The second property comes from the fact that when we create a new future $f$ (in the {\em Asynchronous calls} case for instance), we map it in $\Delta'$ to its father process,
  which will then reference $f$ because of the $\rtcontract{\cdot}$ function.
Hence, if we consider the relation of which future references which other future in $\rcntc$, we get a dependency graph in the shape of a directed tree, where the root is $f_{\it start}$.
So, $\fsimpl{\rcntc}$ reduces to a simple pair of contract of the form $\mf{f_{\it start}}{\mpair{{\rm start}}{\cntc}{\cntc'}}$ where $\cntc$ and $\cntc'$ are closed.\qed
\end{proof}

In the following, we will suppose that all runtime contracts $\rcntc$ come from a type derivation constructed as in Theorem~\ref{th_subjred1}.

\begin{definition}
The {\em semantics} of a closed runtime pair (unique up to remaning of cog names) for the saturation at $i$, noted $\csem{n}{\mf{f}{\mpair{c}{\cntc}{\cntc'}}}$, is defined as
  $\csem{n}{\mf{f}{\mpair{c}{\cntc}{\cntc'}}}=(\cntc(\act_{[n]})_{c})\cseq (\cntc'(\act_{[n]})_{c})$.
We extend that definition for any runtime contract with $\csem{n}{\rcntc}\definedas\csem{n}{\fsimpl{\rcntc}}$.
\end{definition}

% 2. Prove that the semantics at runtime is included in the semantics at static time
Now that we can compute the dependencies of a runtime contract, we can prove our first property:
 the analysis performed at static time contains all the dependencies of the initial runtime contract of the program
 (note that $\sigma(\cntc)(\actn{n})\cseq \sigma(\cntc')(\actn{n})$ is the analysis performed at static time,
  and $\csem{n}{\sigma(\mf{f_{\it start}}{\mpair{{\rm start}}{\cntc}{\cntc'}})}$ is the set of dependencies of the initial runtime contract of the program):
\begin{proposition}\label{prop:ap1}
Let $P=\bar{I}\ \bar{C}\ \{\many{T\ x}; s\}$ be a \coreABS\  program and let $\inferpp{\Gamma}{P}{\cct,\,\pair{\cntc}{\cntc'}}{\mathcal{U}}$.
Let also $\sigma$ be a substitution satisfying $\mathcal{U}$.
Then we have that $\csem{n}{\sigma(\mf{f_{\it start}}{\mpair{{\rm start}}{\cntc}{\cntc'}})}\Subset\sigma(\cntc)(\actn{n})\cseq \sigma(\cntc')(\actn{n})$.
\end{proposition}
\begin{proof}
The result is direct with an induction on $\cntc$ and $\cntc'$, and with the fact that $\cor$, $\fatsemi$ and $\cpar$ are monotone with respect to $\Subset$.
\qed
\end{proof}

% 3. Prove that the semantics at runtime includes the dependencies of the program
We now prove the second property: all the dependencies of a program at a given time is included in the dependencies generated from its contract.
\begin{proposition}\label{prop:ap2}
Let suppose $\Delta\typep_R \nt{cn}\typed \rcntc$ and let $A$ be the set of dependencies of $\nt{cn}$.
Then, with $\csem{n}{\rcntc}=\pair{\LM}{\LM'}$, we have $A\subset\LM$.
\end{proposition}
\begin{proof}
By Definition~\ref{def.obj-circularity}, if $\nt{cn}$  has a dependency $(c,c')$,
 then there exist $\nt{cn}_1=\nt{ob}(o,a,\{l | x=e.\get;s\},q)\in\nt{cn}$, $\nt{cn}_2=\nt{fut}(f,\bot)\in\nt{cn}$
 and $\nt{cn}_3=\nt{ob}(o',a',p',q') \in \nt{cn}$ such that $\feval{e}{(a+l)} = l'(\text{destiny}) = f$, $\{l'\mid s'\}\in p'\cup q'$ 
 and $a(\nt{cog})=c$ and $a'(\nt{cog})=c'$.
By runtime typing rules \rulename{TR-Object}, \rulename{TR-Process}, \rulename{TR-Seq} and \rulename{TR-Get-Runtime},
 the contract of $\nt{cn}_1$ is $$\mf{l(\nt{destiny})}{\mpair{a(\ncog)}{\csync{f}{\cdeps{c}{c'}}\fatsemi\cntc_{s}}{\cntc_s'}}\cpar\rcntc_q$$
 we indeed know that the dependency in the contract is toward $c'$ because of \rulename{TR-Invoc} or \rulename{TR-Process}.
Hence $\rcntc=\mf{l(\nt{destiny})}{\mpair{a(\ncog)}{\csync{f}{\cdeps{c}{c'}}\fatsemi\cntc_{s}}{\cntc_s'}}\cpar\rcntc'$.
It follows, with the lam transformation rule \rulename{L-GAinvk}, that $(c,c')$ is in $\LM$.
\qed
\end{proof}

% 4. Prove that the semantics of \rcntc includes the semantics of \rcntc' when \rcntc\cls\rcntc'

\begin{proposition}\label{prop:ap3}
Given two runtime contracts $\rcntc$ and $\rcntc'$ with $\rcntc\cls\rcntc'$, we have that $\csem{n}{\rcntc'}\Subset\csem{n}{\rcntc}$.
\end{proposition}
\begin{proof}
We refer to the rules \rulename{LS-*} of the later-stage relation defined in
 Figure~\ref{fig:sr:lr} and to the lam transformation rules  \rulename{L-*} defined in Figure~\ref{fig:lamTrans} . 
The result is clear for the rules \rulename{LS-Global}, \rulename{LS-Fut}, \rulename{LS-Empty}, \rulename{LS-Delete} and \rulename{LS-Plus}.
The result for the rule \rulename{LS-Bind} is a consequence of \rulename{L-AInvk}.
The result for the rule \rulename{LS-AInvk} is a consequence of the definition of $\rewritet$.
The result for the rule \rulename{LS-SInvk} is a consequence of the definition of $\rewritet$ and \rulename{L-SInvk}.
The result for the rule \rulename{LS-RSInvk} is a consequence of the definition of $\rewritet$ and \rulename{L-RSInvk}.
Finally, the result for the rule \rulename{LS-DepNull} is a consequence of the definition of $\rewritet$.
\qed
\end{proof}

% 5. Concludes that \Gamma\typep P\typed \rcntc: if there is a deadlock in P, then there is a loop in the semantics of \rcntc
We can finally conclude by putting all these results together:
\begin{theorem}~\label{thm:analysis-correct}
If a program $P$ has a deadlock at runtime, then its {\em abstract semantics saturated at n} contains a circle.
\end{theorem}
\begin{proof}
This property is a direct consequence of Propositions~\ref{prop:ap1},~\ref{prop:ap2} and~\ref{prop:ap3}.\qed
\end{proof}

\section{Properties of Section~\ref{sec.mutanalysis}}\label{sec:mutanalysis-proof}
%!TEX root = SoSyM.tex

The next theorem states the correctness of our model-checking
technique.

Below we write  $\abstractsemantics{\nt{cn}}=(\sem{\pairl{\cntcp_{n}}{\cntcp_{n}'}_{\rm start}})^\flat$, if $\Delta\vdash_R \nt{cn}:\pairl{\cntcp_1}{\cntcp_1'}$ and $n$ is the order of 
$\pairl{\cntcp_1}{\cntcp_1'}_{\rm start}$.

\begin{theorem}\label{thm:mutanalysis-correct}
 Let  $(\ct , \{ \vect{T \ x \sseq} s \}, \cct)$ be a \coreABS{} program
    and $\nt{cn}$ be a configuration of its operational semantics.
    \begin{enumerate}
    \item If $\nt{cn}$ has a circularity, then a circularity occurs in
      $\abstractsemantics{\nt{cn}}$;
    \item if $\nt{cn}\to \nt{cn}'$ and $\abstractsemantics{\nt{cn}'}$
      has a circularity, then a circularity is already present in
      $\abstractsemantics{\nt{cn}}$;
    \item let $\imath$ be an injective renaming of cog names;
      $\abstractsemantics{\nt{cn}}$ has a circularity if and only if
      $\abstractsemantics{\imath(\nt{cn})}$ has a circularity.
    \end{enumerate}
\end{theorem}

\begin{proof}
   To demonstrate item 1, let 
   \[
   \abstractsemantics{\nt{cn}} =(\sem{\pairl{\cntcp_{n}}{\cntcp_{n}'}_{\rm start}})^\flat\; .
   \]
   We prove that every dependencies occurring in $\nt{cn}$ is also contained in one state of $(\sem{\pairl{\cntcp_{n}}{\cntcp_{n}'}_{\rm start}})^\flat$.
By Definition~\ref{def.obj-circularity}, if  $\nt{cn}$  has a dependency $(\obj,\obj')$ 
then it contains  $\nt{cn}'' = \nt{ob}(o,a,\{l | x=e.\get;s\},q) \quad \nt{fut}(f,\bot)$,
where $f=\feval{e}{(a+l)}$, $a(\nt{cog})=c$ and there is $\nt{ob}(o',a',\{l' | s'\},q')\in \nt{cn}$ such that
$a'(\nt{cog})=c'$ and $l'(\text{destiny})=f$.
By the typing rules, the contract of  
$\nt{cn}'$ is $f\seqpoint(\obj,\obj')\fatsemi\cntc_{s}$, where, by typing rule \rulename{T-Configurations},
 $f$ is actually replaced by a $\C!\m\,\frr(\bar{\frs})\rightarrow\frs$ produced by a 
concurrent $\nt{invoc}$ configuration, or by the contract pair $\pairl{\cntc_\m}{\cntc'_\m}$ corresponding to the method body.

As a consequence 
$\abstractsemantics{\nt{cn}''} = (\sem{\CP[\pairl{\cntc''\addition (c,c')}{\cntc'''}_{c}]_{c''}})^\flat$.

Let $\abstractsemantics{\nt{ob}(o',a',\{l' | s'\},q') }=(\sem{\CP'[\pairl{\cntc_\m}{\cntc'_\m}_{c}]_{c''}})^\flat$, with $\feval{e}{(a+l)}=l'(\text{destiny})$, then 
\[
\begin{array}{l}
\abstractsemantics{
\nt{ob}(o',a',\{l' | s'\},q')\quad \nt{cn}''} = 
\\
\qquad \qquad \qquad (\sem{\CP[\pairl{\cntc''\addition (c,c')}{\cntc'''}_{c}]_{c''}\rfloor\CP'[\pairl{\cntc_\m}{\cntc'_\m}_{c'}]_{c'''}})^\flat \; .
\end{array}
\]
In general, if $k$ dependencies occur in a state $\nt{cn}$, then
there is $\nt{cn}''\subseteq\nt{cn}$ that collects all the tasks manifesting
the dependencies. % such that
\[
\begin{array}{l}
\abstractsemantics{\nt{cn}''} \; =
\\
\quad 
(\sem{\CP_1[\pairl{\cntc''_1\addition (c_1,c'_1)}{\cntc'''_1}_{c_1}]_{c''_1}\rfloor\CP'_1[\pairl{\cntc_{\m_1}}{\cntc'_{\m_1}}_{c'_1}]_{c'''_1}})^\flat
 \\
\qquad \rfloor  \,\cdots
\rfloor\,(\sem{\CP_k[\pairl{\cntc''_k\addition (c_k,c'_k)}{\cntc'''_k}_{c_k}]_{c''_k}\rfloor\CP'_k[\pairl{\cntc_{\m_k}}{\cntc'_{\m_k}}_{c'_k}]_{c'''_k}})^\flat
\end{array}
\]

By definition of $\rfloor$ composition in Section~\ref{sec.contractanalysis}, the initial state contains all the above pairs $(\obj_i,\objb_i)$.
\bigskip

 Let us prove the item 2. We show that the transition $\nt{cn} \lred{} \nt{cn}'$ does not 
 produce new dependencies. That is, the set of dependencies in the states of $\abstractsemantics{\nt{cn}'}$ is equal or smaller than the set of dependencies in the states of $\abstractsemantics{\nt{cn}}$.

 By Theorem~\ref{th_subjred1}, if $\Delta\vdash_R\nt{cn}:\concntc$
 then   $\Delta'\vdash_R\nt{cn}':\concntc'$, with
 $\concntc\cls\concntc'$. 
%The argument is by induction on the proof tree of $\concntc\cls\concntc'$.
 We refer to the rules \rulename{LS-*} of the later-stage relation defined in
 Figure~\ref{fig:sr:lr} and to the contract reduction rules
 \rulename{Red-*} defined in Figure~\ref{fig:contred} . 
The result is clear for the rules \rulename{LS-Global}, \rulename{LS-Fut}, \rulename{LS-Empty}, \rulename{LS-Delete} and \rulename{LS-Plus}.
The result for the rule \rulename{LS-Bind} is a consequence of \rulename{Red-AInvk}.
The result for the rule \rulename{LS-AInvk} is a consequence of the definition of $\rewritet$.
The result for the rule \rulename{LS-SInvk} is a consequence of the definition of $\rewritet$ and \rulename{Red-SInvk}.
The result for the rule \rulename{LS-RSInvk} is a consequence of the definition of $\rewritet$ and \rulename{Red-RSInvk}.
Finally, the result for the rule \rulename{LS-DepNull} is a consequence of the definition of $\rewritet$.

\bigskip

 Item 3 is obvious because circularities are preserved by injective renamings of cog 
names.
\qed
\end{proof}

%%% Local Variables:
%%% mode: latex
%%% TeX-master: "SoSyM"
%%% End:

\end{document}